\documentclass[12pt,letterpaper]{article}
\usepackage{times}
\usepackage[utf8]{inputenc}
\usepackage[authoryear,round]{natbib}
\usepackage{amsmath,amsfonts,amssymb,amsthm,enumitem}
\usepackage{bibunits}
\usepackage{rotating}
\usepackage{caption}
\usepackage[normalem]{ulem}
\usepackage{subcaption}
\usepackage{graphicx}
\usepackage{pgf}
\usepackage{booktabs}
\usepackage{multirow}
\usepackage[margin=1in]{geometry}
\usepackage{hyperref}
\usepackage{bbm}
\usepackage{comment}
\usepackage{float}
\usepackage{mathrsfs}
\usepackage{longtable}
\usepackage{pdflscape}
\usepackage{threeparttable}
\usepackage{setspace}
\usepackage{titlesec}
\titleformat{\section}{\Large\bfseries}{\thesection}{1em}{}
\hypersetup{
    colorlinks=true,
    linkcolor=black,
    citecolor=black,
    urlcolor=black,
    pdfborder={0 0 0}
}
\usepackage{tikz,calc}
\usetikzlibrary{chains}
\usetikzlibrary{positioning}

\renewcommand\labelitemi{-}
\usepackage{accents}

\newtheoremstyle{boldname}
  {\topsep}   
  {\topsep}   
  {\itshape}  
  {}          
  {\bfseries} 
  {.}         
  { }         
  {\thmname{#1}~\thmnumber{#2}\thmnote{~(#3)}} 

\theoremstyle{boldname}
\newtheorem{theorem}{Theorem}
\newtheorem{remark}{Remark}
\newtheorem{proposition}{Proposition}
\newtheorem{lemma}{Lemma}

\makeatletter
\renewcommand{\maketitle}{%
  \begin{center}
    {\LARGE \bfseries \@title \par}
    \vskip 0.8em
    {\normalsize
      \lineskip .4em
      \begin{tabular}[t]{c}
        \@author
      \end{tabular}\par}
    \vskip 0.8em
    {\small \@date \par}
  \end{center}
  \vskip 0.8em
}
\makeatother
\begin{document}

\title{The Anatomy of Commodity Risk: Micro, Market, and Economy-Wide Sources}

\author{
Nektarios Aslanidis$^{a}$,
Aurelio F. Bariviera$^{a}$,
George Kapetanios$^{b}$,
Vasilis Sarafidis$^{c}$,
Alexia Ventouri$^{b}$\\[0.5em]
\small
$^{a}$Universitat Rovira i Virgili, ECO-SOS, Department of Economics, Spain\\
\small
$^{b}$King's College London, Department of Banking and Finance, UK\\
\small
$^{c}$Brunel University London, Brunel Business School, UK\\[0.3em]
}

\date{}

\maketitle

\vspace{-0.2em}

\begin{abstract}
We study the anatomy of commodity risk by distinguishing micro, market-level, and economy-wide sources. We develop a two-stage ``divide-and-conquer'' framework that allows sensitivities to these risk sources to vary across commodities while treating economy-wide risk as latent. The first stage uses defactored instrumental-variable estimation to recover commodity-specific sensitivities to micro and market conditions. The second combines principal components with high-dimensional variable selection to identify an observable representation of macro-financial risk. We then construct Risk Intensity Indices (RIIs), which combine estimated sensitivities with prevailing risk conditions to quantify the relative importance of each risk source on a common scale. Market risk is the largest component on average, accounting for about two fifths of total risk intensity and more than half for energy commodities. Risk intensity is also highly concentrated across individual commodities: the top 20\% account for approximately half of micro and market risk intensity, whereas macro risk is more broadly dispersed. The composition of risk varies substantially across sectors and over time, with market risk becoming particularly prominent during episodes of commodity-market stress. Micro and market RIIs also contain information about future volatility and absolute returns. These findings provide investors, risk managers, and policymakers with a diagnostic of where commodity risk is concentrated, which risk layers are most important, and how their importance changes over time. More broadly, our divide-and-conquer framework provides a flexible approach to decomposing layered risk in settings where common risk is latent.
\end{abstract}



{\small
\noindent\textit{Acknowledgements:} The authors would like to thank T. Diasakos, T. Panagiotidis, D. Robertson, C. Savva, A. Urquhart, as well as participants at the 31st IPDC (Exeter, 6--7 July 2026), 29th IPDC (Orléans, 4--5 July 2024), and the Macro-Finance Cluster Seminar of the University of York (22 October 2024) for useful comments and suggestions. Aslanidis acknowledges financial support by the Spanish Government, Ministry of Science and Innovation under Project reference PID2022-137382NB-I00.
}

\doublespacing
\section{Introduction \label{sec:intro}}

Commodity markets occupy a distinctive position in the global economy, linking financial markets to the production, consumption, storage, and trade of essential inputs ranging from energy and metals to agricultural goods. Commodity prices are influenced by physical supply and demand conditions, inventories, and other commodity-specific fundamentals
\citep{GortonRouwenhorst2006,GortonHayashiRouwenhorst2013}. At the same time, the increasing participation of financial investors has strengthened the integration of commodity futures with broader financial markets and increased comovement across commodities \citep{TangXiong2012,BasakPavlova2016,Baker2021}. In particular, index investment can generate common movements across commodities that are not fully explained by commodity-specific fundamentals, while recent evidence suggests that index trading can also facilitate the transmission of nonfundamental shocks across commodities included in an index
\citep{DaTangTaoYang2024}. Commodity returns therefore reflect both highly commodity-specific conditions and forces operating more broadly across commodity and financial markets.

Taken together, these considerations suggest that commodity risk has an inherently layered structure. At the micro level, each commodity is exposed to idiosyncratic risks arising from its own price volatility, liquidity, trading activity, inventories, and supply-demand fundamentals. At the market level, commodities are exposed to common movements generated by developments elsewhere in commodity and financial markets, including market returns, volatility, and speculative activity. The importance of such common sources is supported by evidence that common commodity factors help explain the cross-sectional variation in commodity returns \citep{BakshiGaoRossi2019}, as well as by the broader literature on commodity-market financialisation \citep{TangXiong2012,BasakPavlova2016,Baker2021,DaTangTaoYang2024}. In particular, financialisation can increase comovement across commodities and strengthen the transmission of financial shocks, while index trading can propagate nonfundamental shocks across commodities within the same index. Beyond these two layers, commodities are exposed to economy-wide conditions, including global economic activity, monetary and financial conditions, inflation, exchange rates, and economic uncertainty, which can affect many commodity prices simultaneously. Commodity risk can therefore be viewed as arising from three conceptually distinct layers: \emph{micro}, \emph{market}, and \emph{economy-wide} sources.

Identifying sensitivity to different risk sources is informative, but sensitivities alone do not reveal their relative importance to a commodity's overall risk. A large sensitivity to a specific risk source may yield minimal risk intensity if the underlying variable rarely moves. Conversely, a moderate sensitivity to a highly volatile variable can be far more consequential. Furthermore, a collection of disparate sensitivity coefficients lacks a common metric to compare risk contributions across commodities or track how the composition of risk evolves over time. We address this limitation by constructing \emph{Risk Intensity Indices} (RIIs) that combine heterogeneous sensitivities with prevailing risk conditions and measure the magnitude of their contribution on a common scale. Separate indices are constructed for idiosyncratic, market, and economy-wide sources. The RIIs therefore shift the focus from asking simply \emph{how sensitive is a commodity to a particular source of risk?} to asking \emph{where does commodity risk come from, how important are its different sources, and how does its anatomy differ across commodities and over time?} In this sense, the indices provide a parsimonious mapping from a potentially large and heterogeneous set of risk sources into an economically interpretable decomposition of commodity risk intensity.

Our approach differs from existing commodity factor and risk-premium decompositions in both its objective and its unit of analysis. \citet{SzymanowskaDeRoonNijmanVanDenGoorbergh2014} decompose commodity futures risk premia into spot and term premia and study the factors explaining their cross-sectional variation, while \citet{BakshiGaoRossi2019} show that average commodity returns, carry, and momentum describe the cross section of commodity returns and investigate the economic sources underlying these factors. Related work examines how financialisation changes commodity prices, risk premia, volatility, and return dependence \citep{Baker2021,DaTangTaoYang2024}. Our objective is different. Rather than explaining the cross section of expected returns or decomposing risk premia, we identify heterogeneous sensitivities of realised commodity returns to micro, market-level, and economy-wide conditions and use these sensitivities, together with prevailing risk conditions, to measure the intensity of each risk source. The model does not impose an asset-pricing restriction linking these sensitivities to expected returns and does not
estimate the price of any source of risk. Whether the exposures identified here are priced in the cross section of expected commodity returns is a distinct question from the one studied in this paper. Our analysis instead uses the resulting risk intensities to characterise the anatomy of commodity risk across individual commodities, sectors, and time.

Building on this distinction, our contribution is fourfold. First, we develop a two-stage ``divide-and-conquer''estimation framework that enables the consistent identification of heterogeneous sensitivities across all three risk layers in the presence of latent macro risk. Second, we map these sensitivities into Risk Intensity Indices (RIIs), which provide comparable measures of risk intensity at the micro, market, and economy-wide levels. Third, we employ the resulting indices to characterize the anatomy of commodity risk. Specifically, we quantify the relative importance of the three sources of risk for the commodity market as a whole, examine differences in their composition across agricultural, energy, and metal commodities as well as across individual commodities, and trace their evolution over time by relating changes in the structure of risk to major economic and financial episodes. Finally, we establish the forward-looking predictive relevance of the RIIs and find that micro- and market-level risk intensity predict future volatility and absolute returns over horizons of up to twelve periods.

A central methodological challenge arises from the inherently latent nature of macro risk. The set of macrofinancial indicators that may proxy for this unobservable risk factor is typically high-dimensional, heterogeneous, and highly collinear, with limited guidance on their empirical relevance. Consequently, the selection of informative proxies entails a fundamental trade-off between omitted-variable bias, which may arise from excluding relevant information, and overfitting, which may result from incorporating an excessive number of noisy and highly correlated indicators.

To address this challenge, we develop a sequential ``divide-and-conquer'' estimation framework that partitions risk into observed micro and market components and latent macro risk. In Stage 1, heterogeneous sensitivities to the observed micro- and market-level risk components are estimated using commodity-specific instrumental-variable (IV) regressions. IV identification exploits a latent common-factor structure in the commodity-specific risk characteristics that is sufficiently rich to span the latent macro component responsible for endogeneity. Despite the endogeneity induced by the latent component, the resulting IV estimator delivers consistent estimates of heterogeneous sensitivities to the observed risk layers. In Stage 2, PCA is applied to the residuals across commodities to recover the dominant directions of common macro risk. These principal components are mapped to a large set of candidate macro-financial proxies, from which relevant predictors are selected using the Boosting with Multiple Testing (BMT) procedure of \citet{KapetaniosEtal2026}. This yields a sparse and observable representation of macro risk while mitigating overfitting. Commodity-specific macro sensitivities are then obtained by regressing each commodity's residual component on the selected macro risk proxies.

From a methodological perspective, the paper contributes along several dimensions. First, before extracting the latent factors, we partial out the observed sector-return space from the commodity-specific risk characteristics. This ensures that factor extraction targets latent common variation that is incremental to observed market-return movements, thereby separating the two sources of common variation. Second, in contrast to standard factor-augmented models, which approximate the macro component using a small and fixed number of latent factors (e.g., \citet{Bai2009,JuodisSarafidis2018,NorkuteEtal2021,JuodisSarafidis2022b}), we search over a high-dimensional set of observable macro-financial indicators, while allowing only a sparse subset to be relevant. Our candidate set comprises 36 indicators spanning uncertainty, sentiment, market attention, and global financial conditions, assembled from a broad range of studies.\footnote{To the best of our knowledge this is
the first study to employ such a diverse set of indicators within a single framework.} This yields an observable representation of macro risk without requiring the relevant indicators to be specified a priori. Third, to our knowledge, integrating PCA with high-dimensional selection in this way is new to the econometric and broader statistical literature: rather than applying boosting directly to the commodity-specific residuals, it is applied to principal components extracted by PCA. Selection therefore targets the common residual variation and produces a common set of observable macro-financial proxies. We examine the finite-sample performance of this procedure through a Monte Carlo study reported in the Appendix, which shows that the proposed approach performs well in recovering the relevant macro-financial proxies and the associated sensitivities. Fourth, we derive a commodity-specific over-identifying restrictions test for the defactored moment conditions, which provides the basis for a panel-level diagnostic obtained by combining individual test statistics across commodities. Finally, with respect to the RIIs, we formally establish conditions under which the three risk layers are non-reducible: in general, micro, market, and economy-wide risk intensity cannot be collapsed into a single aggregate risk component without losing information about the underlying structure of commodity risk.

The empirical analysis reveals a clear anatomy of commodity risk. For the full sample, market risk is the dominant source, accounting for 42.1\% of total risk intensity, followed by economy-wide and micro risk at 29.7\% and 28.2\%, respectively. The composition of risk nevertheless varies substantially across commodity sectors. Market risk is particularly important for energy commodities, where it accounts for 52.4\% of total risk intensity, compared with 35.4\% for agricultural commodities and 38.1\% for metals. Economy-wide risk, by contrast, is considerably more prominent in agriculture and metals than in energy. The time-varying RIIs also exhibit pronounced increases around major episodes of commodity-market stress, including the COVID-19 pandemic and associated oil-market collapse in 2020, the global energy-market disruptions of late 2021, Russia’s invasion of Ukraine in 2022, the escalation of global trade tensions in April 2025, and the renewed commodity-market turmoil in early 2026. Finally, the RIIs contain forward-looking information about subsequent commodity risk: micro- and market-level risk intensity significantly predict future volatility and absolute returns over horizons of up to twelve periods.

The distribution of risk intensity across commodities reveals substantial differences in concentration across the three risk layers. Micro and market-level risk are relatively concentrated among a subset of commodities: the top 20\% account for approximately half of total intensity in both cases, and the top 10\% account for around one-third. Economy-wide risk is more dispersed, with the corresponding shares falling to 40.6\% and 26.6\%, respectively. Aggregate risk intensity is more evenly distributed still, indicating that concentration within individual risk layers does not translate mechanically into concentration of overall commodity risk.

These findings have implications for both market participants and policymakers. Commodity risk reflects the interaction of commodity-specific conditions, common market forces, and broader macro-financial conditions, with their relative importance varying substantially across commodities and over time. For investors and risk managers, the RIIs provide information on both the concentration of risk across commodities and the sources underlying prevailing risk conditions. Their ability to predict future volatility and absolute returns further gives the indices a forward-looking dimension, suggesting their potential usefulness for monitoring emerging risk conditions and supporting portfolio and risk-management decisions. For policymakers, distinguishing market-level from economy-wide risk can help assess whether periods of elevated commodity risk are primarily associated with conditions within commodity markets or with broader economic and financial developments. More broadly, the “divide-and-conquer” framework can be applied to other settings in which risk reflects multiple sources and a common economy-wide component is latent, particularly when a large set of observable macro-financial indicators is available to proxy the underlying latent component.

\section{Model and Measurement of Layered Risk \label{sec:specification}}
\subsection{Model Specification and Risk Measurement}

We consider the following heterogeneous panel data model of commodity log-returns:
\begin{equation}
r_{i,t}
=
\rho_{i} r_{-i,t} + \boldsymbol{\alpha}_{i}^{\prime}\mathbf{x}_{i,t}
+
\boldsymbol{\beta}_{i}^{\prime}\mathbf{y}_{i,t}
+
\boldsymbol{\gamma}_{i}^{\prime}\mathbf{g}_{t}
+
\varepsilon_{i,t},
\quad
i=1,\dots,N; \quad t=1,\dots,T,
\label{model}
\end{equation}
where $r_{i,t}$ denotes the log-return of commodity $i$ in week $t$.
Let $\mathcal{S}_i$ denote the set of commodities belonging to the same sector as commodity $i$, and let $N_i=|\mathcal{S}_i|$, the cardinality of the $i$th sector. We define the leave-one-out sector return as
\begin{equation}
r_{-i,t}
=
\frac{1}{N_i-1}
\sum_{\substack{j\in\mathcal{S}_i\\j\neq i}}
r_{j,t}.
\label{eq:loo_return}
\end{equation}
Thus, $r_{-i,t}$ captures contemporaneous return movements elsewhere in commodity $i$'s sector while excluding commodity $i$ itself, and $\rho_i$ measures commodity $i$'s heterogeneous sensitivity to these market-level return movements.

The vector
\[
\mathbf{x}_{i,t}
=
(\mathrm{VLT}_{i,t},\mathrm{SKW}_{i,t},
\mathrm{KTS}_{i,t},\mathrm{OI}_{i,t})'
\in\mathbb{R}^{K_x}
\]
collects time-varying commodity-specific risk characteristics: volatility, realised skewness, realised kurtosis, and logged open interest (see Section~\ref{sec:data} for more details). These variables capture distinct dimensions of commodity-specific conditions. Volatility measures the dispersion of returns, while skewness and kurtosis capture asymmetry and tail behaviour, respectively. Open interest captures the scale of outstanding futures positions and provides information about market participation and trading activity specific to each commodity.

The vector
\[
\mathbf{y}_{i,t}
=
(\mathrm{CVLT}_{i,t},\mathrm{CSKW}_{i,t},
\mathrm{CKTS}_{i,t},\mathrm{COI}_{i,t})'
\in\mathbb{R}^{K_y}
\]
collects the leave-one-out sector-level counterparts of the commodity-specific risk characteristics in $\mathbf{x}_{i,t}$. Specifically, $\mathrm{CVLT}_{i,t}$, $\mathrm{CSKW}_{i,t}$, $\mathrm{CKTS}_{i,t}$, and $\mathrm{COI}_{i,t}$ measure volatility, realised skewness, realised kurtosis, and open interest, respectively, among the remaining commodities in commodity $i$'s sector. Each element is constructed as a leave-one-out average, so that $\mathbf{y}_{i,t}$ varies across both $i$ and $t$. The leave-one-out construction captures risk conditions prevailing elsewhere in the same sector while excluding commodity $i$'s own contribution. It therefore distinguishes commodity-specific risk conditions, $\mathbf{x}_{i,t}$, from corresponding sector-level risk conditions, $\mathbf{y}_{i,t}$.

Finally, $\mathbf{g}_{t}\in\mathbb{R}^{K_g}$ represents broader macro-financial risk, including uncertainty, sentiment, market attention, and global financial conditions. Unlike the micro and market layers, the relevant macro-financial variables are not specified a priori and are treated as latent at the estimation stage. Section~\ref{sec:Estimation} describes how an observable sparse representation of this layer is recovered from a
high-dimensional set of candidate macro-financial indicators. The disturbance $\varepsilon_{i,t}$ captures the remaining commodity-specific variation.

The slope coefficients $\boldsymbol{\alpha}_{i}$, $\boldsymbol{\beta}_{i}$, and $\boldsymbol{\gamma}_{i}$ are commodity-specific and capture heterogeneous sensitivities to the three risk layers. The coefficients $\boldsymbol{\alpha}_{i}$ measure sensitivity to commodity-specific risk characteristics, while $\boldsymbol{\beta}_{i}$ capture sensitivity to conditions
prevailing elsewhere in the same sector, which we refer to as market risk. The coefficients $\boldsymbol{\gamma}_{i}$ measure sensitivity to pervasive economy-wide macro-financial conditions. Throughout the paper, we refer to the three sources as the micro, market, and macro risk layers, respectively. Importantly, these sensitivities should be distinguished from prices of risk. They are estimated from conditional equations for realised commodity returns, whereas a risk-pricing interpretation would require an additional asset-pricing restriction linking exposures to expected returns and the estimation of the associated risk premia. Neither is imposed here.

Although the micro, market, and macro risk layers may be correlated, the coefficients represent partial effects conditional on the remaining sources of risk. This distinction is important because commodity-specific conditions, common movements elsewhere in the commodity market, and broader macro-financial conditions may coexist and need not affect individual commodities in the same way. The decomposition is therefore not intended to imply that the three layers originate from mutually exclusive primitive shocks; rather, it distinguishes their separate contributions to commodity returns conditional on the remaining risk sources.

Heterogeneous sensitivities alone do not reveal the relative importance of the different sources of commodity risk. We therefore map the estimated sensitivities into Risk Intensity Indices (RIIs). Throughout, $\widetilde{x}_{k,i,t}$, $\widetilde{y}_{k,i,t}$, and $\widetilde{g}_{k,t}$ denote standardized variables with zero mean and unit variance. Commodity-specific risk intensity is defined as the time-average
magnitude of the sensitivity-weighted impact associated with each risk layer:
\begin{equation}
\mathrm{RII}^{\mathrm{micro}}_{i}
=
\frac{1}{T}\sum_{t=1}^{T}
\left|
\sum_{k=1}^{K_x}
\alpha_{k,i}\widetilde{x}_{k,i,t}
\right|;
\quad
\mathrm{RII}^{\mathrm{market}}_{i}
=
\frac{1}{T}\sum_{t=1}^{T}
\left|
\sum_{k=1}^{K_y}
\beta_{k,i}\widetilde{y}_{k,i,t}
\right|;
\quad
\mathrm{RII}^{\mathrm{macro}}_{i}
=
\frac{1}{T}\sum_{t=1}^{T}
\left|
\sum_{k=1}^{K_g}
\gamma_{k,i}\widetilde{g}_{k,t}
\right|.
\label{eq:RII_i_all}
\end{equation}

Total commodity-level risk intensity is defined as
\begin{equation}
\mathrm{RII}^{\mathrm{total}}_{i}
=
\mathrm{RII}^{\mathrm{micro}}_{i}
+
\mathrm{RII}^{\mathrm{market}}_{i}
+
\mathrm{RII}^{\mathrm{macro}}_{i}.
\label{eq:RII_i_total}
\end{equation}

The RIIs combine the magnitude of estimated sensitivities with prevailing risk conditions. Thus, a large sensitivity need not imply high risk intensity when the corresponding risk condition is close to its typical level, while a more moderate sensitivity can be consequential when the associated risk condition is unusually pronounced. Standardisation places the variables on a common scale, allowing risk intensity to be compared across variables and layers.\footnote{Sensitivities are estimated using variables in their original economic units, preserving their economic interpretation. The RIIs are constructed using standardised variables to remove differences in measurement units and unconditional variability across risk measures.}

Absolute values are used in Eq. \eqref{eq:RII_i_all} because the RIIs measure the intensity rather than the direction of the sensitivity-weighted impact. The sign of an estimated coefficient determines the direction in which commodity returns are associated with a given risk variable, whereas the RII is concerned with the magnitude of that association under prevailing conditions. For example, our estimates indicate that, on average, commodity returns are positively associated with commodity-specific skewness, which captures asymmetry in the commodity's own return distribution, but negatively to market skewness, which captures asymmetry in returns elsewhere in the commodity's sector. These sensitivities operate in opposite directions, but both can represent important sources of risk intensity when the corresponding risk conditions are pronounced. Importantly, the absolute value is applied only after aggregating across variables within each risk layer. This allows individual variables within the same layer to reinforce or offset one another and ensures that the index measures the magnitude of the net sensitivity-weighted impact of that layer.

To examine how the anatomy of commodity risk evolves over time, we construct
cross-sectional average risk intensities for each layer:
\begin{align}
\mathrm{RII}^{\mathrm{micro}}_{t}
&=
\frac{1}{N}\sum_{i=1}^{N}
\left|
\sum_{k=1}^{K_x}
\alpha_{k,i}\widetilde{x}_{k,i,t}
\right|,
\nonumber\\
\mathrm{RII}^{\mathrm{market}}_{t}
&=
\frac{1}{N}\sum_{i=1}^{N}
\left|
\sum_{k=1}^{K_y}
\beta_{k,i}\widetilde{y}_{k,i,t}
\right|,
\nonumber\\
\mathrm{RII}^{\mathrm{macro}}_{t}
&=
\frac{1}{N}\sum_{i=1}^{N}
\left|
\sum_{k=1}^{K_g}
\gamma_{k,i}\widetilde{g}_{k,t}
\right|.
\label{eq:RII_t_layers}
\end{align}

The corresponding total risk intensity is
\begin{equation}
\mathrm{RII}^{\mathrm{total}}_{t}
=
\mathrm{RII}^{\mathrm{micro}}_{t}
+
\mathrm{RII}^{\mathrm{market}}_{t}
+
\mathrm{RII}^{\mathrm{macro}}_{t}.
\label{eq:RII_t_total}
\end{equation}

\subsection{Conceptual Framework and Motivation}

The specification in Eq.~\eqref{model} distinguishes three sources of commodity risk: conditions specific to commodity $i$, conditions prevailing elsewhere in the commodity market, and broader macro-financial conditions. The distinction between the first two layers deserves particular attention because the market variables are constructed as leave-one-out aggregates of
the corresponding commodity-specific characteristics. Although this creates an explicit link between the micro and market layers, it does not make them equivalent. The micro layer captures the conditions of commodity $i$ itself, whereas the market layer captures its sensitivity to conditions prevailing among the remaining commodities.

More generally, correlation across the three risk layers does not imply that one can be recovered from another. Aggregating commodity-specific risk characteristics does not reproduce the sensitivity-weighted market contribution for a given commodity, while correlation between market and macro-financial conditions does not make their respective contributions equivalent. The
following theorem formalises this non-reducibility and establishes the conditions under which such equivalence could arise.

\begin{theorem}[Non-reducibility of layered risk]
\label{theorem:non_reducibility}
For expositional simplicity, consider one scalar characteristic from each risk layer and define the model-implied layered contribution for commodity $i$ as $\ell_{i,t}=\alpha_i x_{i,t}+\beta_i y_{i,t}+\gamma_i g_t$. Let $\mathcal S_i$ denote the sector containing commodity $i$, with $N_i=|\mathcal S_i|$, and suppose that the market  characteristic is constructed as the leave-one-out sector average $y_{i,t}=(N_i-1)^{-1}\sum_{j\in\mathcal S_i,j\neq i}x_{j,t}$. Suppose further that market and macro conditions are related according to $y_{i,t}=\lambda_i g_t+\eta_{i,t}$, where $\lambda_i\neq0$ and $g_t$ is not identically zero. Then:

\begin{enumerate}
\renewcommand{\labelenumi}{(\roman{enumi})}

\item The aggregated micro contribution of the remaining commodities in sector $\mathcal S_i$ and the market contribution are distinct unless
\[
\sum_{\substack{j\in\mathcal S_i\\j\neq i}}
\left(\alpha_j-\frac{\beta_i}{N_i-1}\right)x_{j,t}=0
\qquad\text{for every }t.
\]
A sufficient condition for the two contributions to coincide is
$\alpha_j=\beta_i/(N_i-1)$ for $j\in\mathcal S_i$, $j\neq i$.
This condition is also necessary if the sample second-moment matrix of the characteristics $\{x_{j,t}:j\in\mathcal S_i,\ j\neq i\}$ is nonsingular, or if equality is required as an identity for every possible realization of these characteristics.

\item The macro contribution $\gamma_i g_t$ and the market contribution $\beta_i y_{i,t}$ are distinct unless
\[
\beta_i\eta_{i,t}=(\gamma_i-\beta_i\lambda_i)g_t
\qquad\text{for every }t.
\]
Thus, when $\beta_i\neq0$, equality between the two contributions requires $\eta_{i,t}=(\gamma_i/\beta_i-\lambda_i)g_t$ for every $t$. In other words, the market-specific component $\eta_{i,t}$ must be exactly proportional to the macro factor $g_t$, with a proportionality coefficient
determined by the market and macro sensitivities. If $g_t$ and $\eta_{i,t}$ are linearly independent over the sample, this condition cannot hold for $\beta_i\neq0$. In that case, equality between the two contributions is possible only in the trivial case $\beta_i=\gamma_i=0$.

\end{enumerate}
\end{theorem}

\begin{proof}
See Online~\ref{reg_conditions_lemmas}.
\end{proof}

The first statement of Theorem~\ref{theorem:non_reducibility} shows that aggregation of micro contributions does not, in general, recover the market contribution. Exact equality requires the coefficient discrepancies between the two contributions to generate a linear combination of the underlying commodity characteristics that vanishes at every sample date. A sufficient condition is $\alpha_j=\beta_i/(N_i-1)$ for all $j\in\mathcal S_i$, $j\neq i$, which is also necessary when the corresponding sample second-moment matrix is nonsingular. Thus, absent exact linear dependence among the underlying characteristics, equality requires the micro sensitivities to satisfy a restrictive proportionality condition determined by commodity $i$'s market sensitivity and the number of commodities in its sector.

The second statement establishes an analogous result for market and macro risk. Given $y_{i,t}=\lambda_i g_t+\eta_{i,t}$, equality between the market
and macro contributions requires the market-specific component
$\eta_{i,t}$ to be exactly proportional to the macro factor $g_t$, with the
proportionality coefficient determined by the market and macro sensitivities.
Thus, the presence of market-specific variation does not by itself establish
non-reducibility; what matters is whether that variation contains a component
that is linearly independent of the macro factor. In particular, if $g_t$
and $\eta_{i,t}$ are linearly independent, the market and macro contributions
cannot coincide except in the trivial case $\beta_i=\gamma_i=0$.

\begin{remark}
Theorem~\ref{theorem:non_reducibility} characterizes non-reducibility in terms of the restrictions required for distinct risk layers to generate identical contributions. Correlation across layers is not sufficient for one layer to subsume another: exact reducibility requires specific restrictions on either the sensitivities and aggregation structure or the dependence between the underlying risk drivers.
\end{remark}

If the macro-risk variables $\mathbf{g}_t$ were directly observed, the model in Eq.~\eqref{model} could, under standard exogeneity conditions, be estimated commodity by commodity using least squares. In practice, however, $\mathbf{g}_t$ is latent. Because latent macro risk may be correlated with the observed micro- and market-level variables, omitting it induces endogeneity and renders ordinary least squares estimation inconsistent.

Since the macro-risk variables are unobserved, they need to be proxied by observable macro-financial indicators. The difficulty is that the set of potentially relevant indicators is large, diverse, and highly collinear, with limited guidance as to which indicators are empirically relevant.\footnote{Macro risk itself is not assumed to be high-dimensional. High dimensionality arises from the set of observable macro-financial indicators considered as potential proxies for the latent macro-risk process.} Uncertainty, sentiment, and financial conditions, for example, can each be represented by several competing measures, often available across different dimensions and horizons.\footnote{A case in point is ``economic uncertainty'', which can be proxied by several alternative indices \citep{BakerBloomDavis2016,CaldaraIacoviello2022,AhirBloomFurceri2022,JuradoLudvigsonNg2015}, often available at different forecast horizons and levels of geographic aggregation.} Identifying the informative subset therefore involves a trade-off between omitted-variable bias, if relevant indicators are excluded, and overfitting, if too many noisy and collinear proxies are  retained.\footnote{This issue parallels the broader problem of variable proliferation in empirical finance, where a growing set of candidate predictors complicates model selection and increases the risk of overfitting and spurious inference \citep{Cochrane2011,HarveyLiuZhu2016}.}

Our ``divide-and-conquer'' procedure addresses these challenges sequentially. In Stage~1, heterogeneous sensitivities to observed micro- and market-level risk are estimated using commodity-specific instrumental-variable regressions. IV identification exploits a latent common-factor structure in the commodity-specific risk characteristics that is sufficiently rich to span the latent macro component responsible for endogeneity. Instruments are constructed from variation in the observed risk characteristics after removing this latent common component. Despite the endogeneity induced by the latent macro component, the resulting estimator consistently recovers the heterogeneous sensitivities to the observed risk layers.

In Stage~2, we extract the $r$ leading principal components from the Stage~1 residuals, where $r$ is determined from the data using eigenvalue-based procedures or information criteria
\citep{AhnHorenstein2013,BaiNg2002}. These components span the dominant directions of common residual variation. We then map each component to a high-dimensional set of observable macro-financial indicators using the Boosting with Multiple Testing procedure of \citet{KapetaniosEtal2026}. Collecting the indicators selected across components yields a sparse set of observable proxies for macro risk.\footnote{Online~\ref{sec:Appendix_MC} examines the finite-sample performance of the PCA--BMT procedure.}

The commodity-level RIIs in Eqs.~\eqref{eq:RII_i_all}--\eqref{eq:RII_i_total} characterise the cross-sectional distribution of risk intensity and allow us to assess its concentration across individual commodities. The time-varying
indices in Eqs.~\eqref{eq:RII_t_layers}--\eqref{eq:RII_t_total} instead track the evolution and changing composition of commodity risk over time. Together, these measures show where risk intensity is concentrated across commodities and how the relative importance of its micro, market, and macro sources evolves over time.

It is useful at this point to distinguish the RIIs from measures used to quantify downside financial risk, such as Value-at-Risk (VaR) and Expected Shortfall (ES). VaR and ES characterise the tail of the return distribution at a specified horizon and confidence level, addressing the magnitude of potential losses under adverse outcomes. The RIIs address a different question: what is the relative importance of the micro, market, and macro sources? Thus, while VaR and ES characterise the distribution and severity of downside losses, the RIIs characterise the relative intensity and composition of the different sources of commodity risk.

\section{Divide-and-Conquer Estimation Framework \label{sec:Estimation}}

This section provides the analytical formulation of the ``divide-and-conquer'' estimation framework for model~\eqref{model}. We begin by making explicit that the macro-risk component is latent:
\begin{equation}
r_{i,t}
=\rho_{i} r_{-i,t} +
\boldsymbol{\alpha}_{i}^{\prime}\mathbf{x}_{i,t}
+
\boldsymbol{\beta}_{i}^{\prime}\mathbf{y}_{i,t}
+
u_{i,t},
\qquad
u_{i,t}
=
\boldsymbol{\gamma}_{i}^{\prime}\mathbf{g}_{t}
+
\varepsilon_{i,t},
\quad
i=1,\dots,N;\quad t=1,\dots,T.
\label{model_with_u}
\end{equation}
The macro-risk vector $\mathbf g_t$ is initially unobserved and is therefore absorbed into the composite error $u_{i,t}$.
Stacking the $T$ observations for commodity $i$ gives
\begin{equation}
\mathbf r_i
=\rho_{i} \mathbf{r}_{-i} +
\mathbf X_i\boldsymbol{\alpha}_i
+
\mathbf Y_i\boldsymbol{\beta}_i
+
\mathbf u_i,
\qquad
\mathbf u_i
=
\mathbf G\boldsymbol{\gamma}_i
+
\boldsymbol{\varepsilon}_i,
\label{model_vectori}
\end{equation}
where $\mathbf r_i=(r_{i,1},\ldots,r_{i,T})'$,
$\mathbf X_i=(\mathbf x_{i,1},\ldots,\mathbf x_{i,T})'$
is $T\times K_x$, $\mathbf Y_i=(\mathbf y_{i,1},\ldots,\mathbf y_{i,T})'$ is $T\times K_y$, $\mathbf G=(\mathbf g_1,\ldots,\mathbf g_T)'$, and $\boldsymbol{\varepsilon}_i
=(\varepsilon_{i,1},\ldots,\varepsilon_{i,T})'$.

A first challenge in estimating Eq.~\eqref{model_vectori} is that the commodity-specific and market-level risk characteristics may be correlated with the latent macro-financial conditions that also enter the commodity
return equation, $\mathbf G\boldsymbol{\gamma}_i$, giving rise to endogeneity.
A separate issue is that the commodity-specific characteristics may co-move with the observed leave-one-out sector return, $r_{-i,t}$. We distinguish this observed source of common variation from the latent common variation in the commodity-specific characteristics through the representation
\begin{equation}
\mathbf{x}_{i,t}
=
\boldsymbol{\phi}_{i} r_{-i,t}
+
\boldsymbol{\Lambda}_{i}\mathbf f_t
+
\mathbf v_{i,t},
\label{x_factor}
\end{equation}
where $\mathbf f_t\in\mathbb R^{K_f}$ denotes a vector of latent common
factors, $\boldsymbol{\phi}_i$ and $\boldsymbol{\Lambda}_i$ contain the
corresponding commodity-specific loadings, and $\mathbf v_{i,t}$ is
idiosyncratic. The first term captures common variation in the
commodity-specific characteristics associated with observed sector-return
movements, while the second captures latent common variation that is
incremental to the observed sector-return space. The precise identifying
separation between these two components is stated below.

Note that because the market-level risk variables $\mathbf{y}_{i,t}$ are constructed as leave-one-out averages of the corresponding commodity-specific characteristics, they inherit both the sector-return-related and latent common variation present in the cross-section of $\mathbf{x}_{j,t}$. In particular, the latent factor structure represented by $\mathbf{f}_t$ is also embedded in $\mathbf{y}_{i,t}$. Hence, $\mathbf{y}_{i,t}$ does not provide an independent source from which the latent factor space needs to be estimated.

For compactness, define
\begin{equation}
\mathbf C_i
=
(\mathbf r_{-i},\mathbf X_i,\mathbf Y_i),
\qquad
\boldsymbol{\theta}_i
=
(\rho_i,\boldsymbol{\alpha}_i',
\boldsymbol{\beta}_i')',
\end{equation}
so that
\begin{align}
\mathbf r_i
&=
\mathbf C_i\boldsymbol{\theta}_i+\mathbf u_i,
\label{compact}\\
\mathbf X_i
&=
\mathbf r_{-i}\boldsymbol{\phi}_i'
+
\mathbf F^0\boldsymbol{\Lambda}_i'
+
\mathbf V_i,
\label{pca_vectori}
\end{align}
where
$\mathbf F^0=(\mathbf f_1,\ldots,\mathbf f_T)'$
denotes the $T\times K_f$ population factor matrix.

Our two-stage estimation proceeds as follows.

\noindent
\textbf{Stage 1: Micro and Market Risk}

We first partial out the observed sector-return variation from the commodity-specific risk characteristics. Let $\mathcal S_1,\ldots,\mathcal S_J$ denote the $J$ commodity sectors, with $N_s=|\mathcal S_s|$ and $\sum_{s=1}^J N_s=N$, and define the sector-average return
$$
\bar{\mathbf r}_s
=
\frac{1}{N_s}
\sum_{j\in\mathcal S_s}\mathbf r_j,
\qquad s=1,\ldots,J.
$$
Collect the sector-average returns in the $T\times J$ matrix $\mathbf R = \left(
\bar{\mathbf r}_1,\ldots,\bar{\mathbf r}_J
\right)$, and define the associated residual-maker matrix $\mathbf M_{\mathbf R} = \mathbf I_T
- \mathbf R(\mathbf R'\mathbf R)^{-1}\mathbf R'$.
Rather than partialling out a different leave-one-out sector return for each commodity before factor extraction, we project all commodity-specific risk characteristics onto the orthogonal complement of the common space spanned by the observed sector returns. This removes the pervasive observed sector-return component in Eq.~\eqref{pca_vectori} while applying the same transformation to every commodity. In particular, for $i\in\mathcal S_s$,
$$
\mathbf r_{-i}
=
\frac{N_s}{N_s-1}\bar{\mathbf r}_s
-
\frac{1}{N_s-1}\mathbf r_i,
$$
and therefore
$$
\mathbf M_{\mathbf R}\mathbf r_{-i}
=
-\frac{1}{N_s-1}
\mathbf M_{\mathbf R}\mathbf r_i.
$$
Hence, under the regularity conditions stated below,  the appropriately normalised magnitude of the sector-return component remaining after the common projection is of order $N_s^{-1}$ and is asymptotically negligible as the sector size increases.

For lagged characteristics, the projection is defined using the correspondingly lagged sector-return space. Specifically, for $\tau=0,\ldots,\max\{\zeta_x,\zeta_y\}$, let
\[
\mathbf R_{-\tau}
=
\left(
\bar{\mathbf r}_{1,-\tau},\ldots,
\bar{\mathbf r}_{J,-\tau}
\right),
\]
where $\bar{\mathbf r}_{s,-\tau}$ denotes the sector-average return
aligned with $\mathbf X_{i,-\tau}$, and define
\[
\mathbf M_{\mathbf R_{-\tau}}
=
\mathbf I_T
-
\mathbf R_{-\tau}
\left(
\mathbf R_{-\tau}'\mathbf R_{-\tau}
\right)^{-1}
\mathbf R_{-\tau}'.
\]
For $\tau=0$, $\mathbf R_{-0}=\mathbf R$ and
$\mathbf M_{\mathbf R_{-0}}=\mathbf M_{\mathbf R}$.

We define $\mathbf F^0_{-\tau}$ as the latent pervasive component of the commodity-specific risk characteristics that remains after accounting for the observed sector-return space. Accordingly, the decomposition is identified through the normalization
$\mathbf R_{-\tau}'\mathbf F^0_{-\tau}
= \mathbf 0$, $\tau=0,\ldots,\max\{\zeta_x,\zeta_y\}$. This normalization distinguishes the latent common component from the observed sector-return component in the reduced-form representation of the commodity-specific characteristics; it does not imply that the economic forces underlying the two components are mutually independent.

It follows that, for $i\in\mathcal S_s$,
$$
\mathbf M_{\mathbf R_{-\tau}}\mathbf X_{i,-\tau}
=
\mathbf F^0_{-\tau}\boldsymbol{\Lambda}_i'
+
\mathbf M_{\mathbf R_{-\tau}}\mathbf V_{i,-\tau}
+
\mathbf D_{i,-\tau},
$$
where
$$
\mathbf D_{i,-\tau}
=
\mathbf M_{\mathbf R_{-\tau}}
\mathbf r_{-i,-\tau}\boldsymbol{\phi}_i',
$$
and the contribution of $\mathbf D_{i,-\tau}$ is asymptotically negligible as $N_s\rightarrow\infty$.

We therefore estimate the latent factor space separately for each $\tau=0,\ldots,\max\{\zeta_x,\zeta_y\}$ by applying PCA to the projected characteristics. Denote the resulting estimated factor matrices by $\widehat{\mathbf F}_{-\tau}$. Up to rotation, they are obtained from the leading eigenvectors of
\begin{equation}
\frac{1}{NT}
\sum_{i=1}^{N}
\mathbf M_{\mathbf R_{-\tau}}
\mathbf X_{i,-\tau}
\mathbf X_{i,-\tau}'
\mathbf M_{\mathbf R_{-\tau}}.
\label{factor_covariance}
\end{equation}
The number of factors, $K_f$, is determined using standard eigenvalue-based criteria.

\begin{remark}
The projection on the observed sector-return space is required because sector returns themselves constitute a pervasive source of common variation in the commodity-specific risk characteristics. If PCA were applied directly to the unprojected characteristics, the extracted components could combine common variation associated with observed sector returns with the latent
common variation of interest. By first removing the finite-dimensional sector-return space, PCA is instead applied to the remaining common variation. Because the same sector-return space is removed from every commodity at a given lag, the projected characteristics retain a common factor representation and PCA can exploit information from the full cross-section.
\end{remark}

Last, we eliminate the influence of $\widehat{\mathbf F}$ from the entire model in Eq.~\eqref{compact} using the transformation
\begin{equation}
\mathbf M_{\widehat{\mathbf F}}\mathbf r_i
=
\mathbf M_{\widehat{\mathbf F}}\mathbf r_{-i}\rho_i
+
\mathbf M_{\widehat{\mathbf F}}\mathbf X_i\boldsymbol{\alpha}_i
+
\mathbf M_{\widehat{\mathbf F}}\mathbf Y_i\boldsymbol{\beta}_i
+
\mathbf M_{\widehat{\mathbf F}}\mathbf u_i
=
\mathbf M_{\widehat{\mathbf F}}\mathbf C_i\boldsymbol{\theta}_i
+
\mathbf M_{\widehat{\mathbf F}}\mathbf u_i .
\label{projected_model}
\end{equation}
The factor space recovered in Eq.~\eqref{factor_covariance} is statistical rather than inherently macro-financial. Common variation across commodity risk characteristics may reflect, for example, common trading or hedging activity, financialisation, sectoral conditions, as well as broader macro-financial forces. We therefore do not require the latent factor matrix $\mathbf F^0$ to coincide with the macro-risk matrix $\mathbf G$. For identification, it is sufficient that $\operatorname{col}(\mathbf G) \subseteq \operatorname{col}(\mathbf F^0)$, or equivalently that $\mathbf G=\mathbf F^0\mathbf A'$ for some conformable
matrix $\mathbf A$. This spanning condition requires the latent macro-financial variation affecting commodity returns to be contained in the common factor space of the commodity-specific risk characteristics.
Consequently, $\mathbf M_{\mathbf F^0}\mathbf G=\mathbf 0$, and consistency of the estimated factor space implies that the transformation in Eq.~\eqref{projected_model} asymptotically removes the latent macro component from the estimating equation. Importantly, the spanning condition forms part of the maintained restrictions underlying instrument validity, which can be assessed indirectly using the over-identification tests developed below.

To estimate $\boldsymbol{\theta}_i$ we construct instruments for the observed micro- and market-level variables in the defactored model. The instrument set is given by\footnote{In practice, the lag orders $\zeta_x$ and $\zeta_y$ are selected using the Lasso approach of
\citet{BelloniEtal2012}, as in \citet{ChenEtal2025}.}
\begin{align}
\widehat{\mathbf Z}_i
=
\Big(
&
\mathbf r_{-i},\;
\mathbf r_{-i,-1},\ldots,
\mathbf r_{-i,-\zeta_r},
\nonumber\\
&
\mathbf M_{\widehat{\mathbf F}}\mathbf X_i,\;
\mathbf M_{\widehat{\mathbf F}_{-1}}\mathbf X_{i,-1},
\ldots,
\mathbf M_{\widehat{\mathbf F}_{-\zeta_x}}
\mathbf X_{i,-\zeta_x},
\nonumber\\
&
\mathbf M_{\widehat{\mathbf F}}\mathbf Y_i,\;
\mathbf M_{\widehat{\mathbf F}_{-1}}\mathbf Y_{i,-1},
\ldots,
\mathbf M_{\widehat{\mathbf F}_{-\zeta_y}}
\mathbf Y_{i,-\zeta_y}
\Big).
\label{Instruments_matrix}
\end{align}
Instrument relevance is provided by the idiosyncratic component of the commodity-specific risk characteristics. Defactoring removes their latent common component but preserves idiosyncratic variation that remains informative about the corresponding endogenous regressors. Formally, this relevance requirement is captured by the full-column-rank condition on $\mathbf A_{i}$ in Assumption 5 below.

For the micro- and market-risk blocks, projecting the contemporaneous and lagged variables onto the orthogonal complement of $\widehat{\mathbf F}_{-\tau}$ asymptotically removes the latent common component underlying their endogeneity. The resulting defactored variables and their selected lags therefore provide valid instruments for the
commodity-specific and market-level risk characteristics under the maintained moment conditions.

The leave-one-out sector return $\mathbf r_{-i}$ is treated separately. Unlike the micro- and market-risk characteristics, it is not defactored when constructing the instrument set. The observed sector-return space is
separated from the latent factor space in the factor decomposition above, while the spanning condition ensures that the latent macro component of the composite error is removed by the transformation in
Eq.~\eqref{projected_model}. Subject to the maintained orthogonality conditions with the remaining idiosyncratic disturbance, the contemporaneous leave-one-out sector return and its selected lags can therefore enter the instrument set directly.

Using the instruments in Eq.~\eqref{Instruments_matrix}, the
commodity-specific IV estimator of $\boldsymbol{\theta}_{i}$ is defined as
\begin{equation}
\widehat{\boldsymbol{\theta}}_{i}
=
\left(
\widehat{\mathbf A}_{i,T}^{\prime}
\widehat{\mathbf B}_{i,T}^{-1}
\widehat{\mathbf A}_{i,T}
\right)^{-1}
\widehat{\mathbf A}_{i,T}^{\prime}
\widehat{\mathbf B}_{i,T}^{-1}
\widehat{\mathbf c}_{i,T},
\label{ivi}
\end{equation}
\begin{equation}
\widehat{\mathbf A}_{i,T}
=
\frac{1}{T}
\widehat{\mathbf Z}_{i}^{\prime}
\mathbf M_{\widehat{\mathbf F}}\mathbf C_i,
\qquad
\widehat{\mathbf B}_{i,T}
=
\frac{1}{T}
\widehat{\mathbf Z}_{i}^{\prime}
\mathbf M_{\widehat{\mathbf F}}
\widehat{\mathbf Z}_{i},
\qquad
\widehat{\mathbf c}_{i,T}
=
\frac{1}{T}
\widehat{\mathbf Z}_{i}^{\prime}
\mathbf M_{\widehat{\mathbf F}}\mathbf r_i .
\label{tabgi}
\end{equation}
The IV estimator in Eq.~\eqref{ivi} exploits the population orthogonality condition that, for each $i$,
\begin{equation}
p\!\lim_{T\to\infty} T^{-1}\mathbf{Z}_{i}^{\prime}\mathbf{M}_{\mathbf{F}^{0}}\mathbf{u}_{i}=\mathbf{0}.
\label{moments_pop}
\end{equation}
For a generic parameter vector $\boldsymbol\theta$, define the sample moment vector
\begin{equation}
\overline{\mathbf g}_{iT}(\boldsymbol\theta)
=
\frac{1}{T}
\widehat{\mathbf Z}_i'
\mathbf M_{\widehat{\mathbf F}}
\left(
\mathbf r_i-\mathbf C_i\boldsymbol\theta
\right)
=
\widehat{\mathbf c}_{i,T}
-
\widehat{\mathbf A}_{i,T}\boldsymbol\theta.
\label{sample_moments}
\end{equation}
The estimator in Eq.~\eqref{ivi} satisfies the first-order conditions
\begin{equation}
\widehat{\mathbf A}_{i,T}'
\widehat{\mathbf B}_{i,T}^{-1}
\overline{\mathbf g}_{iT}
\left(
\widehat{\boldsymbol\theta}_i
\right)
=
\mathbf 0.
\label{moments_sample}
\end{equation}
When $q_i>p$, these $p$ first-order conditions do not in general imply that all $q_i$ sample moments in
$\overline{\mathbf g}_{iT}(\widehat{\boldsymbol\theta}_i)$ are individually equal to zero.

Under the regularity conditions stated below and instrument validity, $\widehat{\boldsymbol{\theta}}_i$ is $\sqrt{T}$-consistent and asymptotically normal. Once the latent common component is projected out, the remaining estimation problem takes the form of a standard commodity-specific IV regression. We summarise below the main conditions governing identification, factor structure, heterogeneity, and limiting inference. More primitive moment, dependence, factor-estimation, and projection conditions used to establish these high-level results are collected in Online Appendix D.


\begin{description}

\item[\textbf{Assumption 1 (Idiosyncratic errors).}]
$\varepsilon_{i,t}$ and $\text{v}_{k,i,t}$ have zero mean, finite $(8+\delta)$ moments for some $\delta>0$, and satisfy weak dependence conditions over $t$ and independence across $i$.

\item[\textbf{Assumption 2 (Common component structure).}]
$\mathbf f_t=\boldsymbol C_x(L)\mathbf q_{f,t}$ and
$\mathbf g_t=\boldsymbol C_g(L)\mathbf q_{g,t}$, where
$\boldsymbol C_x(L)$ and $\boldsymbol C_g(L)$ are absolutely summable,
$\mathbf q_{f,t}\sim i.i.d.(\mathbf 0,\boldsymbol\Sigma_f)$ and
$\mathbf q_{g,t}\sim i.i.d.(\mathbf 0,\boldsymbol\Sigma_g)$, with finite
fourth moments. Letting
$\mathbf F^0=(\mathbf f_1,\ldots,\mathbf f_T)'$ and
$\mathbf G=(\mathbf g_1,\ldots,\mathbf g_T)'$, we assume $\mathrm{col}(\mathbf G) \subseteq
\mathrm{col}(\mathbf F^0)$.

\item[\textbf{Assumption 3 (Sector structure and factor loadings).}]
The number of sectors $J$ is fixed and, for each
$s=1,\ldots,J$, $\frac{N_s}{N}\longrightarrow\pi_s$,
$0<\pi_s<1$. The loadings $\boldsymbol{\phi}_i$, $\boldsymbol{\Lambda}_i$, and $\boldsymbol{\gamma}_i$ have finite fourth moments and are independent of $\varepsilon_{i,t}$ and $\mathbf v_{i,t}$. Moreover, for each sector, $\frac{1}{N_s} \sum_{i\in\mathcal S_s}
\boldsymbol{\Lambda}_i \boldsymbol{\Lambda}_i'
\overset{p}{\longrightarrow} \mathbf Q_{\Lambda,s}$,
$\mathbf Q_{\Lambda,s}>0$. The relevant remaining second-moment matrices are positive definite.

\item[\textbf{Assumption 4 (Random coefficients).}]
$\boldsymbol{\theta}_i
=
\boldsymbol{\theta}+\boldsymbol{\eta}_i$,
where
\[
\boldsymbol{\theta}_i
=
(\rho_i,
 \boldsymbol{\alpha}_i',
 \boldsymbol{\beta}_i')',
\]
and $\boldsymbol{\eta}_i$ is i.i.d. with mean zero, finite fourth moments, and sub-exponential tails.

\item[\textbf{Assumption 5 (Identification and limiting distribution).}]
Let $\mathcal H$ denote the sigma-field generated by the heterogeneous coefficients and loadings. For each fixed $i$, the following conditions hold conditionally on $\mathcal H$, for almost every realization, as $N,T\to\infty$ with $N/T\to c\in(0,\infty)$. The instrument dimension $q_i\geq p$ is fixed. We assume
\[
\widehat{\mathbf A}_{i,T}\overset{p}{\to}\mathbf A_i,
\qquad
\widehat{\mathbf B}_{i,T}\overset{p}{\to}\mathbf B_i,
\]
where $\mathbf A_i$ has full column rank $p$ and $\mathbf B_i$ is positive definite. Moreover,
\[
\frac{1}{\sqrt T}
\mathbf Z_i'\mathbf M_{\mathbf F^0}\mathbf u_i
\overset{d}{\longrightarrow}
N(\mathbf 0,\boldsymbol\Omega_i),
\]
where $\boldsymbol\Omega_i$ is positive definite. For uniform results over $i$, $\mathbf B_i$ and $\boldsymbol\Omega_i$ are uniformly positive definite.

\end{description}
The full-column-rank condition on $\mathbf A_i$ is the standard IV relevance condition, requiring the instrument set to retain sufficient independent variation to identify the $p$ elements of $\boldsymbol{\theta}_{i}$. In the present setting, this identifying variation is provided by the idiosyncratic component of the observed risk characteristics that remains after their latent common component has been removed.

Additional high-level moment, dependence, pervasiveness, and regularity conditions required for the asymptotic arguments are reported in
Online~\ref{reg_conditions_lemmas}. These conditions ensure the uniform laws of large numbers and central limit results required for PCA-based factor extraction and commodity-specific IV estimation, as well as the
asymptotic negligibility of the residual leave-one-out sector-return component after projection on the observed sector-return space.

Theorem~\ref{TH2} formalises that, despite the endogeneity of the observed micro- and market-level variables induced by latent macro risk, the proposed commodity-specific IV estimator consistently recovers the heterogeneous sensitivities and supports standard inference under the stated conditions.

\begin{theorem}[Asymptotic Distribution of the Commodity-Specific IV Estimator]
\label{TH2}
Consider the model in Eqs.~\eqref{compact} and \eqref{x_factor}, under Assumptions 1--5 and the additional regularity conditions reported in the Online Appendix. Then, as $N,T\to\infty$ such that $N/T\rightarrow c$, with $0<c<\infty$, for each fixed $i$ and almost every realization of $\mathcal H$,
\begin{equation}
\sqrt{T}
\left(
\widehat{\boldsymbol{\theta}}_i-\boldsymbol{\theta}_i
\right)
\overset{d}{\longrightarrow}
N\left(
\mathbf 0,\,
\mathbf V_i
\right)
\qquad\text{conditionally on }\mathcal H,
\end{equation}
where
\begin{equation}
\mathbf V_i
=
(\mathbf A_i'\mathbf B_i^{-1}\mathbf A_i)^{-1}
\mathbf A_i'\mathbf B_i^{-1}
\boldsymbol{\Omega}_i
\mathbf B_i^{-1}\mathbf A_i
(\mathbf A_i'\mathbf B_i^{-1}\mathbf A_i)^{-1}.
\end{equation}
\end{theorem}

\begin{proof}
See Online~\ref{reg_conditions_lemmas}.
\end{proof}

The ``divide-and-conquer'' framework generates over-identifying restrictions that allow the validity of the instrument set to be assessed separately for each commodity. Importantly, these restrictions provide an indirect diagnostic relevant to the spanning condition in Assumption~2, according to which the latent macro-risk component is contained in the common factor space extracted from the commodity-specific risk characteristics after accounting for observed sector-return variation. We therefore derive a commodity-specific over-identification test, which also provides the basis for the panel-level diagnostic reported later in Section~\ref{sec:results}. Because the defactored moment conditions may exhibit serial dependence, the test is based on a heteroskedasticity and autocorrelation consistent (HAC) estimator of the long-run covariance matrix.

Let $\widehat{\mathbf u}_i = \mathbf r_i-\mathbf C_i\widehat{\boldsymbol\theta}_i$ denote the residual vector from the baseline IV estimator in
Eq.~\eqref{ivi}, and define
\[
\widehat{\mathbf m}_{i,t}
=
\widehat{\mathbf z}_{i,t}^{*}\,
\widehat u_{i,t}^{*},
\qquad
t=1,\dots,T,
\]
where $\widehat{\mathbf z}_{i,t}^{*}$ denotes the $t$th row of
$\mathbf M_{\widehat{\mathbf F}}\widehat{\mathbf Z}_i$ and
$\widehat u_{i,t}^{*}$ is the $t$th element of
$\mathbf M_{\widehat{\mathbf F}}\widehat{\mathbf u}_i$.
Define the sample autocovariance matrices
\begin{equation}
\widehat{\boldsymbol\Gamma}_{i,\ell}
=
\frac{1}{T}
\sum_{t=\ell+1}^{T}
\widehat{\mathbf m}_{i,t}
\widehat{\mathbf m}_{i,t-\ell}^{\prime},
\qquad
\ell=0,1,\dots,L_T,
\label{vc1}
\end{equation}
and the HAC estimator of the long-run covariance matrix
\begin{equation}
\widehat{\boldsymbol\Omega}_{iT}
=
\widehat{\boldsymbol\Gamma}_{i,0}
+
\sum_{\ell=1}^{L_T}
k\!\left(\frac{\ell}{L_T+1}\right)
\left(
\widehat{\boldsymbol\Gamma}_{i,\ell}
+
\widehat{\boldsymbol\Gamma}_{i,\ell}^{\prime}
\right),
\label{hac_omega}
\end{equation}
where $k(\cdot)$ is a kernel function and $L_T$ is a bandwidth parameter satisfying $L_T\to\infty$ and $L_T/T\to0$ as $T\to\infty$. Using $\widehat{\boldsymbol\Omega}_{iT}$ as the second-step weighting matrix, define the efficient two-step estimator
\begin{equation}
\widetilde{\boldsymbol\theta}_i
=
\left(
\widehat{\mathbf A}_{i,T}'
\widehat{\boldsymbol\Omega}_{iT}^{-1}
\widehat{\mathbf A}_{i,T}
\right)^{-1}
\widehat{\mathbf A}_{i,T}'
\widehat{\boldsymbol\Omega}_{iT}^{-1}
\widehat{\mathbf c}_{i,T}.
\label{efficient_iv}
\end{equation}
The corresponding residual vector is
\[
\widetilde{\mathbf u}_i
=
\mathbf r_i-\mathbf C_i\widetilde{\boldsymbol\theta}_i.
\]
The estimator $\widetilde{\boldsymbol\theta}_i$ is introduced solely for constructing the over-identification test; the baseline estimator $\widehat{\boldsymbol\theta}_i$ in Eq.~\eqref{ivi} remains the estimator used in the two-stage empirical procedure.
The commodity-specific over-identification statistic is then defined as
\begin{equation}
S_{iT}
=
\frac{1}{T}
\left(
\widetilde{\mathbf u}_i'
\mathbf M_{\widehat{\mathbf F}}
\widehat{\mathbf Z}_i
\right)
\widehat{\boldsymbol\Omega}_{iT}^{-1}
\left(
\widehat{\mathbf Z}_i'
\mathbf M_{\widehat{\mathbf F}}
\widetilde{\mathbf u}_i
\right).
\label{sargan}
\end{equation}
Equivalently,
\[
S_{iT}
=
T\,
\overline{\mathbf g}_{iT}
(\widetilde{\boldsymbol\theta}_i)'
\widehat{\boldsymbol\Omega}_{iT}^{-1}
\overline{\mathbf g}_{iT}
(\widetilde{\boldsymbol\theta}_i),
\]
where $\overline{\mathbf g}_{iT}(\boldsymbol\theta)$ is defined in
Eq.~\eqref{sample_moments}.
The limiting distribution of the statistic, allowing for serial correlation in the moment process, is given in the following theorem.\footnote{Related IV estimators for panel models with latent common components have been studied in the literature (e.g., \citet{NorkuteEtal2021}). To the best of our knowledge, however, the commodity-specific over-identification test for the defactored moment conditions considered
here has not previously been developed.  Theorem~\ref{TH3} therefore provides a formal basis for assessing instrument validity within the present divide-and-conquer framework.}

\begin{theorem}[Commodity-Specific Over-Identification Test]
\label{TH3}
Consider the model in Eqs.~\eqref{compact}--\eqref{pca_vectori}, and suppose
that Assumptions 1--5 hold together with the additional regularity conditions
reported in the Online Appendix, including Assumption~A.5 on consistency of
the long-run covariance estimator. Let $S_{iT}$ be defined as in
Eq.~\eqref{sargan}, with the sample moments evaluated at
$\widetilde{\boldsymbol\theta}_i$ defined in Eq.~\eqref{efficient_iv}.
Under the null hypothesis that the $q_i$ population moment conditions are
valid, as $N,T\to\infty$ with $N/T\to c\in(0,\infty)$, for each fixed $i$
and almost every realization of $\mathcal H$,
\[
S_{iT}
\overset{d}{\longrightarrow}
\chi^2_{q_i-p}
\qquad\text{conditionally on }\mathcal H,
\]
where $q_i$ is the fixed number of instruments for commodity $i$ and
$p=\dim(\boldsymbol\theta_i)$.
\end{theorem}

\begin{proof}
See Online~\ref{reg_conditions_lemmas}.
\end{proof}

\noindent
\textbf{Stage 2: Latent Economy-Wide Common Risk Sources}

In Stage~2, we use the commodity-specific IV residuals from Stage~1 to identify observable proxies for latent economy-wide macro risk and estimate the corresponding heterogeneous sensitivities. We begin by forming $\widehat{\mathbf{u}}_{i} = \mathbf{r}_{i} - \mathbf{C}_{i}\widehat{\boldsymbol{\theta}}_{i}$
which consistently estimates $\mathbf u_i
=
\mathbf G\boldsymbol\gamma_i+\boldsymbol\varepsilon_i$ and captures the variation in commodity returns remaining after controlling for the observed sector-return, micro-risk, and market-risk sources. We next extract the $r$ leading principal components from the empirical covariance matrix $(NT)^{-1} \sum_{i=1}^{N} \widehat{\mathbf{u}}_{i} \widehat{\mathbf{u}}_{i}^{\prime}$, where $r$ is determined using standard eigenvalue-based methods \citep{BaiNg2002,AhnHorenstein2013}. Under the maintained factor structure, these principal components recover, up to rotation, the common component of the Stage~1 residuals generated by $\mathbf G$. They therefore provide a statistical representation of the latent macro-risk space remaining after the observed risk sources have been controlled for. We then apply the Boosting with Multiple Testing (BMT) algorithm of \citet{KapetaniosEtal2026} to each extracted principal component using a high-dimensional set of candidate macro-financial variables $\boldsymbol{\mathcal{Z}} \in \mathbb{R}^{T \times n}$, where $\boldsymbol{\mathcal{Z}} = \left(\boldsymbol{\mathcal{Z}}_{1}, \dots, \boldsymbol{\mathcal{Z}}_{n}\right)$. Methodological details of BMT are provided in the Online \ref{subsec:BMT_method}. The high dimensionality arises because macro risk is latent, while observable macro-financial indicators available to proxy it are numerous, highly collinear, and provide limited guidance as to which are truly informative. The procedure selects, for each component $j=1,\dots,r$, a sparse index set $\mathcal{S}_{j} \subset {1,\dots,n}$.
The union of predictors selected across the $r$ components defines the set of observable macro-financial proxies for the latent macro-risk
space. Formally, letting $\mathcal{S} = \bigcup_{j=1}^{r} \mathcal S_{j}$, the corresponding design matrix is denoted by $\boldsymbol{\mathcal Z}_{(\mathcal{S})}$. Finally, we regress $\widehat{\mathbf{u}}_{i}$ on $\boldsymbol{\mathcal Z}_{(\mathcal{S})}$ to obtain commodity-specific sensitivities to the identified macro risk sources, $\widehat{\boldsymbol{\gamma}}_i$.\footnote{\citet{KapetaniosEtal2026} show that the post selection least squares estimator is asymptotically equivalent to the infeasible oracle estimator that would be obtained if the true model were known in advance.}

The PCA--BMT step therefore links the common variation remaining in commodity returns after Stage~1 to observable macro-financial risk sources through a
structured dimension-reduction and selection procedure. PCA recovers the common residual variation associated with latent macro risk, while BMT identifies the macro-financial indicators that are informative about this
variation. The resulting predictor set provides a sparse and interpretable set of observable proxies for the latent macro-risk space and permits estimation of heterogeneous macro-risk sensitivities across commodities. Finite-sample evidence in Online~\ref{sec:Appendix_MC} documents the performance of the procedure in high-dimensional settings and compares it with alternative regularisation methods such as Lasso.

\begin{remark}
\textnormal{The treatment of the latent macro-risk component differs from standard factor-augmented approaches, which typically retain a small number of estimated latent factors in the final specification \citep[e.g.,][]{Bai2009,NorkuteEtal2021}. Here, the extracted factors serve an intermediate role: they recover the common residual variation to which BMT is subsequently applied. The high-dimensionality arises from the candidate set of observable macro-financial indicators, not from the dimension of the underlying macro-risk process. BMT selects a sparse subset of these indicators to provide an observable representation of the latent common component.
}
\end{remark}

\begin{remark}
\textnormal{To the best of our knowledge, applying BMT to principal components extracted from the commodity-specific IV residuals, rather than directly to the residuals themselves, is novel. The distinction is important. Direct commodity-by-commodity selection could produce different sets of macro-financial indicators across commodities, making it difficult to identify a common observable representation of macro risk. Moreover, non-selection of an indicator could reflect either a zero sensitivity or insufficient power to pass the selection threshold, complicating aggregation and subsequent Mean Group inference. The PCA--BMT procedure instead performs selection on the common residual variation and then estimates heterogeneous commodity-specific sensitivities conditional on a common selected indicator set.
}
\end{remark}

Once the commodity-specific coefficients $\widehat{\boldsymbol{\theta}}_i$ and
$\widehat{\boldsymbol{\gamma}}_i$ have been obtained, they can be aggregated to conduct population-level inference while preserving heterogeneity in the underlying commodity-specific sensitivities. In particular, the Mean Group (MG) estimator of the population average $\boldsymbol{\theta}$ is
\begin{equation}
\widehat{\boldsymbol{\theta}}
=
\frac{1}{N}
\sum_{i=1}^{N}
\widehat{\boldsymbol{\theta}}_i .
\label{ivmg}
\end{equation}

The following result establishes that averaging the commodity-specific estimates consistently recovers the population-average sensitivities and permits standard inference on these averages.
\noindent
\begin{proposition}[Asymptotic Distribution of the Mean Group Estimator]
\label{prop:MG_distribution}
Consider the model in Eqs.~(\ref{compact})--(\ref{pca_vectori}), under Assumptions 1--5. Then, as $N,T\to\infty$ such that $N/T\to c$, with $0<c<\infty$,
\begin{equation}
\sqrt{N}
\left(
\widehat{\boldsymbol{\theta}}-\boldsymbol{\theta}
\right)
\overset{d}{\rightarrow}
N\left(
\mathbf{0},\boldsymbol{\Sigma}_{\eta}
\right),
\end{equation}
and
\begin{equation}
\widehat{\boldsymbol{\Sigma}}_{\eta}
-
\boldsymbol{\Sigma}_{\eta}
\overset{p}{\rightarrow}
\mathbf{0},
\end{equation}
where
\begin{equation}
\widehat{\boldsymbol{\Sigma}}_{\eta}
=
\frac{1}{N-1}
\sum_{i=1}^{N}
\left(
\widehat{\boldsymbol{\theta}}_i-\widehat{\boldsymbol{\theta}}
\right)
\left(
\widehat{\boldsymbol{\theta}}_i-\widehat{\boldsymbol{\theta}}
\right)' .
\label{MG_varcov}
\end{equation}
\end{proposition}

\begin{proof}
See Online~\ref{reg_conditions_lemmas}.
\end{proof}

\noindent
We also summarize the heterogeneous macro-risk sensitivities using the Mean Group estimator
\begin{equation}
\widehat{\boldsymbol{\gamma}}
=
\frac{1}{N}
\sum_{i=1}^{N}
\widehat{\boldsymbol{\gamma}}_i,
\end{equation}
with its variance-covariance matrix estimated from the cross-sectional dispersion of the commodity-specific estimates, analogously to Eq.~\eqref{MG_varcov}. The post-selection oracle property established by \citet{KapetaniosEtal2026} provides the basis for  inference following the BMT selection step.

\begin{remark}
\textnormal{
In the empirical implementation, we also consider a Median Group estimator as a robustness check. This replaces the cross-sectional mean with the cross-sectional median of the commodity-specific coefficients and is less sensitive to outliers and heavy-tailed heterogeneity. Results are reported in Online~\ref{sec:Appendix_Robustness}.
}
\end{remark}

\section{Data \label{sec:data}}

Our sample consists of $N=52$ exchange-traded commodity futures contracts obtained from LSEG Datastream (Eikon). For each commodity we retrieve the daily settlement price of the nearest-maturity  series (Datastream mnemonic suffix \texttt{TRc1}) together with the corresponding open interest, defined as the total number of contracts entered into and not yet liquidated, that is, the aggregate purchase or sale commitment outstanding at the close of trading.

The cross-section is deliberately broad along three dimensions. First, it spans the three conventional sectors of the commodity complex: agriculture and livestock (24 contracts), energy and petrochemicals (14 contracts), and metals (14 contracts). Second, it spans venues in North America (CBOT, CME, CSCE, KCBT, NYMEX, COMEX, NYCE), Europe (ICE, LIFFE, LME), and Asia (DCE, ZCE, MCX, NCDEX, KLSE, SICOM, GME). Third, and most important for our purposes, it deliberately mixes contracts that are constituents of the major investable commodity indices, such as WTI and Brent crude, gold, copper, and the LME base metals, with contracts that lie outside the index universe and are traded predominantly by local commercial participants, such as cardamom, kapas, mentha oil, barley and maize on the Indian exchanges, eggs on the Dalian Commodity Exchange, and the Chinese ferroalloys. This design is central to the identification of heterogeneous exposures: the financialisation literature predicts that index membership, rather than physical characteristics alone, governs the strength of an asset's comovement with the commodity market as a whole \citep{TangXiong2012,BasakPavlova2016}, and a sample restricted to index constituents cannot test that prediction. Table~\ref{tab:commodity_list} lists the full cross-section by sector, settlement currency, and Datastream series.

Our classification assigns the petrochemical contracts, namely linear low-density polyethylene, polypropylene, polyvinyl chloride, and methanol, to the energy sector, since their production costs and price dynamics are driven principally by crude oil, naphtha, and coal feedstocks rather than by agricultural or metallurgical fundamentals. We assign the ferroalloys (ferrosilicon, silicon manganese) and the steel complex (iron ore, hot-rolled coil steel) to metals. Livestock and dairy contracts are grouped with agriculture.

The sample runs at weekly frequency from August 2014 to April 2026, yielding a balanced panel of $T=611$ weekly observations per commodity and $NT=31,772$ observations in total. The start date is dictated by the availability of continuous price and open-interest histories for the Asian contracts, which are the most recently listed in the cross-section.

Thirteen of the fifty-two contracts settle in currencies other than the US dollar: the MCX and NCDEX contracts in Indian rupees, the DCE and ZCE contracts in Chinese renminbi, and the KLSE palm oil contract in Malaysian ringgit. We convert all settlement prices to US dollars at the contemporaneous spot exchange rate before computing returns, realised moments, and market aggregates, so that the panel is expressed throughout in a single numeraire and the cross-section is directly comparable. This is the convention of the commodity asset-pricing literature \citep{GortonRouwenhorst2006,TangXiong2012}, and it has the natural interpretation that returns are those realised by an unhedged dollar-based investor.

{color{red} Based on daily and weekly data, returns and higher moments of both individual contracts and market measures are computed, as detailed in Online Appendix A. }

Stage 2 of our sequential approach involves projecting dominant principal components onto a high-dimensional set of candidate macro-financial risk indicators. To this end, we collect 36 variables that serve as proxies for a range of  risk channels, including uncertainty, sentiment, market attention, and global financial conditions. These variables have been carefully assembled from a wide array of influential papers, and to the best of our knowledge this is the first study to employ such a diverse set of indicators within a single framework.
The full set of these variables and their definitions are listed in Table~\ref{tab:meta-variables}.
This list is by no means exhaustive. Additional indices not considered here include the global financial uncertainty index by \cite{Caggiano2023} (only available up to 2020);  and the global financial cycle index by \cite{Miranda-Agrippino2020} (only recently extended to include recent years).


\section{Results \label{sec:results}}

\subsection{Micro, Market, and Macro Risk Sensitivities \label{sec:risk_sensitivities}}

Table~\ref{tab:MG_sensitivities} reports Mean Group estimates of the commodity return sensitivities to the micro, market, and macro risk sources identified by our two-stage procedure. Results are shown for the full sample and separately for agriculture, energy, and metals.

At the micro level, in the full sample, commodity-specific volatility is negatively associated with returns (\(-0.047\)). This negative contemporaneous association is consistent with related evidence linking volatility and commodity returns \citep{FuertesMiffreFernandezPerez2015} and, more broadly, with the negative contemporaneous relationship documented between unexpected stock returns and unexpected changes in market volatility \citep{FrenchSchwertStambaugh1987}. This relationship is strongest for energy (\(-0.110\)), weaker for agriculture (\(-0.052\)), and absent for metals. Commodity-specific skewness is positively and highly significantly associated with returns in all three sectors, with remarkably similar coefficients ranging from \(0.019\) to \(0.025\). The positive coefficient indicates that, conditional on market-level distributional characteristics, weeks characterised by an increase in the right-skewness of an individual commodity’s return distribution are associated with higher contemporaneous returns.\footnote{This contemporaneous sensitivity should not be interpreted as evidence on a skewness risk premium. Existing asset-pricing studies, such as \citet{FernandezPerezFrijnsFuertesMiffre2018}, examine whether skewness predicts subsequent futures returns. The two objects address different economic questions and need not have the same sign.} Thus, while asymmetry in the commodity's own return distribution is systematically related to contemporaneous returns, the corresponding average sensitivity to own tail thickness is considerably weaker. Open interest has no statistically significant Mean Group coefficient in any sector. This is an average cross-commodity result and does not imply that open interest is irrelevant for every individual commodity.

Market conditions have a strong association with individual commodity returns. The leave-one-out market return enters positively and significantly in every sector, with a full-sample sensitivity of \(0.509\). The association is particularly strong for energy, where the coefficient exceeds one (\(1.038\)), compared with \(0.489\) for metals and \(0.213\) for agriculture. This strong comovement is consistent with evidence that commodity returns contain an important common component \citep{ChristoffersenLundeOlesen2019,DieboldLiuYilmaz2018}. The sectoral estimates further show that this common market component is especially important for energy returns.

Market volatility is significantly positively associated with returns in the energy sector (\(0.198\)), but significantly negatively associated with returns in the metals sector (\(-0.067\)). The difference in signs is consistent with the distinct exposure of these sectors to commodity-market stress. Energy markets may be particularly sensitive to episodes in which supply disruptions simultaneously increase uncertainty and prices, whereas metals may be more closely exposed to the deterioration in global demand conditions that often accompanies broader periods of market stress. This interpretation is also consistent with the commodity-market connectedness evidence of \citet{DieboldLiuYilmaz2018}, who identify energy as an important transmitter of volatility shocks within commodity markets. By contrast, the corresponding coefficient for agriculture is small and statistically insignificant. At the aggregate level, the opposing sensitivities of energy and metals may therefore contribute to the lack of statistical significance of the pooled market-volatility coefficient. The opposing signs can also be interpreted in the context of the layered specification. Conditional on market volatility, higher commodity-specific volatility represents an increase in risk that is relatively concentrated in the individual commodity. Conversely, conditional on a commodity’s own volatility, higher market volatility captures a change in broader risk conditions affecting other commodities in the relevant market. Thus, although the two volatility measures may be related, they capture different dimensions of volatility exposure once each is conditioned on the other. Consequently, there is no theoretical requirement for them to exhibit the same association with contemporaneous returns.

This distinction between commodity-specific and market-level risk is even more systematic when considering higher moments. Market skewness enters negatively and highly significantly in all three sectors, in contrast to the positive coefficient on commodity-specific skewness. Conditional on market skewness, an increase in a commodity’s own skewness captures a change in the asymmetry of its return distribution that is specific to that commodity. Conversely, conditional on its own skewness, an increase in market skewness reflects a change in the asymmetry of return distributions elsewhere in the sector. The opposing coefficients therefore indicate that the association between distributional asymmetry and returns depends on the level at which the change in skewness occurs. In other words, the economic relevance of higher-moment variation is not determined solely by whether returns become more or less skewed, but also by whether the change in skewness is commodity-specific or occurs at the broader market level. A similar, although less pronounced, distinction emerges for kurtosis: commodity-specific kurtosis is weakly negative, whereas market kurtosis is positive and highly significant across all three sectors. Taken together, these results provide empirical evidence that the level at which higher-moment variation occurs matters for its association with returns, supporting a distinction between commodity-specific and market-level higher-moment risk rather than treating skewness and kurtosis as undifferentiated measures of risk.

Before turning to the latent macro component, we assess the empirical validity of the Stage~1 identification strategy using standard IV diagnostics. Based on the stacked specification across commodities, the Kleibergen--Paap test rejects the null of underidentification ($LM=50.602$, $p=0.004$), while the Kleibergen--Paap rk Wald $F$-statistic is 276.612, providing no indication of weak identification. The panel over-identification test does not reject the null of valid moment conditions ($J_{NT}=28.610$, $p=0.329$).\footnote{The test builds directly on the commodity-specific over-identification test developed in Theorem 3: letting $p_{iT}$ denote the p-value of the test for commodity $i$, the individual p-values are combined using the Fisher statistic ($J_{NT}=-2\sum_{i=1}^{N}\ln p_{iT}$), with inference obtained using a cross-sectional bootstrap to accommodate dependence across commodities.} This result is particularly relevant in our setting because instrument validity relies on the spanning condition in Assumption~2: if the extracted common factor space failed to absorb the latent macro component, the defactored moment conditions would generally be violated.

PCA identifies a single dominant component in the first-stage residuals, accounting for 91.4\% of total residual variation. BMT selects six observable macro-financial indicators: the broad U.S. dollar index (DTWEXBGS), Treasury-market volatility (MOVE), consumer sentiment (UMCSENT), five-year expected inflation (EXPINF5YR), equity-market volatility (VIX), and oil-related geopolitical risk. The bottom panel of Table~\ref{tab:MG_sensitivities} reports Mean Group estimates of the corresponding sensitivities. All six indicators exhibit statistically significant average sensitivities, with consistent signs across the agriculture, energy, and metals sectors. This degree of cross-sector consistency is particularly noteworthy given the substantially greater heterogeneity observed in some of the market-level sensitivities.

The positive sensitivity to the dollar index warrants a conditional interpretation. The conventional negative relationship between the U.S. dollar and commodity prices is generally associated with pricing and invoicing effects, as well as global demand channels, that operate through broad movements in commodity markets. In our Stage 2 specification, however, these market-level effects have already been removed. Consequently, the estimated dollar coefficient should not be interpreted as capturing the conventional unconditional dollar–commodity relationship. Rather, it measures the association between movements in the dollar and the residual common component of commodity returns, conditional on—and after accounting for—the observed market-level effects. Under this interpretation, the dollar may also reflect its role as a global financial risk factor, operating through financial conditions and the risk-taking capacity of market participants. This interpretation is consistent with the evidence in \citet{AvdjievEtal2019}, who emphasize the broader role of the U.S. dollar in global financial conditions.

This interpretation is reinforced by a robustness exercise that addresses a potential measurement concern. Because all settlement prices are converted to U.S. dollars prior to the computation of returns, one might suspect that the estimated dollar exposure reflects
currency translation rather than an economic channel, since for contracts denominated in foreign currency the dollar return embeds the movement of the local currency against the dollar. We therefore re-estimate Stage 2 separately for the thirty-nine dollar-denominated
contracts and for the thirteen contracts settling in Chinese renminbi, Indian rupees, or Malaysian ringgit. The dollar exposure remains positive and highly significant in both subsamples (see Table \ref{tab:robustness_dollar}, at $0.0026$ ($z=10.50$) and $0.0033$ ($z=11.38$) respectively. The former
estimate is decisive for the concern at hand: for contracts that already settle in dollars no translation term can arise by construction, so the positive loading cannot be a measurement artefact. The modest difference between the two estimates is not specific to the dollar. All six macro coefficients, together with the constant, are larger in absolute value in the foreign-currency subsample by a nearly uniform factor of approximately $1.29$
(standard deviation $0.03$ across the seven parameters), and the dollar coefficient relative to the VIX coefficient is virtually identical in the two groups ($1.07$ and $1.04$). This pattern is consistent with a difference in the scale of the residual component rather than with a currency-specific channel, as these locally traded contracts exhibit weak market-level exposures in Stage 1 and therefore retain more residual variation.

Financial-market uncertainty provides a particularly informative interpretation of the latent macro-financial component. Both the VIX and MOVE enter with negative and statistically significant sensitivities, indicating that heightened uncertainty in equity and Treasury markets is associated with lower commodity returns, even after conditioning on the observed commodity-market component. This pattern is consistent with the latent factor captures a dimension of financial stress that is not fully reflected in broad commodity-market movements. The result is consistent with the broader literature documenting stronger linkages between commodity and financial markets during periods of heightened uncertainty. \citet{SilvennoinenThorp2013}, for example, show that increases in the VIX are associated with higher commodity-return volatility and stronger commodity–equity correlations, particularly for a substantial fraction of commodity–equity pairs. \citet{ChengKirilenkoXiong2015} provide complementary evidence from investor behaviour, documenting reductions in financial traders’ net-long positions in commodity futures in response to heightened market distress. \citet{AdamsGluck2015} likewise document persistent comovement between commodity and equity markets in the post-financialization period. Finally, \citet{ChristoffersenLundeOlesen2019} identify a pronounced common factor in commodity futures volatility that is related to stock-market volatility and the business cycle. Our estimates complement this evidence by showing that the financial-stress channel extends to Treasury-market volatility and remains present after the observed commodity-market component has been removed. The negative sensitivities to both VIX and MOVE therefore suggest that the latent factor captures a residual dimension of financial conditions that operates across commodity sectors beyond the variation explained by the common observed market layer.

Consumer sentiment also exhibits a negative and statistically significant sensitivity across all three commodity sectors. This coefficient warrants a similarly conditional interpretation. Because broad commodity-market movements have already been accounted for in Stage 1, the Stage 2 estimate should not be interpreted as the unconditional effect of changes in consumer sentiment on commodity returns. Rather, it captures the association between consumer sentiment and the residual common component of commodity returns, conditional on the variation already explained by the observed commodity-market component. Economically, this residual relationship may reflect changes in risk appetite, expectations, or broader perceptions of economic conditions that are not fully captured by movements in the commodity market as a whole.

Expected inflation provides a distinct perspective on commodity price dynamics. The literature on the inflation-hedging properties of commodities has traditionally focused on the relationship between commodity returns and realised or unexpected inflation, with evidence suggesting that the strength of this relationship varies substantially across commodity classes and over time. By contrast, EXPINF5YR is a forward-looking measure of medium-term inflation expectations and therefore captures a different dimension of the inflation environment. Its coefficient captures the association between changes in medium-term inflation expectations and the residual common component of commodity returns, conditional on the observed commodity-market layer. One possible interpretation is that higher inflation expectations convey information about anticipated monetary and financial conditions. To the extent that stronger inflation expectations are accompanied by expectations of tighter monetary policy and higher real interest rates, they may exert downward pressure on commodity prices through several channels, including the opportunity cost of holding inventories, expectations of future commodity demand, and portfolio reallocation. This interpretation is consistent with the literature emphasizing the role of real interest rates and monetary policy in commodity-price determination. Frankel (2008), for example, highlights the role of real interest rates in determining real commodity prices, while more recent evidence documents cost-of-carry, expected-demand, and financial-market channels through which monetary policy affects commodity prices.

Finally, oil-related geopolitical risk enters positively in all sectors, with the largest coefficient for energy. This is consistent with supply-disruption and precautionary-demand channels through which geopolitical tensions raise commodity prices, particularly for energy-intensive markets; see, for example, \citet{CaldaraIacoviello2022}. The larger energy sensitivity is economically intuitive given the direct exposure of oil markets to geopolitical disruptions.

Taken together, the Stage~2 estimates point to a common macro-financial component that is broad in scope but remarkably consistent across sectors. The signs should be interpreted jointly and conditionally, because the selected indicators form an observable representation of residual common variation after the micro and market components have been removed.

\begin{table}[!htbp]
  \centering
  \caption{Mean Group Estimates of Micro, Market, and Macro Risk Exposures}
  \small
  \setlength{\tabcolsep}{5pt}
  \renewcommand{\arraystretch}{0.95}
  \begin{tabular}{lcccc}
    \toprule
      & (1) & (2) & (3) & (4) \\
      & Full Sample & Agriculture & Energy & Metals \\
    \midrule

    \multicolumn{5}{l}{\textit{Stage 1: Micro and Market Risk}} \\[0.3em]

VLTY
    & -0.047$^{**}$  & -0.052$^{*}$   & -0.110$^{***}$ & 0.024 \\
    & (0.022)        & (0.031)         & (0.036)         & (0.044) \\

SKEW
    & 0.021$^{***}$  & 0.021$^{***}$  & 0.025$^{***}$  & 0.019$^{***}$ \\
    & (0.001)        & (0.001)         & (0.003)         & (0.001) \\

KTS
    & -0.000          & -0.001$^{**}$  & -0.000          & -0.000 \\
    & (0.000)         & (0.000)         & (0.001)         & (0.000) \\

OI
    & 0.000           & 0.000           & 0.000           & 0.000 \\
    & (0.000)         & (0.000)         & (0.001)         & (0.000) \\

$r_{-i}$
    & 0.509$^{***}$   & 0.213$^{***}$  & 1.038$^{***}$   & 0.489$^{***}$ \\
    & (0.071)         & (0.040)         & (0.163)         & (0.094) \\

CVLTY
    & 0.047           & 0.027           & 0.198$^{***}$   & -0.067$^{**}$ \\
    & (0.030)         & (0.046)         & (0.050)         & (0.034) \\

CSKEW
    & -0.073$^{***}$  & -0.079$^{***}$ & -0.076$^{***}$  & -0.060$^{***}$ \\
    & (0.005)         & (0.008)         & (0.011)         & (0.005) \\

CKTS
    & 0.008$^{***}$   & 0.008$^{***}$  & 0.008$^{***}$   & 0.006$^{***}$ \\
    & (0.001)         & (0.001)         & (0.001)         & (0.001) \\

COI
    & 0.002$^{***}$   & 0.003$^{***}$  & 0.002$^{***}$   & 0.000 \\
    & (0.001)         & (0.001)         & (0.001)         & (0.001) \\
    
    \midrule
    \multicolumn{5}{l}{\textit{Stage 2: Economy-Wide Risk}} \\[0.3em]

DTWEXBGS
    & 0.0028$^{***}$ & 0.0032$^{***}$ & 0.0027$^{***}$ & 0.0023$^{***}$ \\
    & (0.0002)        & (0.0003)        & (0.0004)        & (0.0002) \\

MOVE
    & -0.0006$^{***}$ & -0.0007$^{***}$ & -0.0006$^{***}$ & -0.0005$^{***}$ \\
    & (0.0000)         & (0.0001)         & (0.0001)         & (0.0000) \\

UMCSENT
    & -0.0009$^{***}$ & -0.0010$^{***}$ & -0.0010$^{***}$ & -0.0008$^{***}$ \\
    & (0.0001)         & (0.0001)         & (0.0001)         & (0.0001) \\

EXPINF5YR
    & -0.0372$^{***}$ & -0.0395$^{***}$ & -0.0401$^{***}$ & -0.0303$^{***}$ \\
    & (0.0026)         & (0.0041)         & (0.0059)         & (0.0031) \\

VIX
    & -0.0026$^{***}$ & -0.0028$^{***}$ & -0.0029$^{***}$ & -0.0021$^{***}$ \\
    & (0.0002)         & (0.0003)         & (0.0003)         & (0.0002) \\

GPR Oil
    & 0.00003$^{***}$ & 0.00003$^{***}$ & 0.00004$^{***}$ & 0.00002$^{***}$ \\
    & (0.00000)        & (0.00000)        & (0.00000)        & (0.00000) \\

    \midrule
Number of commodities
    & 52 & 24 & 14 & 14 \\

Observations
    & 31,616 & 14,592 & 8,512 & 8,512 \\

    \bottomrule
   \end{tabular}

  \vspace{0.4em}
  \begin{minipage}{0.98\textwidth}
    \scriptsize
    \textit{\textbf{Notes}:}
    The table reports Mean Group estimates for the full sample and separately
    for agricultural, energy, and metal commodities. Stage~1 reports
    heterogeneous exposures to commodity-specific (micro) and observed
    market-level risk sources. Stage~2 reports heterogeneous exposures to the
    macro-financial variables selected by the PCA--BMT procedure from the
    commodity-specific IV residuals. Standard errors are reported in
    parentheses. Significance levels are
    $^{*}$ 10\%, $^{**}$ 5\%, and $^{***}$ 1\%.
  \end{minipage}

  \label{tab:MG_sensitivities}
\end{table}

\subsection{Risk Intensity Indices \label{sec:risk_indices}}

The coefficient estimates establish economically distinct sensitivities to micro, market, and macro risk, but they do not by themselves reveal the realised importance of each risk layer. A large sensitivity may be associated with a risk condition that is close to its typical level, whereas a more moderate sensitivity may be consequential when the corresponding risk condition is unusually pronounced. We therefore turn to the Risk Intensity Indices (RIIs), which combine the estimated heterogeneous sensitivities with prevailing risk conditions and place the three risk layers on a common scale. The RIIs should be interpreted as measures of realised risk intensity conditional on the estimated sensitivities and observed risk conditions, rather than as structural variance decompositions or causal attributions of commodity returns. Their shares therefore describe the composition of the RII itself.

Figure~\ref{fig:component_rii} plots the aggregate Micro, Market, and Macro RIIs over time. The three indices display markedly different patterns of time variation. Most strikingly, the Market RII exhibits pronounced short-lived spikes around episodes of severe commodity-market disruption. The first major episode occurs in April 2020, during the COVID-19 shock and the extraordinary dislocation in oil markets. In particular, the sharp increase around 22 April closely follows the collapse of the May WTI futures contract into negative territory on 20 April, amid the severe demand and storage pressures prevailing at the time. The Market RII rises again sharply in late 2021, with the spike around 1 December coinciding with the emergence of the Omicron variant and a period of already elevated stress in global energy markets. The largest movements around 2022 are also consistent with the index capturing periods of broad commodity-market stress. The exceptionally large spike in early March occurs immediately following Russia’s invasion of Ukraine, which generated substantial uncertainty and supply concerns across energy, agricultural commodities, and industrial metals. Its timing also coincides closely with the extraordinary disruption in the nickel market, culminating in the suspension of nickel trading by the London Metal Exchange on 8 March. A further increase around November 2022 occurred amid aggressive global monetary tightening, a strong U.S. dollar, weakening global growth expectations, and uncertainty surrounding Chinese commodity demand, providing a different macro-financial context for elevated commodity-market risk. More recently, the pronounced spike in April 2025 coincides with the sharp escalation in U.S. tariff and trade-policy uncertainty and the associated adjustment in oil and industrial-metal markets. The renewed elevation toward the end of the sample appears broader in nature and is therefore more appropriately interpreted as a period of heightened commodity-market turbulence rather than attributed to a single event.

These episodes are informative about the economic content of the Market RII. Its largest movements coincide with shocks that affected several segments of the commodity complex, including energy, metals, and agricultural markets. The resulting pattern is consistent with the Market RII assigning greater intensity to periods in which market-wide commodity risk conditions become unusually pronounced. The Micro RII also increases during some of these episodes, but its movements are generally less extreme, consistent with a more stable contribution from commodity-specific risk alongside occasional episodes of heightened idiosyncratic stress. The differences across the indices should nevertheless be interpreted as differences in the realised intensity of the estimated risk layers rather than as direct evidence that any particular event caused the movements in a given index.

The Macro RII displays a distinctly different temporal profile. Rather than the sharp spikes characterising the Market RII, it evolves more smoothly over the sample. It is relatively subdued around 2019–20, increases noticeably through 2021, and remains at a higher level over much of the subsequent period. This pattern is consistent with the macro-financial component varying more gradually than the market layer over the sample. Because the Macro RII is constructed from the selected macro-financial representation of the residual common component, however, its interpretation is conditional on the information set used to calculate that component. Micro risk lies between the market and macro patterns, displaying a relatively stable underlying intensity punctuated by occasional increases. Taken together, the three indices indicate that commodity risk varies not only in magnitude but also in the relative intensity of its micro, market, and macro-financial components.

\begin{figure}[htbp]
  \centering
  \includegraphics[width=0.90\textwidth]{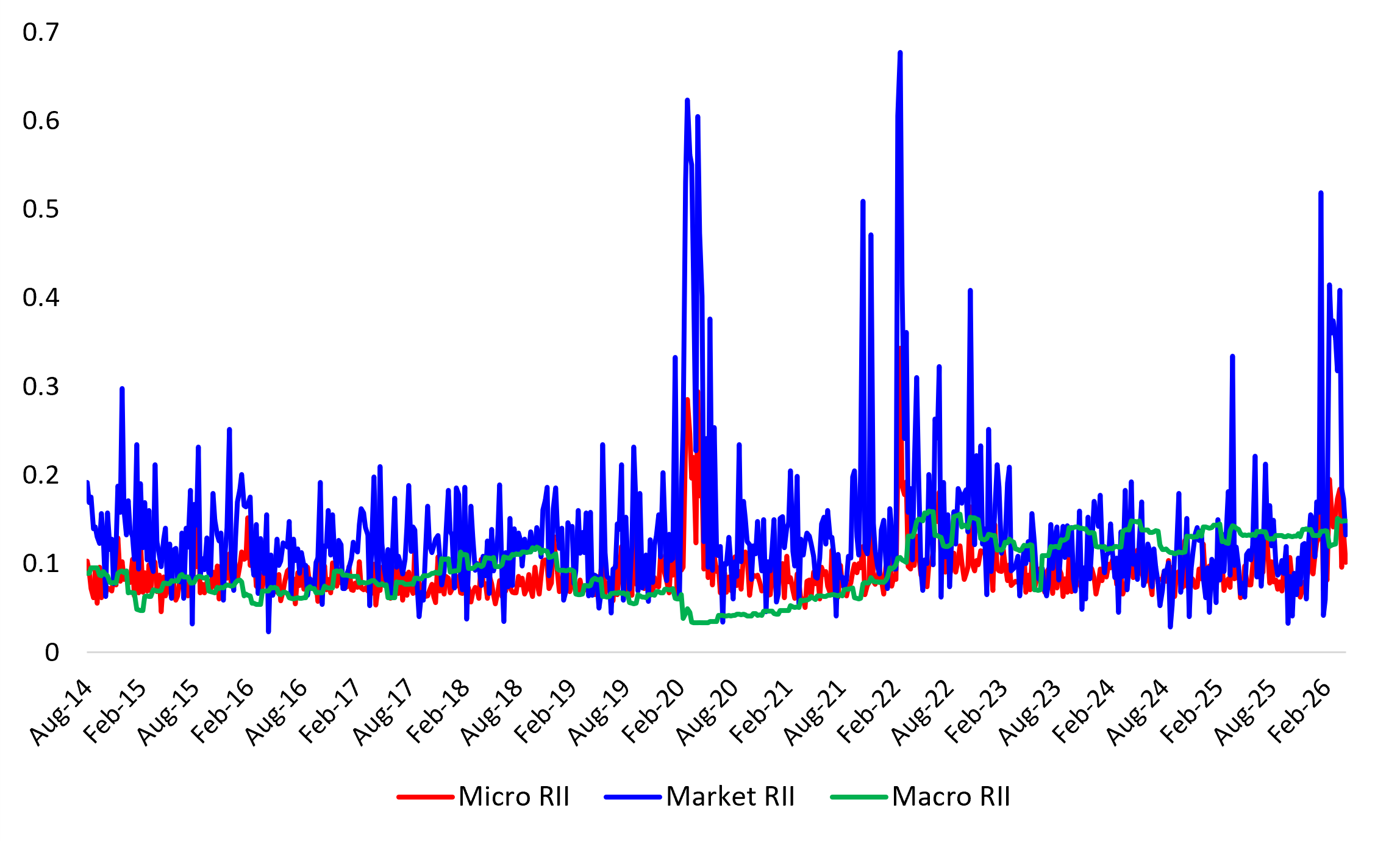}
  \caption{Total Risk Intensity Index (RII)}
  \label{fig:component_rii}
\end{figure}
\noindent

The changing composition of risk becomes clearer when the
three RIIs are expressed as shares of total RII, as in
Figure~2. Around the COVID-19 and oil-market
collapse,\footnote{This is consistent with evidence of
heightened connectedness and risk transmission across
commodity markets during the COVID-19 pandemic
\citep{FaridEtAl2022,QiaoHan2023}.} and again during the commodity-market disruptions surrounding Russia's invasion of Ukraine, the market layer accounts for a very large share of total risk intensity. These episodes therefore involve not only elevated commodity risk but also a pronounced shift in its composition toward market risk. Outside such episodes, the market share declines, while the relative contributions
of the micro and, particularly, macro RIIs become larger. The macro share becomes more prominent after 2021 and remains substantial throughout much of the latter part of the sample, whereas the micro share is comparatively stable. The composition of commodity risk intensity is therefore strongly state dependent. Importantly, these shifts occur although the estimated sensitivities are time-invariant, illustrating why sensitivities alone are insufficient to assess the relative importance of the three risk layers under prevailing market conditions.

\begin{figure}[htbp]
  \centering
  \includegraphics[width=0.90\textwidth]{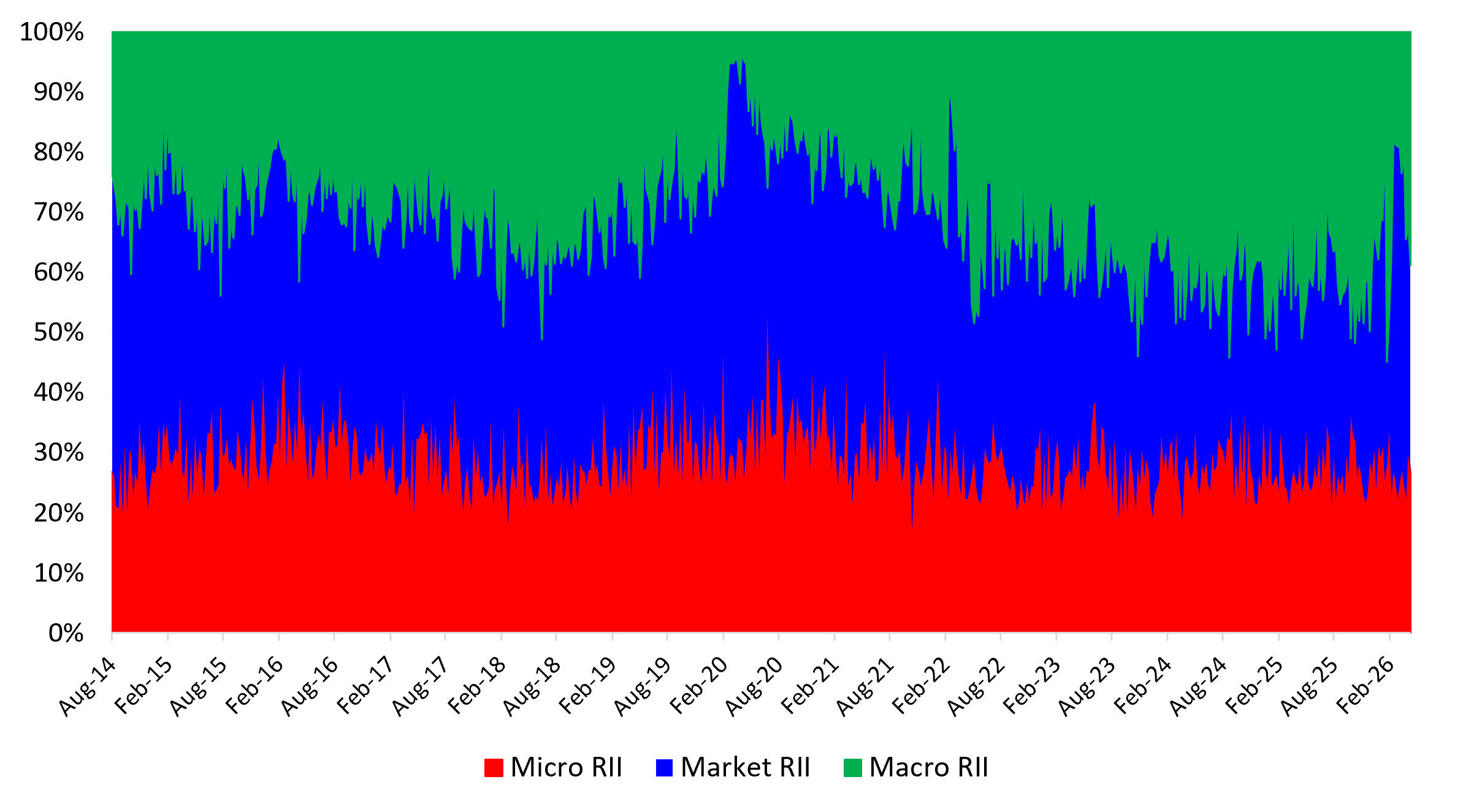}
  \caption{Share of Total RII(\%)}
  \label{fig:RII_shares}
\end{figure}
\noindent

Average risk composition also differs substantially across sectors. For the full sample, market risk accounts for 42.1\% of total RII, while macro and micro risk account for 29.7\% and 28.2\%, respectively. Thus, although the market layer is the largest component on average, the remaining RII is almost evenly divided between the macro and micro layers. These percentages should be interpreted as shares of total RII, rather than as percentages of return variance or as structural causal contributions.

The sectoral estimates reveal important differences behind these aggregate shares.
For energy, market risk accounts for 52.4\% of total risk intensity, compared with 26.2\% for micro risk and 21.4\% for macro risk. Energy is the only sector in which a single source accounts for more than half of total intensity. This strong market orientation is consistent with the pronounced market sensitivities documented in the previous subsection and reinforces the distinctive role of energy within commodity markets. Agriculture has a more balanced risk composition. Macro risk accounts for 36.1\% of total intensity, closely followed by market risk at 35.4\%, with micro risk accounting for 28.5\%. Agriculture is the only sector in which macro risk is the largest source on average, although its share is only marginally higher than that of market risk. For metals, market risk is again the largest source at 38.1\%, followed by micro risk at 33.0\% and macro risk at 28.9\%. The relatively narrow range across the three shares contrasts with the strong concentration of energy risk in the market layer. Sectoral differences therefore extend beyond the magnitude of risk intensity to its composition: energy is strongly market-oriented, agriculture assigns nearly equal importance to market and macro risk, while metals show the most even distribution across the three layers.

The level of total risk intensity also differs across sectors. The aggregate Total RII is 4.990 for energy, compared with 3.617 for agriculture and 2.740 for metals. These differences reflect the interaction between the prevailing risk conditions and the heterogeneous sensitivities estimated for individual commodities. In particular, a sector can exhibit a relatively high total RII because its commodities have stronger estimated sensitivities to one or more risk layers, because the associated risk conditions are more pronounced, or because both occur simultaneously. The RII framework therefore provides information that cannot be obtained from the coefficient estimates alone.

\begin{table}[htbp]
\centering
\caption{Risk Intensity Across Asset Groups: Micro, Market, Macro and Total Decomposition}
\label{tab:rii_groups}
\begin{tabular}{lccccccc}
\toprule
 & Full Sample & Green & Non-Green & Stable & Non-Stable & DeFi & Non-DeFi \\
\midrule
\multicolumn{8}{c}{\textbf{Micro Risk Intensity Index (RII)}} \\
 & 1.229 & 1.107 & 1.302 & 0.074 & 1.433 & 1.097 & 1.317 \\
 & 16.8\% & 13.8\% & 18.8\% & 2.9\% & 17.6\% & 16.0\% & 17.2\% \\
\addlinespace
\multicolumn{8}{c}{\textbf{Market Risk Intensity Index (RII)}} \\
 & 3.754 & 3.888 & 3.674 & 0.669 & 4.298 & 3.998 & 3.591 \\
 & 51.2\% & 48.6\% & 53.1\% & 26.1\% & 52.6\% & 58.5\% & 46.9\% \\
\addlinespace
\multicolumn{8}{c}{\textbf{Macro Risk Intensity Index (RII)}} \\
 & 2.343 & 3.002 & 1.947 & 1.819 & 2.435 & 1.743 & 2.742 \\
 & 32.0\% & 37.5\% & 28.1\% & 71.0\% & 29.8\% & 25.5\% & 35.8\% \\
\addlinespace
\multicolumn{8}{c}{\textbf{Total Risk Intensity Index (RII)}} \\
 & 7.326 & 7.997 & 6.923 & 2.562 & 8.167 & 6.838 & 7.651 \\
 & 100.0\% & 100.0\% & 100.0\% & 100.0\% & 100.0\% & 100.0\% & 100.0\% \\
\bottomrule
\end{tabular}
\end{table}

Cross-sectional concentration provides a further dimension of the RII results. Table~\ref{tab:rii_concentration} shows that risk intensity is concentrated among a relatively small share of commodities, particularly in the micro and market layers. For micro risk, the top 10\% of commodities account for 30.3\% of aggregate intensity and the top 20\% for almost half (48.8\%). At the other end of the distribution, the lowest 20\% account for only 4.8\% and the lowest 10\% for 2.2\%. Micro risk intensity is therefore highly concentrated across commodities.
Market risk exhibits a similar degree of concentration. The top 10\% of commodities account for 33.4\% of aggregate Market RII (the largest top-decile share among the three layers) and the top 20\% account for 45.9\%, compared with only 7.6\% for the lowest 20\%. This finding is particularly informative because the underlying market variables are common to commodities within each sector, yet their contribution to risk intensity differs substantially across commodities. Exposure to a common market environment therefore generates highly heterogeneous risk intensity through differences in commodity-specific sensitivities.

Macro risk is less concentrated, although sizeable differences remain across the distribution. The top 10\% and 20\% of commodities account for 26.6\% and 40.6\% of aggregate Macro RII, respectively, compared with 7.6\% for the lowest 20\%. Total risk intensity is less concentrated still: the top 20\% account for 40.1\% of total intensity, whereas the lowest 20\% account for 8.9\%. The concentration of macro risk is particularly noteworthy given the relatively similar Mean Group macro sensitivities across sectors reported in Table~\ref{tab:MG_sensitivities}. Similar sector-average sensitivities do not imply similar risk intensity across individual commodities, because sector averages can conceal substantial heterogeneity in the underlying commodity-specific sensitivities. When these heterogeneous sensitivities interact with common macro-financial conditions, they generate substantial cross-sectional differences in Macro RII. Common macro-financial risk can therefore be broad in scope while contributing differently to the realised risk intensity of individual commodities. This distinction illustrates why the RII measures add information beyond average sensitivity estimates: sensitivities describe exposure, whereas the RIIs quantify the resulting intensity conditional on the prevailing risk conditions.

Taken together, the RII results provide a decomposition of the composition and magnitude of realised commodity risk intensity across risk layers, sectors, individual commodities, and time. Market risk accounts for the largest share of aggregate RII and becomes particularly prominent during major commodity-market disruptions, while the relative contributions of the micro and macro layers are more stable outside such episodes. At the same time, substantial cross-sectional concentration remains even within a common risk environment, reflecting heterogeneous commodity-specific sensitivities. These findings support the interpretation of the three RIIs as complementary measures of distinct risk layers rather than interchangeable components of a single aggregate risk measure. The predictive implications of these measures are examined separately in Section 5.3.

\begin{table}[htbp]
\centering
\caption{Risk Intensity Concentration Across Upper and Lower Percentiles}
\label{tab:rii_concentration}
\begin{tabular}{lcccc}
\toprule
 & Top 10\% & Top 20\% & Lowest 20\% & Lowest 10\% \\
\midrule
\textbf{Micro RII}  & 29.1\% & 50.1\% & 98.4\% & 99.8\% \\
\textbf{Market RII} & 20.5\% & 37.9\% & 97.9\% & 99.9\% \\
\textbf{Macro RII}  & 30.4\% & 49.1\% & 98.1\% & 99.7\% \\
\textbf{Total RII}  & 20.6\% & 36.6\% & 97.1\% & 99.8\% \\
\bottomrule
\end{tabular}
\end{table}

\subsection{Predictive Content of the Risk Intensity Indices}
\label{sec:RII_prediction}

We next ask a distinct question from the contemporaneous decomposition in Section~5.2: whether the RIIs contain information about subsequent commodity risk. The purpose of this exercise is not to provide an alternative decomposition of the relative importance of the three risk layers, but to examine whether the risk conditions summarised by the RIIs have forward-looking content. We consider two measures of subsequent risk, volatility and absolute returns, at horizons of one week ($h=1$), four weeks ($h=4$), and twelve weeks ($h=12$). For \(h=1\), the dependent variable is the risk measure in the following week. For \(h=4\) and \(h=12\), we use the average risk measure over non-overlapping subsequent four- and twelve-week intervals, respectively, thereby avoiding mechanically induced serial correlation from overlapping observations. Each specification controls for the corresponding contemporaneous risk measure to account for persistence, includes commodity fixed effects, and uses standardized RIIs, so that the coefficients measure the association between a one-standard-deviation increase in current risk intensity and subsequent risk.

Table~\ref{tab:RII_predictive} reports the in-sample predictive regressions. The main result is that both the Micro and Market RIIs contain information about subsequent commodity risk beyond that contained in its current level. Their coefficients are positive and statistically significant for both volatility and absolute returns at all three horizons. By contrast, the coefficient on the Macro RII is small and statistically insignificant throughout. Thus, risk intensity originating at the commodity-specific and market levels has forward-looking content, whereas the economy-wide component, despite accounting for a substantial share of contemporaneous total risk intensity in Section~5.2, does not provide incremental information about subsequent volatility or absolute returns.

The relative magnitudes of the Micro and Market coefficients should not be over-interpreted. Although their point estimates display somewhat different patterns across horizons, equality of the two coefficients cannot be rejected in any of the six specifications, with $p$-values ranging from 0.121 to 0.845. Table~\ref{tab:RII_predictive} therefore does not establish that one of these risk layers is systematically more important than the other for predicting future risk. Its more robust implication is that both contain incremental forward-looking information after conditioning on current risk.

We assess out-of-sample predictive content using an expanding-window forecasting exercise with the Heterogeneous Autoregressive (HAR) model of \citet{Corsi2009} as the benchmark. The HAR model is well suited to this setting because it parsimoniously captures the strong persistence of volatility through components measured over different horizons. It therefore provides a demanding benchmark for assessing whether the RIIs contain predictive information beyond that already embedded in past volatility. Given the strong short-horizon predictive content of the Market RII in the in-sample volatility regressions, we augment the HAR specification with this index.

\begin{table}[!htbp]
\centering
\caption{In-Sample Predictive Content of the Risk Intensity Indices}
\label{tab:RII_predictive}
\begin{tabular}{lcccccc}
\toprule
& \multicolumn{3}{c}{Volatility}
& \multicolumn{3}{c}{Absolute Return} \\
\cmidrule(lr){2-4}\cmidrule(lr){5-7}
& $h=1$ & $h=4$ & $h=12$
& $h=1$ & $h=4$ & $h=12$ \\
\midrule

Micro RII
& 0.083$^{*}$
& 0.083$^{***}$
& 0.102$^{***}$
& 0.178$^{***}$
& 0.125$^{***}$
& 0.107$^{***}$ \\
& (0.047)
& (0.027)
& (0.030)
& (0.050)
& (0.033)
& (0.028) \\[0.3em]

Market RII
& 0.212$^{***}$
& 0.149$^{***}$
& 0.047$^{**}$
& 0.163$^{***}$
& 0.143$^{***}$
& 0.058$^{***}$ \\
& (0.046)
& (0.031)
& (0.021)
& (0.045)
& (0.027)
& (0.021) \\[0.3em]

Macro RII
& -0.016
& -0.014
& -0.020
& -0.031
& -0.026
& -0.022 \\
& (0.043)
& (0.045)
& (0.048)
& (0.049)
& (0.050)
& (0.051) \\

\midrule
Current risk & Yes & Yes & Yes & Yes & Yes & Yes \\
Commodity fixed effects & Yes & Yes & Yes & Yes & Yes & Yes \\
\midrule
$p$-val.: Market = Micro
& 0.121 & 0.127 & 0.183
& 0.845 & 0.632 & 0.187 \\
Commodities
& 52 & 52 & 52 & 52 & 52 & 52 \\
Observations & 31,720 & 7,904 & 2,600 & 31,720 & 7,904 & 2,600 \\
\bottomrule
\end{tabular}

\vspace{0.5em}
\begin{minipage}{0.98\textwidth}
\footnotesize
\textit{\textbf{Notes}:} The table reports fixed-effects regressions of future commodity risk on the standardized Micro, Market, and Macro Risk Intensity Indices (RIIs). Coefficients and standard errors are multiplied by 100 and therefore express changes in future risk in percentage points. For $h=1$, the dependent variable is risk in the following week. For $h=4$ and $h=12$, it is average risk over non-overlapping subsequent four- and twelve-week intervals, respectively. The use of non-overlapping intervals avoids mechanically inducing serial correlation through overlapping multi-period outcomes. All specifications control for the corresponding contemporaneous risk measure and commodity fixed effects. Standard errors, clustered at the commodity level, are reported in parentheses. The reported $p$-values correspond to tests of equality between the coefficients on the Market and Micro RIIs. Significance levels: $^{*}$ 10\%, $^{**}$ 5\%, $^{***}$ 1\%.
\end{minipage}
\end{table}

Table~\ref{tab:oos_forecasts} reports the out-of-sample forecasting results. The Clark--West test indicates that the Market RII contains incremental predictive information at short horizons. Using the conventional one-sided test, the null of no incremental predictive ability is rejected at the 5\% level for $h=1$ ($p=0.029$) and at the 10\% level for $h=4$ ($p=0.075$), but not at $h=12$ ($p=0.539$). Thus, the out-of-sample evidence supports the in-sample finding that the Market RII contains forward-looking information, particularly at short horizons.\footnote{The positive Clark--West statistics at the shorter horizons are compatible with the RMSE ratios being close to, but slightly above, one and the out-of-sample $R^2$ values being slightly negative. Because HAR is nested within HAR+RII, estimation of the additional RII coefficient introduces estimation noise under the null that its population coefficient is zero, which can raise the raw mean squared forecast error of the larger model. The Clark--West statistic adjusts the forecast-error differential for this effect and therefore tests a different null from an unadjusted comparison of forecast losses.}

\begin{table}[htbp]
\centering
\caption{Out-of-Sample Forecast Performance: HAR versus HAR + Market RII}
\label{tab:oos_forecasts}
\small
\begin{tabular}{l@{\hspace{1.2cm}}c@{\hspace{1.4cm}}c@{\hspace{1.4cm}}c}
\hline\hline
                              & $h=1$  & $h=4$  & $h=12$ \\
\hline
HAR RMSE                      & 0.0254 & 0.0157 & 0.0125 \\
HAR + RII RMSE                & 0.0254 & 0.0157 & 0.0126 \\
RMSE ratio                    & 1.0000 & 1.0008 & 1.0017 \\
                              &        &        &        \\
HAR MAE                       & 0.0149 & 0.0099 & 0.0084 \\
HAR + RII MAE                 & 0.0149 & 0.0100 & 0.0084 \\
MAE ratio                     & 1.0006 & 1.0026 & 0.9993 \\
                              &        &        &        \\
OOS $R^2$ (\%)                & -0.007 & -0.157 & -0.332 \\
                              &        &        &        \\
Clark--West $t$-statistic     & 1.911  & 1.446  & -0.098 \\
Clark--West $p$-value         & 0.029  & 0.075  & 0.539 \\
\hline\hline

\multicolumn{4}{p{0.82\textwidth}}{\footnotesize
\textit{Notes:} The benchmark model is the HAR specification. The augmented
model adds the Market Risk Intensity Index (RII). RMSE and MAE ratios
are defined as the loss of the HAR+RII model relative to the HAR model, so
values below one favour HAR+RII. OOS $R^2$ is measured relative to the HAR
benchmark. Clark--West $p$-values are one-sided.} \\

\end{tabular}
\end{table}

The magnitude of the associated gains in point-forecast accuracy is nevertheless small. The RMSE ratios are 1.0000, 1.0008, and 1.0017 at the one-, four-, and twelve-week horizons, respectively, while the MAE ratios are similarly close to one. The out-of-sample $R^2$ values are also close to zero. Hence, although the Clark--West test identifies incremental predictive information in the Market RII at the shorter horizons, incorporating this information produces little change in conventional point-forecast accuracy relative to the HAR benchmark.

\section{Concluding Remarks \label{sec:conclusion}}
This paper studies the sources of commodity risk and their relative importance across commodities and over time. We distinguish three layers: commodity-specific conditions, conditions prevailing elsewhere in the commodity market, and broader macro-financial conditions. The empirical challenge is that sensitivities to these sources are heterogeneous, while the relevant macro-financial component is not directly observed.
We address this challenge using a two-stage ``divide-and-conquer'' framework that accommodates heterogeneous sensitivities to observed micro and market risk and identifies the latent macro-financial component from a high-dimensional set of candidate indicators.

The results show that the three layers play economically different roles. Market sensitivities vary substantially across sectors, with particularly strong market exposure in energy, whereas sensitivities to the selected macro-financial sources are much more similar across agriculture, energy, and metals. These coefficient estimates, however, provide only part of the picture. Once sensitivities are combined with prevailing risk conditions through the Risk Intensity Indices, substantial differences emerge in both the magnitude and composition of commodity risk across sectors, individual commodities, and time.
Market risk is the largest source of commodity risk intensity on average, accounting for roughly 42\% of the total, compared with approximately 30\% for macro risk and 28\% for micro risk. These averages conceal pronounced sectoral differences. Market risk accounts for more than half of total risk intensity in energy, whereas agriculture assigns almost equal importance to market and macro risk, and metals exhibit a more balanced composition across the three layers. Risk composition also changes markedly over time, with the market layer becoming dominant during major commodity-market disruptions.
Risk intensity is also highly concentrated across commodities. The top 20\% of commodities account for almost half of aggregate micro and market risk intensity and about 40\% of macro risk intensity. Importantly, this concentration persists even though average macro sensitivities are relatively similar across sectors. Common macro-financial conditions can therefore affect commodities very differently once heterogeneous commodity-specific sensitivities are taken into account.
The RIIs also contain information about future commodity risk. In-sample regressions show that both the Micro and Market RIIs predict future volatility and absolute returns after controlling for current risk, whereas the Macro RII has no incremental predictive content. Out of sample, the Market RII contains additional forecasting information at short horizons according to the Clark--West test, although any associated gains in point-forecast accuracy appear economically small. The predictive evidence therefore complements the contemporaneous decomposition: the RIIs characterize the composition of current commodity risk and, for the micro and market layers, also contain information about its future evolution.

Several extensions appear particularly useful. The observable representation of macro risk is necessarily conditional on the candidate macro-financial information set, and expanding this information set may uncover additional dimensions of common risk. The market layer could also be generalized beyond sector-based leave-one-out measures to incorporate economic or financial networks linking commodities. Finally, while the RIIs contain short-horizon predictive information, their value for portfolio allocation, hedging, and tail-risk forecasting remains to be explored.

More broadly, the divide-and-conquer approach offers a promising framework for studying layered risk beyond commodity markets. Many economic and financial settings involve risk generated at multiple levels, heterogeneous sensitivities across units, and common components that are latent rather than directly observed. Separating observed local and market sources from latent common risk provides a way to measure their relative importance without imposing homogeneous responses across units.

\begingroup
\singlespacing
\setlength{\bibsep}{1pt}
\setlength{\parskip}{1pt}
\bibliographystyle{plainnat}
\bibliography{commodities}

\begin{thebibliography}{37}
\providecommand{\natexlab}[1]{#1}
\providecommand{\url}[1]{\texttt{#1}}
\expandafter\ifx\csname urlstyle\endcsname\relax
  \providecommand{\doi}[1]{doi: #1}\else
  \providecommand{\doi}{doi: \begingroup \urlstyle{rm}\Url}\fi

\bibitem[Adams and Gl{\"u}ck(2015)]{AdamsGluck2015}
Zeno Adams and Thorsten Gl{\"u}ck.
\newblock Financialization in commodity markets: A passing trend or the new normal?
\newblock \emph{Journal of Banking \& Finance}, 60:\penalty0 93--111, 2015.
\newblock \doi{10.1016/j.jbankfin.2015.07.008}.

\bibitem[Ahir et~al.(2022)Ahir, Bloom, and Furceri]{AhirBloomFurceri2022}
Hites Ahir, Nicholas Bloom, and Davide Furceri.
\newblock {The World Uncertainty Index}.
\newblock NBER Working Papers 29763, National Bureau of Economic Research, Inc, February 2022.
\newblock URL \url{https://ideas.repec.org/p/nbr/nberwo/29763.html}.

\bibitem[Ahn and Horenstein(2013)]{AhnHorenstein2013}
Seung~C. Ahn and Alex~R. Horenstein.
\newblock Eigenvalue ratio test for the number of factors.
\newblock \emph{Econometrica}, 81\penalty0 (3):\penalty0 1203--1227, 2013.
\newblock \doi{10.3982/ECTA9628}.

\bibitem[Avdjiev et~al.(2019)Avdjiev, Bruno, Koch, and Shin]{AvdjievEtal2019}
Stefan Avdjiev, Valentina Bruno, Cath{\'e}rine Koch, and Hyun~Song Shin.
\newblock The dollar exchange rate as a global risk factor: Evidence from investment.
\newblock \emph{IMF Economic Review}, 67\penalty0 (1):\penalty0 151--173, 2019.
\newblock \doi{10.1057/s41308-019-00074-4}.

\bibitem[Bai(2009)]{Bai2009}
Jushan Bai.
\newblock {Panel Data Models with Interactive Fixed Effects}.
\newblock \emph{Econometrica}, 77\penalty0 (4):\penalty0 1229--1279, 2009.

\bibitem[Bai and Ng(2002)]{BaiNg2002}
Jushan Bai and Serena Ng.
\newblock Determining the number of factors in approximate factor models.
\newblock \emph{Econometrica}, 70\penalty0 (1):\penalty0 191--221, 2002.

\bibitem[Baker et~al.(2016)Baker, Bloom, and Davis]{BakerBloomDavis2016}
Scott~R. Baker, Nicholas Bloom, and Steven~J. Davis.
\newblock {Measuring Economic Policy Uncertainty}.
\newblock \emph{The Quarterly Journal of Economics}, 131\penalty0 (4):\penalty0 1593--1636, nov 2016.
\newblock ISSN 0033-5533.
\newblock \doi{10.1093/QJE/QJW024}.
\newblock URL \url{https://dx.doi.org/10.1093/qje/qjw024}.

\bibitem[Baker(2021)]{Baker2021}
Steven~D. Baker.
\newblock The financialization of storable commodities.
\newblock \emph{Management Science}, 67\penalty0 (1):\penalty0 471--499, 2021.
\newblock \doi{10.1287/mnsc.2019.3445}.

\bibitem[Bakshi et~al.(2019)Bakshi, Gao, and Rossi]{BakshiGaoRossi2019}
Gurdip Bakshi, Xiaohui Gao, and Alberto~G. Rossi.
\newblock Understanding the sources of risk underlying the cross section of commodity returns.
\newblock \emph{Management Science}, 65\penalty0 (2):\penalty0 619--641, 2019.
\newblock \doi{10.1287/mnsc.2017.2840}.

\bibitem[Basak and Pavlova(2016)]{BasakPavlova2016}
Suleyman Basak and Anna Pavlova.
\newblock A model of financialization of commodities.
\newblock \emph{Journal of Finance}, 71\penalty0 (4):\penalty0 1511--1556, 2016.
\newblock \doi{10.1111/jofi.12408}.

\bibitem[Belloni et~al.(2012)Belloni, Chen, Chernozhukov, and Hansen]{BelloniEtal2012}
A.~Belloni, D.~Chen, V.~Chernozhukov, and C.~Hansen.
\newblock Sparse models and methods for optimal instruments with an application to eminent domain.
\newblock \emph{Econometrica}, 80:\penalty0 2369--2429, 2012.

\bibitem[Caggiano and Castelnuovo(2023)]{Caggiano2023}
Giovanni Caggiano and Efrem Castelnuovo.
\newblock Global financial uncertainty.
\newblock \emph{Journal of Applied Econometrics}, 38:\penalty0 432--449, 4 2023.
\newblock ISSN 0883-7252.
\newblock \doi{10.1002/jae.2958}.
\newblock URL \url{https://onlinelibrary.wiley.com/doi/10.1002/jae.2958}.

\bibitem[Caldara and Iacoviello(2022)]{CaldaraIacoviello2022}
Dario Caldara and Matteo Iacoviello.
\newblock {Measuring Geopolitical Risk}.
\newblock \emph{American Economic Review}, 112\penalty0 (4):\penalty0 1194--1225, apr 2022.
\newblock ISSN 0002-8282.
\newblock \doi{10.1257/aer.20191823}.
\newblock URL \url{https://pubs.aeaweb.org/doi/10.1257/aer.20191823}.

\bibitem[Chen et~al.(2025)Chen, Cui, Sarafidis, and Yamagata]{ChenEtal2025}
J.~Chen, G.~Cui, V.~Sarafidis, and T.~Yamagata.
\newblock Iv estimation of heterogeneous spatial dynamic panel models with interactive effects.
\newblock mimeo, 2025.

\bibitem[Cheng et~al.(2015)Cheng, Kirilenko, and Xiong]{ChengKirilenkoXiong2015}
Ing-Haw Cheng, Andrei Kirilenko, and Wei Xiong.
\newblock Convective risk flows in commodity futures markets.
\newblock \emph{Review of Finance}, 19\penalty0 (5):\penalty0 1733--1781, 2015.
\newblock \doi{10.1093/rof/rfu043}.

\bibitem[Christoffersen et~al.(2019)Christoffersen, Lunde, and Olesen]{ChristoffersenLundeOlesen2019}
Peter Christoffersen, Asger Lunde, and Kasper~V. Olesen.
\newblock Factor structure in commodity futures return and volatility.
\newblock \emph{Journal of Financial and Quantitative Analysis}, 54\penalty0 (3):\penalty0 1083--1115, 2019.
\newblock \doi{10.1017/S0022109018000765}.

\bibitem[Cochrane(2011)]{Cochrane2011}
J.~H. Cochrane.
\newblock Presidential address: Discount rates.
\newblock \emph{Journal of Finance}, 66\penalty0 (4):\penalty0 1047--1108, 2011.

\bibitem[Corsi(2009)]{Corsi2009}
Fulvio Corsi.
\newblock A simple approximate long-memory model of realized volatility.
\newblock \emph{Journal of Financial Econometrics}, 7\penalty0 (2):\penalty0 174--196, 2009.

\bibitem[Da et~al.(2024)Da, Tang, Tao, and Yang]{DaTangTaoYang2024}
Zhi Da, Ke~Tang, Yubo Tao, and Liyan Yang.
\newblock Financialization and commodity markets serial dependence.
\newblock \emph{Management Science}, 70\penalty0 (4):\penalty0 2122--2143, 2024.
\newblock \doi{10.1287/mnsc.2023.4797}.

\bibitem[Diebold et~al.(2018)Diebold, Liu, and Yilmaz]{DieboldLiuYilmaz2018}
Francis~X. Diebold, Laura Liu, and Kamil Yilmaz.
\newblock Commodity connectedness.
\newblock In Enrique~G. Mendoza, Ernesto Past{\'e}n, and Diego Saravia, editors, \emph{Monetary Policy and Global Spillovers: Mechanisms, Effects and Policy Measures}, volume~25, chapter~4, pages 97--136. Central Bank of Chile, Santiago, Chile, 2018.

\bibitem[Farid et~al.(2022)Farid, Naeem, Paltrinieri, and Nepal]{FaridEtAl2022}
Saqib Farid, Muhammad~Abubakr Naeem, Andrea Paltrinieri, and Rabindra Nepal.
\newblock Impact of {COVID-19} on the quantile connectedness between energy, metals and agriculture commodities.
\newblock \emph{Energy Economics}, 109:\penalty0 105962, 2022.
\newblock \doi{10.1016/j.eneco.2022.105962}.

\bibitem[Fernandez-Perez et~al.(2018)Fernandez-Perez, Frijns, Fuertes, and Miffre]{FernandezPerezFrijnsFuertesMiffre2018}
Adrian Fernandez-Perez, Bart Frijns, Ana-Maria Fuertes, and Jo{\"e}lle Miffre.
\newblock The skewness of commodity futures returns.
\newblock \emph{Journal of Banking \& Finance}, 86:\penalty0 143--158, 2018.
\newblock \doi{10.1016/j.jbankfin.2017.06.015}.

\bibitem[French et~al.(1987)French, Schwert, and Stambaugh]{FrenchSchwertStambaugh1987}
Kenneth~R. French, G.~William Schwert, and Robert~F. Stambaugh.
\newblock Expected stock returns and volatility.
\newblock \emph{Journal of Financial Economics}, 19\penalty0 (1):\penalty0 3--29, 1987.
\newblock \doi{10.1016/0304-405X(87)90026-2}.

\bibitem[Fuertes et~al.(2015)Fuertes, Miffre, and Fernandez-Perez]{FuertesMiffreFernandezPerez2015}
Ana-Maria Fuertes, Jo{\"e}lle Miffre, and Adrian Fernandez-Perez.
\newblock Commodity strategies based on momentum, term structure, and idiosyncratic volatility.
\newblock \emph{Journal of Futures Markets}, 35\penalty0 (3):\penalty0 274--297, 2015.
\newblock \doi{10.1002/fut.21656}.

\bibitem[Gorton and Rouwenhorst(2006)]{GortonRouwenhorst2006}
Gary Gorton and K.~Geert Rouwenhorst.
\newblock Facts and fantasies about commodity futures.
\newblock \emph{Financial Analysts Journal}, 62\penalty0 (2):\penalty0 47--68, 2006.

\bibitem[Gorton et~al.(2013)Gorton, Hayashi, and Rouwenhorst]{GortonHayashiRouwenhorst2013}
Gary~B. Gorton, Fumio Hayashi, and K.~Geert Rouwenhorst.
\newblock The fundamentals of commodity futures returns.
\newblock \emph{Review of Finance}, 17\penalty0 (1):\penalty0 35--105, 2013.

\bibitem[Harvey et~al.(2016)Harvey, Liu, and Zhu]{HarveyLiuZhu2016}
Campbell~R. Harvey, Yan Liu, and Heqing Zhu.
\newblock … and the cross-section of expected returns.
\newblock \emph{The Review of Financial Studies}, 29\penalty0 (1):\penalty0 5--68, 2016.
\newblock \doi{10.1093/rfs/hhv059}.

\bibitem[Juodis and Sarafidis(2018)]{JuodisSarafidis2018}
A.~Juodis and V.~Sarafidis.
\newblock Fixed t dynamic panel data estimators with multi-factor errors.
\newblock \emph{Econometric Reviews}, 38\penalty0 (8):\penalty0 893--929, 2018.

\bibitem[Juodis and Sarafidis(2022)]{JuodisSarafidis2022b}
A.~Juodis and V.~Sarafidis.
\newblock A linear estimator for factor-augmented fixed-t panels with endogenous regressors.
\newblock \emph{Journal of Business and Economic Statistics}, 40\penalty0 (1):\penalty0 1--15, 2022.

\bibitem[Jurado et~al.(2015)Jurado, Ludvigson, and Ng]{JuradoLudvigsonNg2015}
Kyle Jurado, Sydney~C Ludvigson, and Serena Ng.
\newblock Measuring uncertainty.
\newblock \emph{American Economic Review}, 105\penalty0 (3):\penalty0 1177--1216, 2015.
\newblock \doi{10.1257/aer.20131193}.

\bibitem[Kapetanios et~al.(2026)Kapetanios, Sarafidis, and Ventouri]{KapetaniosEtal2026}
George Kapetanios, Vasilis Sarafidis, and Alexia Ventouri.
\newblock Model selection in high-dimensional linear regression using boosting with multiple testing, 2026.
\newblock URL \url{https://arxiv.org/pdf/2602.19705}.
\newblock arXiv preprint.

\bibitem[Miranda-Agrippino and Rey(2020)]{Miranda-Agrippino2020}
Silvia Miranda-Agrippino and Hélène Rey.
\newblock {U.S.} monetary policy and the global financial cycle.
\newblock \emph{The Review of Economic Studies}, 87:\penalty0 2754--2776, 11 2020.
\newblock ISSN 0034-6527.
\newblock \doi{10.1093/restud/rdaa019}.
\newblock URL \url{https://academic.oup.com/restud/article/87/6/2754/5834728}.

\bibitem[Norkute et~al.(2021)Norkute, Sarafidis, Yamagata, and Cui]{NorkuteEtal2021}
Milda Norkute, Vasilis Sarafidis, Takashi Yamagata, and Guowei Cui.
\newblock {Instrumental variable estimation of dynamic linear panel data models with defactored regressors and a multifactor error structure}.
\newblock \emph{Journal of Econometrics}, 220\penalty0 (2):\penalty0 416--446, 2021.

\bibitem[Qiao and Han(2023)]{QiaoHan2023}
Tongshuai Qiao and Liyan Han.
\newblock {COVID-19} and tail risk contagion across commodity futures markets.
\newblock \emph{Journal of Futures Markets}, 43\penalty0 (2):\penalty0 242--272, 2023.
\newblock \doi{10.1002/fut.22388}.

\bibitem[Silvennoinen and Thorp(2013)]{SilvennoinenThorp2013}
Annastiina Silvennoinen and Susan Thorp.
\newblock Financialization, crisis and commodity correlation dynamics.
\newblock \emph{Journal of International Financial Markets, Institutions and Money}, 24:\penalty0 42--65, 2013.
\newblock \doi{10.1016/j.intfin.2012.11.007}.

\bibitem[Szymanowska et~al.(2014)Szymanowska, de~Roon, Nijman, and van~den Goorbergh]{SzymanowskaDeRoonNijmanVanDenGoorbergh2014}
Marta Szymanowska, Frans de~Roon, Theo Nijman, and Rob van~den Goorbergh.
\newblock An anatomy of commodity futures risk premia.
\newblock \emph{Journal of Finance}, 69\penalty0 (1):\penalty0 453--482, 2014.

\bibitem[Tang and Xiong(2012)]{TangXiong2012}
Ke~Tang and Wei Xiong.
\newblock Index investment and the financialization of commodities.
\newblock \emph{Financial Analysts Journal}, 68\penalty0 (6):\penalty0 54--74, 2012.

\end{thebibliography}


\begin{thebibliography}{35}
\providecommand{\natexlab}[1]{#1}
\providecommand{\url}[1]{\texttt{#1}}
\expandafter\ifx\csname urlstyle\endcsname\relax
  \providecommand{\doi}[1]{doi: #1}\else
  \providecommand{\doi}{doi: \begingroup \urlstyle{rm}\Url}\fi

\bibitem[Abiad and Qureshi(2023)]{Abiad_Oil_uncertainty}
Abdul Abiad and Irfan~A. Qureshi.
\newblock The macroeconomic effects of oil price uncertainty.
\newblock \emph{Energy Economics}, 125:\penalty0 106839, 2023.
\newblock ISSN 0140-9883.
\newblock \doi{https://doi.org/10.1016/j.eneco.2023.106839}.
\newblock URL \url{https://www.sciencedirect.com/science/article/pii/S0140988323003377}.

\bibitem[Ahir et~al.(2022)Ahir, Bloom, and Furceri]{AhirBloomFurceri2022}
Hites Ahir, Nicholas Bloom, and Davide Furceri.
\newblock {The World Uncertainty Index}.
\newblock NBER Working Papers 29763, National Bureau of Economic Research, Inc, February 2022.
\newblock URL \url{https://ideas.repec.org/p/nbr/nberwo/29763.html}.

\bibitem[Ahn and Horenstein(2013)]{AhnHorenstein2013}
Seung~C. Ahn and Alex~R. Horenstein.
\newblock Eigenvalue ratio test for the number of factors.
\newblock \emph{Econometrica}, 81\penalty0 (3):\penalty0 1203--1227, 2013.
\newblock \doi{10.3982/ECTA9628}.

\bibitem[Amaya et~al.(2015)Amaya, Christoffersen, Jacobs, and Vasquez]{Amaya2015}
Diego Amaya, Peter Christoffersen, Kris Jacobs, and Aurelio Vasquez.
\newblock Does realized skewness predict the cross-section of equity returns?
\newblock \emph{Journal of Financial Economics}, 118\penalty0 (1):\penalty0 135--167, 2015.
\newblock ISSN 0304-405X.
\newblock \doi{https://doi.org/10.1016/j.jfineco.2015.02.009}.
\newblock URL \url{https://www.sciencedirect.com/science/article/pii/S0304405X15001257}.

\bibitem[Angelini et~al.(2025)Angelini, Bontempi, Angelis, Neri, and Sorge]{AngeliniEtAl2025}
Giovanni Angelini, Maria~Elena Bontempi, Luca~De Angelis, Paolo Neri, and Marco~Maria Sorge.
\newblock Shocking concerns: public perception about climate change and the macroeconomy.
\newblock Papers 2505.04669, arXiv.org, May 2025.
\newblock URL \url{https://ideas.repec.org/p/arx/papers/2505.04669.html}.

\bibitem[Aruoba et~al.(2009)Aruoba, Diebold, and Scotti]{ArubaDieboldScotti2009}
S.~Bora\u{g}an Aruoba, Francis~X. Diebold, and Chiara Scotti.
\newblock Real-time measurement of business conditions.
\newblock \emph{Journal of Business \& Economic Statistics}, 27\penalty0 (4):\penalty0 417--427, 2009.
\newblock \doi{10.1198/jbes.2009.07205}.

\bibitem[Aslanidis et~al.(2022)Aslanidis, Bariviera, and Óscar G.~López]{Aslanidis2022}
Nektarios Aslanidis, Aurelio~F. Bariviera, and Óscar G.~López.
\newblock The link between cryptocurrencies and google trends attention.
\newblock \emph{Finance Research Letters}, 47:\penalty0 102654, 6 2022.
\newblock ISSN 15446123.
\newblock \doi{10.1016/j.frl.2021.102654}.
\newblock URL \url{https://linkinghub.elsevier.com/retrieve/pii/S1544612321005833}.

\bibitem[Bai(2003)]{Bai2003}
Jushan Bai.
\newblock {Inferential Theory for Factor Models of Large Dimensions}.
\newblock \emph{Econometrica}, 71\penalty0 (1):\penalty0 135--171, 2003.

\bibitem[Bai(2009)]{Bai2009}
Jushan Bai.
\newblock {Panel Data Models with Interactive Fixed Effects}.
\newblock \emph{Econometrica}, 77\penalty0 (4):\penalty0 1229--1279, 2009.

\bibitem[Bai and Ng(2002)]{BaiNg2002}
Jushan Bai and Serena Ng.
\newblock Determining the number of factors in approximate factor models.
\newblock \emph{Econometrica}, 70\penalty0 (1):\penalty0 191--221, 2002.

\bibitem[Baker et~al.(2016)Baker, Bloom, and Davis]{BakerBloomDavis2016}
Scott~R. Baker, Nicholas Bloom, and Steven~J. Davis.
\newblock {Measuring Economic Policy Uncertainty}.
\newblock \emph{The Quarterly Journal of Economics}, 131\penalty0 (4):\penalty0 1593--1636, nov 2016.
\newblock ISSN 0033-5533.
\newblock \doi{10.1093/QJE/QJW024}.
\newblock URL \url{https://dx.doi.org/10.1093/qje/qjw024}.

\bibitem[Baker et~al.(2019)Baker, Bloom, Davis, and Kost]{NBERw25720}
Scott~R Baker, Nicholas Bloom, Steven~J Davis, and Kyle~J Kost.
\newblock Policy news and stock market volatility.
\newblock Working Paper 25720, National Bureau of Economic Research, March 2019.
\newblock URL \url{http://www.nber.org/papers/w25720}.

\bibitem[{Board of Governors of the Federal Reserve System}(2025)]{NominalUSDindex}
{Board of Governors of the Federal Reserve System}.
\newblock Nominal broad u.s. dollar index [dtwexbgs], 2 2025.
\newblock URL \url{https://fred.stlouisfed.org/series/DTWEXBGS}.
\newblock retrieved 02/02/2025.

\bibitem[Caldara and Iacoviello(2022)]{CaldaraIacoviello2022}
Dario Caldara and Matteo Iacoviello.
\newblock {Measuring Geopolitical Risk}.
\newblock \emph{American Economic Review}, 112\penalty0 (4):\penalty0 1194--1225, apr 2022.
\newblock ISSN 0002-8282.
\newblock \doi{10.1257/aer.20191823}.
\newblock URL \url{https://pubs.aeaweb.org/doi/10.1257/aer.20191823}.

\bibitem[Caldara et~al.(2020)Caldara, Iacoviello, Molligo, Prestipino, and Raffo]{CaldaraEtAl2020}
Dario Caldara, Matteo Iacoviello, Patrick Molligo, Andrea Prestipino, and Andrea Raffo.
\newblock The economic effects of trade policy uncertainty.
\newblock \emph{Journal of Monetary Economics}, 109:\penalty0 38--59, 2020.
\newblock \doi{10.1016/j.jmoneco.2019.11.002}.

\bibitem[{Cboe Global Markets}()]{CBOEVIX}
{Cboe Global Markets}.
\newblock Cboe volatility index (vix).
\newblock Cboe Global Markets, Chicago.
\newblock Retrieved from LSEG Datastream.

\bibitem[Cui et~al.(2022)Cui, Norkutė, Sarafidis, and Yamagata]{CuiEtal2022}
Guowei Cui, Milda Norkutė, Vasilis Sarafidis, and Takashi Yamagata.
\newblock {Two-stage instrumental variable estimation of linear panel data models with interactive effects}.
\newblock \emph{The Econometrics Journal}, 25\penalty0 (2):\penalty0 340--361, 10 2022.
\newblock \doi{10.1093/ectj/utab029}.
\newblock URL \url{https://doi.org/10.1093/ectj/utab029}.

\bibitem[Davis(2016)]{Davis2016}
Steven~J Davis.
\newblock An index of global economic policy uncertainty.
\newblock Working Paper 22740, National Bureau of Economic Research, October 2016.
\newblock URL \url{http://www.nber.org/papers/w22740}.

\bibitem[{Federal Reserve Bank of Chicago}(2025)]{NFCILeverage}
{Federal Reserve Bank of Chicago}.
\newblock {C}hicago {F}ed {N}ational {F}inancial {C}onditions {L}everage {S}ubindex.
\newblock \url{https://fred.stlouisfed.org/series/NFCILEVERAGE}, 2025.

\bibitem[{Federal Reserve Bank of Cleveland}(2025)]{InflationExpectation}
{Federal Reserve Bank of Cleveland}.
\newblock 5-year expected inflation, 2 2025.
\newblock URL \url{https://fred.stlouisfed.org/series/EXPINF5YR}.
\newblock retrieved 02/02/2025.

\bibitem[{Federal Reserve Bank of St. Louis}(2025)]{StLouisStressIndex}
{Federal Reserve Bank of St. Louis}.
\newblock {S}t. {L}ouis {F}ed {F}inancial {S}tress {I}ndex [{STLFSI4}].
\newblock \url{https://fred.stlouisfed.org/series/STLFSI4}, 2025.
\newblock Accessed: 2025-01-21.

\bibitem[Gavriilidis et~al.(2026)Gavriilidis, Känzig, Raghavan, and Stock]{GavriilidisEtAl2026}
Konstantinos Gavriilidis, Diego~R Känzig, Ramya Raghavan, and James~H Stock.
\newblock The macroeconomic effects of climate policy uncertainty.
\newblock Working Paper 34762, National Bureau of Economic Research, January 2026.
\newblock URL \url{http://www.nber.org/papers/w34762}.

\bibitem[Hakkio and Keeton(2009)]{HakkioKeeton2009}
Craig~S. Hakkio and William~R. Keeton.
\newblock Financial stress: What is it, how can it be measured, and why does it matter?
\newblock \emph{Federal Reserve Bank of Kansas City Economic Review}, 94\penalty0 (2):\penalty0 5--50, 2009.

\bibitem[{ICE Benchmark Administration}()]{ICEMOVE}
{ICE Benchmark Administration}.
\newblock Ice bofa move index.
\newblock Intercontinental Exchange, New York.
\newblock Retrieved from LSEG Datastream.

\bibitem[Juodis and Sarafidis(2022)]{JuodisSarafidis2022}
A.~Juodis and V.~Sarafidis.
\newblock An incidental parameters free inference approach for panels with common shocks.
\newblock \emph{Journal of Econometrics}, 229:\penalty0 19--54, 2022.

\bibitem[Jurado et~al.(2015)Jurado, Ludvigson, and Ng]{JuradoLudvigsonNg2015}
Kyle Jurado, Sydney~C Ludvigson, and Serena Ng.
\newblock Measuring uncertainty.
\newblock \emph{American Economic Review}, 105\penalty0 (3):\penalty0 1177--1216, 2015.
\newblock \doi{10.1257/aer.20131193}.

\bibitem[Kapetanios(2008)]{Kapetanios2008}
George Kapetanios.
\newblock A bootstrap procedure for panel data sets with many cross-sectional units.
\newblock \emph{The Econometrics Journal}, 11\penalty0 (2):\penalty0 377--395, 2008.
\newblock \doi{10.1111/j.1368-423X.2008.00243.x}.

\bibitem[Kapetanios et~al.(2026)Kapetanios, Sarafidis, and Ventouri]{KapetaniosEtal2026}
George Kapetanios, Vasilis Sarafidis, and Alexia Ventouri.
\newblock Model selection in high-dimensional linear regression using boosting with multiple testing, 2026.
\newblock URL \url{https://arxiv.org/pdf/2602.19705}.
\newblock arXiv preprint.

\bibitem[Moon and Weidner(2015)]{MoonWeidner2015}
Hyungsik~Roger Moon and Martin Weidner.
\newblock Linear regression for panel with unknown number of factors as interactive fixed effects.
\newblock \emph{Econometrica}, 83\penalty0 (4):\penalty0 1543--1579, 2015.

\bibitem[Mor\~{a}o(2025)]{Morao2025}
Hugo Mor\~{a}o.
\newblock The economic effects of tensions in energy transportation.
\newblock \emph{Research in Transportation Economics}, 112, 2025.
\newblock \doi{10.1016/j.retrec.2025.101586}.

\bibitem[Norkute et~al.(2021)Norkute, Sarafidis, Yamagata, and Cui]{NorkuteEtal2021}
Milda Norkute, Vasilis Sarafidis, Takashi Yamagata, and Guowei Cui.
\newblock {Instrumental variable estimation of dynamic linear panel data models with defactored regressors and a multifactor error structure}.
\newblock \emph{Journal of Econometrics}, 220\penalty0 (2):\penalty0 416--446, 2021.

\bibitem[Ongan et~al.(2025)Ongan, Gocer, and Isik]{OnganEtAl2025}
Serdar Ongan, Ismet Gocer, and Cem Isik.
\newblock Introducing the new esg-based sustainability uncertainty index (esgui).
\newblock \emph{Sustainable Development}, 2025.
\newblock \doi{10.1002/sd.3351}.

\bibitem[Pesaran(2006)]{Pesaran2006}
M.~Hashem Pesaran.
\newblock Estimation and inference in large heterogeneous panels with a multifactor error structure.
\newblock \emph{Econometrica}, 74\penalty0 (4):\penalty0 967--1012, 2006.

\bibitem[Sarafidis and Wansbeek(2012)]{SarafidisWansbeek2012}
Vasilis Sarafidis and Tom Wansbeek.
\newblock {Cross-Sectional Dependence in Panel Data Analysis}.
\newblock \emph{Econometric Reviews}, 31\penalty0 (5):\penalty0 483--531, 2012.

\bibitem[{University of Michigan}(2025)]{ConsumerMichigan}
{University of Michigan}.
\newblock {S}urveys of {C}onsumers.
\newblock \url{http://www.sca.isr.umich.edu/}, 2025.
\newblock Accessed: 2025-01-21.

\end{thebibliography}
\endgroup

\appendix

\begin{bibunit}[plainnat]

\renewcommand{\thesection}{Appendix}
\renewcommand{\thesubsection}{Appendix.\arabic{subsection}}

\renewcommand{\theequation}{A.\arabic{equation}}
\setcounter{equation}{0}

\renewcommand{\thetable}{A.\arabic{table}}
\setcounter{table}{0}

\renewcommand{\thelemma}{A.\arabic{lemma}}
\setcounter{lemma}{0}

\renewcommand{\theproposition}{A.\arabic{proposition}}
\setcounter{proposition}{0}

\renewcommand{\thecorollary}{A.\arabic{corollary}}
\setcounter{corollary}{0}



\clearpage


\clearpage
\setcounter{page}{1}
\setcounter{footnote}{0}

\begin{titlepage}
\centering

{\LARGE\bfseries The Anatomy of Commodity Risk: Micro, Market, and Economy-Wide Sources\par}
\vspace{0.6em}
{\Large\bfseries Online Appendix\par}

\vspace{1.5em}
{\large
\begin{tabular}{c}
Nektarios Aslanidis$^{1}$\\
Aurelio F. Bariviera$^{1}$\\
George Kapetanios$^{2}$\\
Vasilis Sarafidis$^{3,*}$\\
Alexia Ventouri$^{2}$
\end{tabular}
\par}

\vspace{0.8em}
{\small $^{1}$Department of Economics, Universitat Rovira i Virgili, ECO-SOS, Av. Universitat 1, Reus 43204, Spain\par}

{\small $^{2}$Department of Banking and Finance, King's College London, Bush House, 30 Aldwych, London WC2B 4BG, UK\par}

{\small $^{3}$Brunel Business School, Brunel University London, Eastern Gateway, Kingston Lane, Uxbridge UB8 3PN, UK\par}


\raggedright
\normalsize

\bigskip
\textbf{Abstract:}
This Online Appendix provides supplementary material for ``The Anatomy of Commodity Risk: Micro, Market, and Economy-Wide Sources.'' It contains additional information on the data and construction of the variables used in the empirical analysis, together with further details on the two-stage divide-and-conquer estimation framework. The Appendix states the additional regularity conditions and provides the technical lemmas and proofs supporting the asymptotic results in the paper. It also reports Monte Carlo evidence on the finite-sample performance of the PCA--BMT procedure and compares its variable-selection performance with pooled and individual-level Lasso alternatives.
\end{titlepage}

\renewcommand{\thetable}{OA.\arabic{table}}
\setcounter{table}{0}
\renewcommand{\theequation}{OA.\arabic{equation}}
\setcounter{equation}{0}

\renewcommand{\thetable}{OA.\arabic{table}}
\setcounter{table}{0}

\setcounter{remark}{0}
\renewcommand{\theremark}{A.\arabic{remark}}

\renewcommand{\thefigure}{OA.\arabic{figure}}
\setcounter{figure}{0}
\setcounter{page}{1}

\setcounter{section}{0}
\setcounter{subsection}{0}

\setcounter{section}{0}
\renewcommand{\thesection}{Appendix \Alph{section}}
\renewcommand{\thesubsection}{\thesection.\arabic{subsection}}
\renewcommand{\thesubsubsection}{\thesubsection.\arabic{subsubsection}}

\section{Data Description\label{sec:Appendix_Data}}

Table \ref{tab:commodity_list} provides a detailed description and classification of the commodities in our sample.

{\footnotesize
\begin{longtable}{p{2.2cm}p{7.3cm}p{4.2cm}p{1.6cm}}

\caption{Commodity Futures Series in the Sample} \\
\toprule
Mnemonic & Name & Sector & Currency \\
\midrule
\endfirsthead

\multicolumn{4}{c}%
{{ \tablename \thetable{} -- continued from previous page}} \\
\toprule Mnemonic & Name & Sector & Currency \\
\midrule
\endhead
\hline \multicolumn{4}{r}{{Continued on next page}} \\
\hline
\endfoot

\hline \hline
\endlastfoot

    \texttt{CCFC.01} & CBT-CORN COMP. & Agriculture and livestock & USD \\
    \texttt{COFC.01} & CBT-OATS COMP. & Agriculture and livestock & USD \\
    \texttt{CMSC.01} & CBT-SOYABEAN MEAL COMP & Agriculture and livestock & USD \\
    \texttt{CSYC.01} & CBT-SOYBEANS COMP. & Agriculture and livestock & USD \\
    \texttt{CWFC.01} & CBT-WHEAT COMPOSITE & Agriculture and livestock & USD \\
    \texttt{CCCC.01} & CME-CLASS IV MILK COMP & Agriculture and livestock & USD \\
    \texttt{CFDC.01} & CME-FEEDER CATTLE COMP. & Agriculture and livestock & USD \\
    \texttt{CLGC.01} & CME-LEAN HOGS COMP. & Agriculture and livestock & USD \\
    \texttt{CLDC.01} & CME-LIVE CATTLE COMP. & Agriculture and livestock & USD \\
    \texttt{CMKC.01} & CME-MILK CLASS III COMP. & Agriculture and livestock & USD \\
    \texttt{NCCC.01} & CSCE-COCOA & Agriculture and livestock & USD \\
    \texttt{NKCC.01} & CSCE-COFFEE 'C' & Agriculture and livestock & USD \\
    \texttt{DE.C.01} & DCE-EGG & Agriculture and livestock & CNY \\
    \texttt{KKKC.01} & KCBT-HRW WHEAT COMP & Agriculture and livestock & USD \\
    \texttt{KPOC.01} & KLSE-CRUDE PALM OIL & Agriculture and livestock & MYR \\
    \texttt{LRRC.01} & LIFFE-ROBUSTA COFFEE REVISED & Agriculture and livestock & USD \\
    \texttt{MCDC.01} & MCX-CARDAMOM & Agriculture and livestock & INR \\
    \texttt{MKTC.01} & MCX-KAPAS & Agriculture and livestock & INR \\
    \texttt{MMTC.01} & MCX-MENTHA OIL & Agriculture and livestock & INR \\
    \texttt{NBLC.01} & NCDEX-BARLEY & Agriculture and livestock & INR \\
    \texttt{NMXC.01} & NCDEX-MAIZE FEED IND & Agriculture and livestock & INR \\
    \texttt{NJOC.01} & NYCE-(FCOJ-A) ORANGE JUICE & Agriculture and livestock & USD \\
    \texttt{SRSC.01} & SICOM-RUBBER RSS3 & Agriculture and livestock & USD \\
    \texttt{STFC.01} & SICOM-RUBBER TSR20 & Agriculture and livestock & USD \\
    \midrule
    \texttt{DL.C.01} & DCE-LLDPE & Energy and petrochemicals & CNY \\
    \texttt{DPLC.01} & DCE-POLYPROPYLENE & Energy and petrochemicals & CNY \\
    \texttt{DV.C.01} & DCE-PVC & Energy and petrochemicals & CNY \\
    \texttt{DOQC.01} & GME-OMAN CRUDE & Energy and petrochemicals & USD \\
    \texttt{LLCC.01} & ICE-BRENT CRUDE OIL & Energy and petrochemicals & USD \\
    \texttt{LMDC.01} & ICE-COAL RICHARD BAY & Energy and petrochemicals & USD \\
    \texttt{LMCC.01} & ICE-COAL ROTTERDAM & Energy and petrochemicals & USD \\
    \texttt{LLEC.01} & ICE-GAS OIL & Energy and petrochemicals & USD \\
    \texttt{NQZC.01} & NYM-ARGUS PROPANE FE INDEX & Energy and petrochemicals & USD \\
    \texttt{NRBC.01} & NYMEX-RBOB GASOLINE & Energy and petrochemicals & USD \\
    \texttt{NCLC.01} & NYM-LIGHT CRUDE OIL & Energy and petrochemicals & USD \\
    \texttt{NNGC.01} & NYM-NATURAL GAS & Energy and petrochemicals & USD \\
    \texttt{NHOC.01} & NYM-NY HARBOR ULSD & Energy and petrochemicals & USD \\
    \texttt{ZMEC.01} & ZCE-METHANOL & Energy and petrochemicals & CNY \\
    \midrule
    \texttt{NGCC.01} & CMX-GOLD 100 OZ & Metals & USD \\
    \texttt{NHGC.01} & CMX-HIGH GRADE COPPER & Metals & USD \\
    \texttt{NSLC.01} & CMX-SILVER 5000 OZ & Metals & USD \\
    \texttt{LAHC.01} & LME-ALUMINIUM & Metals & USD \\
    \texttt{LEDC.01} & LME-LEAD & Metals & USD \\
    \texttt{LNIC.01} & LME-NICKEL & Metals & USD \\
    \texttt{LTIC.01} & LME-TIN & Metals & USD \\
    \texttt{LZZC.01} & LME-ZINC & Metals & USD \\
    \texttt{NTJC.01} & NYMEX-IRON ORE 62 & Metals & USD \\
    \texttt{NHCC.01} & NYM-HR COIL STEEL US & Metals & USD \\
    \texttt{NPAC.01} & NYM-PALLADIUM & Metals & USD \\
    \texttt{NPLC.01} & NYM-PLATINUM & Metals & USD \\
    \texttt{ZFSC.01} & ZCE-FERROSILICON & Metals & CNY \\
    \texttt{ZMSC.01} & ZCE-SILICON MANGANES & Metals & CNY \\
  \label{tab:commodity_list}%
\end{longtable}%
}

\begin{minipage}{\textwidth}
{\footnotesize \textit{Notes:} Mnemonics and contract names are as reported by LSEG Datastream.``Currency'' reports the currency in which the contract settles on its home exchange. Thirteen of the
fifty-two contracts settle in a currency other than the US dollar, namely those traded on the Dalian
Commodity Exchange and the Zhengzhou Commodity Exchange (CNY), the Multi Commodity Exchange of India and
the National Commodity and Derivatives Exchange (INR), and Bursa Malaysia Derivatives (MYR). All
settlement prices are converted to US dollars at the contemporaneous spot exchange rate prior to the
computation of returns, realised moments, and market aggregates, as described in
Section~\ref{sec:data}.}
\end{minipage}

Let $P_{i,d}$ denote the settlement price of commodity $i$ on trading day $d$. Daily log-returns are $r_{i,d}=\ln(P_{i,d}/P_{i,d-1})$, and the weekly log-return used as the dependent variable in Equation (\ref{model_with_u}) is their within-week sum,
\begin{equation}
r_{i,t}=\sum_{d\in\mathcal{D}_{t}} r_{i,d},
\end{equation}
where $\mathcal{D}_{t}$ denotes the set of trading days in week $t$ and $n_{i,t}=|\mathcal{D}_{t}|$ its cardinality for commodity $i$.

Following the realised-moment approach of \citet{Amaya2015}, we construct weekly realised volatility, skewness, and kurtosis from the within-week distribution of daily returns:
\begin{align}
\mathrm{VLTY}_{i,t} &= \sqrt{\sum_{d\in\mathcal{D}_{t}} r_{i,d}^{2}}, \label{eq:rv}\\[4pt]
\mathrm{SKEW}_{i,t} &= \frac{\sqrt{n_{i,t}}\sum_{d\in\mathcal{D}_{t}} r_{i,d}^{3}}{\mathrm{VLTY}_{i,t}^{3}}, \label{eq:rskew}\\[4pt]
\mathrm{KTS}_{i,t} &= \frac{n_{i,t}\sum_{d\in\mathcal{D}_{t}} r_{i,d}^{4}}{\mathrm{VLTY}_{i,t}^{4}}. \label{eq:rkurt}
\end{align}
The scaling factors $\sqrt{n_{i,t}}$ and $n_{i,t}$ ensure that the realised third and fourth moments are comparable across weeks with differing numbers of trading days.

Open interest enters the micro layer in logarithms, $\mathrm{OI}_{i,t}=\ln(\mathrm{OI}_{i,t})$, measured at the last trading day of week $t$, and captures the depth of outstanding commitment in contract $i$.

The market layer aggregates the same risk characteristics across the cross-section. To avoid the mechanical endogeneity that would arise if the market aggregate contained the asset whose return it is used to explain, every market-level variable is constructed on a leave-one-out basis. For any asset-level characteristic $z_{i,t}\in\{r_{i,t},\mathrm{VLTY}_{i,t},\mathrm{SKEW}_{i,t},\mathrm{KTS}_{i,t}, \mathrm{OI}_{i,t}\}$, the corresponding market component facing commodity $i$ is
\begin{equation}
z^{\mathrm{mkt}}_{i,t}=\frac{1}{N-1}\sum_{\substack{j\in\mathcal{S}_i\\j\neq i}} z_{j,t},
\label{eq:loo}
\end{equation}
which yields $r_{-i,t}$, $\mathrm{CVLTY\_idx}_{i,t}$, $\mathrm{CSKEW\_idx}_{i,t}$, $\mathrm{CKTS\_idx}_{i,t}$, and $\mathrm{COI}_{i,t}$ respectively. Each commodity therefore faces a slightly different market benchmark, from which its own contribution has been purged.

This construction serves two purposes.
Econometrically, it removes the mechanical own-observation component that would arise if the market aggregate used to explain $r_{i,t}$) contained commodity (i)'s own return or risk characteristic. This is particularly important in our setting because the sensitivities to market-level conditions are themselves objects of interest: retaining commodity (i) in the corresponding aggregate would mechanically induce comovement between the commodity-specific and market-level variables. Economically, the leave-one-out construction provides a clean interpretation of the market layer as conditions prevailing elsewhere in the same commodity sector. The coefficients ($\boldsymbol{\beta}_{i}$) therefore measure commodity (i)'s sensitivity to market-level risk conditions constructed exclusively from other commodities, rather than to an aggregate that partly reflects its own behaviour.

Table \ref{tab:summary_stats} presents summary statistics for the micro and market-level variables. The between standard deviation captures variation across assets, while the within standard deviation reflects time-series variation within each asset.

\begin{table}[H]
  \centering

  \caption{Summary Statistics for Idiosyncratic and Market-Level Factors}
  \renewcommand{\arraystretch}{0.9}
  \scalebox{0.85}{
  \begin{tabular}{lcccccc}
    \toprule
    \multirow{2}{*}{Variable} & \multirow{2}{*}{Mean} & \multicolumn{3}{c}{Standard Deviation} & \multirow{2}{*}{Min} & \multirow{2}{*}{Max} \\
    & & Overall & Within & Between & & \\
    \midrule
    $r$        & 0.004 & 0.179 & 0.179 & 0.009 & -4.444 & 1.509 \\
    ILQ        & 48.458 & 1153.2 & 1112.789 & 306.435 & 0.000 & 57929 \\
    VLT        & 0.147 & 0.134 & 0.114 & 0.072 & 0.000 & 3.390 \\
    SKW        & 0.007 & 1.156 & 1.154 & 0.079 & -2.646 & 2.573 \\
    KTS        & 2.661 & 1.012 & 0.997 & 0.175 & 1.095 & 6.999 \\
    CILQ       & 26.064 & 63.951 & 63.951 & 0.000 & 0.123 & 649.541 \\
    CVLT       & 0.065 & 0.039 & 0.039 & 0.000 & 0.002 & 0.311 \\
    CSKW       & 0.094 & 0.694 & 0.694 & 0.000 & -1.628 & 1.482 \\
    CKTS       & 2.614 & 0.574 & 0.574 & 0.000 & 1.695 & 4.660 \\
    CMKT       & 0.011 & 0.098 & 0.098 & 0.000 & -0.353 & 0.311 \\
    CSIZE      & 0.018 & 0.066 & 0.066 & 0.000 & -0.147 & 0.305 \\
    CMOM       & 0.015 & 0.066 & 0.066 & 0.000 & -0.179 & 0.381 \\
    SMKT       & 0.001 & 0.038 & 0.038 & 0.000 & -0.141 & 0.125 \\
    SVLT       & 24.899 & 8.948 & 8.948 & 0.000 & 12.100 & 82.690 \\
    BMKT       & -0.001 & 0.010 & 0.010 & 0.000 & -0.033 & 0.038 \\
    BVLT       & 78.747 & 32.810 & 32.810 & 0.000 & 36.970 & 156.950 \\
    GMKT       & 0.001 & 0.022 & 0.022 & 0.000 & -0.090 & 0.077 \\
    GVLT       & 18.950 & 4.597 & 4.597 & 0.000 & 10.910 & 41.430 \\
    \bottomrule
  \end{tabular}
  }
  \vspace{0.5em}
  \begin{minipage}{\textwidth}
    \footnotesize \textit{\textbf{Notes}:} Standard deviations are reported separately as overall, within, and between. All variables are expressed in levels. ILQ values are divided by $10^6$ for ease of interpretation.
  \end{minipage}
  \label{tab:summary_stats}
\end{table}

The full set of indicators used in Stage 2, together with definitions, sources, and references, is reported in Table~\ref{tab:meta-variables}.

Table \ref{tab:summary_stats_macro} presents summary statistics for the macro-wide variables selected in the second stage. Since these variables are common across commodities at each point in time, there is no between-commodity variation. Consequently, the within-commodity standard deviation coincides with the overall standard deviation.

\begin{table}[H]
  \centering

  \caption{Summary Statistics for Selected Macro-Financial Risk Variables}
  \label{tab:summary_stats_macro}

  \renewcommand{\arraystretch}{0.95}
  \scalebox{0.90}{
  \begin{tabular}{lcccc}
    \toprule
    Variable & Mean & Standard Deviation & Min & Max \\
    \midrule

    DTWEXBGS
      & 115.7161 & 6.3321 & 94.4754 & 129.2775 \\

    MOVE
      & 78.4907 & 27.3547 & 0.0000 & 169.6500 \\

    UMCSENT
      & 80.2678 & 15.8370 & 49.8000 & 101.4000 \\

    EXPINF5YR
      & 1.8721 & 0.3926 & 0.8649 & 2.5514 \\

    VIX
      & 18.1777 & 6.7972 & 9.1500 & 76.4500 \\

    GPR Oil
      & 149.0457 & 198.6463 & 15.3200 & 1844.4300 \\

    \bottomrule
  \end{tabular}
  }

  \vspace{0.5em}
  \begin{minipage}{0.95\textwidth}
    \footnotesize
    \textit{\textbf{Notes}:}
    The table reports summary statistics for the macro-financial variables selected in the second stage. Since these variables are common across commodities at each point in time, there is no between-commodity variation.
    Consequently, the within-commodity standard deviation coincides with the overall standard deviation, and we report only the latter. DTWEXBGS denotes the broad U.S.\ dollar index, MOVE the bond market volatility index, UMCSENT the University of Michigan consumer sentiment index, EXPINF5YR five-year expected inflation, VIX the equity market volatility index, and GPR Oil the oil-related geopolitical risk index. The sample contains 52 commodities observed over 611 weeks, for a total of 31,772 observations.
  \end{minipage}
\end{table}

\begin{landscape}
{\footnotesize	
\centering
\begin{longtable}{p{5cm}p{10cm}p{4cm}}%

\caption{Full Set of Candidate Macro-Financial Risk Components} \\
\toprule
Name  & Short description & Source \\
\midrule
\endfirsthead

\multicolumn{3}{c}%
{{ \tablename \thetable{} -- continued from previous page}} \\
\toprule  Name  & Short description & Source \\
\midrule
\endhead
\hline \multicolumn{3}{r}{{Continued on next page}} \\ 
\hline
\endfoot

\hline \hline
\endlastfoot
    
    University of Michigan consumer sentiment & Monthly survey that gauges how confident U.S. consumers feel about their personal finances and the overall economy, indexed to 1966:Q1 $=100$ & \cite{ConsumerMichigan} \\
    Expected inflation (1-year, 3-year, 5-year) & Estimates calculated with a model that uses Treasury yields, inflation data, inflation swaps, and survey-based measures of inflation expectations & \cite{InflationExpectation} \\
    Equity market volatility tracker (EMV) & Newspaper-based tracker of equity market volatility that moves with the VIX and with the realized volatility of returns on the S\&P 500 & \cite{NBERw25720} \\
    Global economic policy uncertainty index (GEPU, current-price and PPP-adjusted weights) & GDP-weighted average of national EPU indices, based on the share of own-country newspaper articles that discuss economic policy uncertainty each month & \cite{Davis2016} \\
    Trade policy uncertainty index (TPU) & Index of uncertainty about trade policy, constructed from the frequency of newspaper articles discussing trade policy in uncertain terms & \cite{CaldaraEtAl2020} \\
    Monetary policy uncertainty index (MPU) & Index constructed as scaled frequency counts of newspaper articles that discuss monetary policy uncertainty & \cite{BakerBloomDavis2016} \\
    Climate policy uncertainty index (ClimPolicyU) & Index based on searches for articles in leading U.S. newspapers containing terms related to uncertainty, climate, environment, CO\textsubscript{2}, and analogous terms & \cite{GavriilidisEtAl2026} \\
    World uncertainty index family (WUI) & Computed by counting the frequency of the word ``uncertain'' (or its variants) in the Economist Intelligence Unit country reports, rescaled by multiplying by 1,000,000; global, trade, and pandemic variants & \cite{AhirBloomFurceri2022} \\
    ESG-based sustainability uncertainty index (ESGUI) & Index capturing uncertainty surrounding environmental, social, and governance conditions relevant to sustainable investment & \cite{OnganEtAl2025} \\
    Geopolitical risk index family (GPR, threats, acts, oil-related) & Measure of adverse geopolitical events based on a tally of newspaper articles covering geopolitical tensions, decomposed into realised acts and prospective threats, and into oil-related tensions & \cite{CaldaraIacoviello2022} \\
    Kansas City financial stress index (KCFSI) & Monthly index of financial stress constructed from eleven financial market variables capturing yield spreads, asset price volatility, and flight to quality & \cite{HakkioKeeton2009} \\
    Nominal broad U.S. dollar index & Measures the value of the U.S. dollar against a broad basket of currencies of major U.S. trading partners, weighted by their bilateral trade with the U.S. & \cite{NominalUSDindex} \\
    Aruoba-Diebold-Scotti business conditions index (ADS) & Real-time index of U.S. business conditions extracted from a dynamic factor model combining macroeconomic indicators observed at mixed frequencies & \cite{ArubaDieboldScotti2009} \\
    Oil price uncertainty index & Index capturing uncertainty specific to crude oil price formation, distinct from general economic and energy policy uncertainty & \cite{Abiad_Oil_uncertainty} \\
    Energy transportation uncertainty index & Newspaper-based index, developed from over fifty global newspapers, tracking uncertainty associated with disruptions to energy transportation routes and infrastructure & \cite{Morao2025} \\
    Climate concern index (CCI) & Index of collective climate perceptions based on disaggregated web-search volumes related to climate change topics & \cite{AngeliniEtAl2025} \\
    Macroeconomic uncertainty indices (MacroU, 1-month, 3-month, and 12-month) & Factor analysis based on 132 mostly macroeconomic series & \cite{JuradoLudvigsonNg2015} \\
    Financial uncertainty indices (FinancialU, 1-month, 3-month, and 12-month) & Factor analysis based on 147 financial time series & \cite{JuradoLudvigsonNg2015} \\
    Real uncertainty indices (RealU, 1-month, 3-month, and 12-month) & Factor analysis based on the real-activity subset of the underlying series & \cite{JuradoLudvigsonNg2015} \\
        CBOE volatility index (VIX) & Option-implied volatility of the S\&P 500 over the next 30 days, computed from a model-free portfolio of index option prices, and widely used as a barometer of equity market uncertainty & \cite{CBOEVIX} \\
    ICE BofA MOVE index & Option-implied volatility of U.S. Treasury yields over the next 30 days, computed as a yield-curve-weighted average of at-the-money swaption implied volatilities, capturing interest rate and funding uncertainty & \cite{ICEMOVE} \\
  \label{tab:meta-variables}%
\end{longtable}%
}
\end{landscape}

\section{Robustness of Estimated Sensitivities\label{sec:Appendix_Robustness}}

This section assesses the robustness of the estimated sensitivities in both stages of the empirical framework. For the micro and market sensitivities in Stage 1, we first examine their stability across the 2014--2019 and 2020--2026 subsamples and then replace Mean Group aggregation with the Median Group estimator, which is less sensitive to extreme commodity-specific coefficient estimates. We subsequently examine the robustness of the economy-wide sensitivities in Stage 2 by comparing Mean and Median Group estimates and by distinguishing between US dollar-denominated and non-US dollar-denominated commodities. Tables~\ref{tab:robustness} and~\ref{tab:robustness_dollar} report the results.

Columns (1)--(3) examine stability across sample periods. Column (1) reproduces the full-sample Mean Group IV estimates, while Columns (2) and
(3) report estimates for 2014--2019 and 2020--2026, respectively. The main relationships are similar across the two periods. Own skewness remains
positive and highly significant, market skewness remains negative and highly significant, and the leave-one-out sector return, $r_{-i,t}$, is positive and significant at the 1\% level in both subsamples.

Some differences in magnitude are nevertheless apparent. The coefficient on $r_{-i,t}$ is 0.264 in 2014--2019 and 0.510 in 2020--2026, indicating stronger commodity-market comovement in the latter period. Market kurtosis is also larger after 2020, while market open interest is insignificant in the earlier period and positive and significant in the later period. Own volatility and own open interest are statistically insignificant in both subsamples. Thus, although the magnitude of some sensitivities varies across periods, the main sign and significance patterns underlying the baseline results are preserved.

Columns (4)--(7) replace the cross-sectional mean of the commodity-specific coefficient estimates with their median. For each coefficient, we compute
\[
\widehat{\boldsymbol{\theta}}_{\mathrm{Median}}
=
\operatorname{median}_{i=1,\ldots,N}
\left(\widehat{\boldsymbol{\theta}}_i\right).
\]
Median Group aggregation is less sensitive to unusually large or small commodity-specific estimates. Standard errors are obtained using the
cross-sectional bootstrap procedure of \citet{Kapetanios2008}.

The full-sample Median Group estimates in Column (4) are close to their Mean Group counterparts for the main relationships. Own skewness is positive, $r_{-i,t}$ is positive, market skewness is negative, and market kurtosis is positive, with all four coefficients highly significant. Their Median Group estimates are 0.020, 0.290, $-0.063$, and 0.007, respectively. Market volatility becomes positive and statistically significant under Median Group aggregation, but the estimate remains relatively small.

The sector-specific Median Group estimates in Columns (5)--(7) provide a similar picture. The coefficient on $r_{-i,t}$ is positive and highly
significant in all three sectors and is largest for energy, at 1.253, compared with 0.206 for agriculture and 0.387 for metals. Own skewness remains positive, market skewness negative, and market kurtosis positive across the three sectors. The negative coefficient on own volatility is again concentrated in energy, where the Median Group estimate is $-0.128$ and statistically significant.

The robustness exercises therefore leave the main empirical findings largely unchanged. The key micro and market relationships are preserved across the two sample periods and under Median Group aggregation. In particular, the strong market sensitivity of energy commodities remains when the cross-sectional median replaces the mean, indicating that this result is not driven by a small number of extreme commodity-specific estimates.

\begin{table}[!htbp]
  \centering
  \caption{Robustness of 1st Stage Group-Level Estimates}
  \label{tab:robustness}
  \footnotesize
  \setlength{\tabcolsep}{3.2pt}
  \renewcommand{\arraystretch}{0.95}
  \resizebox{\textwidth}{!}{
    \begin{tabular}{lccccccc}
    \toprule
    & \multicolumn{3}{c}{Mean Group}
    & \multicolumn{4}{c}{Median Group} \\
    \cmidrule(lr){2-4} \cmidrule(lr){5-8}
    
    & (1) & (2) & (3) & (4) & (5) & (6) & (7) \\
    & Full Sample & 2014--2019 & 2020--2026
    & Full Sample & Agriculture & Energy & Metals \\
    \midrule

VLTY
    & -0.047$^{**}$
    & -0.023
    & -0.039
    & -0.033
    & -0.035
    & -0.128$^{**}$
    & 0.026 \\
    & (0.022)
    & (0.025)
    & (0.030)
    & (0.032)
    & (0.053)
    & (0.060)
    & (0.053) \\

SKEW
    & 0.021$^{***}$
    & 0.020$^{***}$
    & 0.023$^{***}$
    & 0.020$^{***}$
    & 0.021$^{***}$
    & 0.024$^{***}$
    & 0.017$^{***}$ \\
    & (0.001)
    & (0.001)
    & (0.001)
    & (0.001)
    & (0.002)
    & (0.003)
    & (0.001) \\

KTS
    & -0.000
    & -0.000
    & -0.001$^{*}$
    & -0.000$^{*}$
    & -0.001$^{**}$
    & -0.000
    & -0.000 \\
    & (0.000)
    & (0.000)
    & (0.000)
    & (0.000)
    & (0.000)
    & (0.000)
    & (0.000) \\

OI
    & 0.000
    & -0.000
    & -0.000
    & 0.000
    & 0.000
    & 0.000
    & 0.000 \\
    & (0.000)
    & (0.000)
    & (0.000)
    & (0.000)
    & (0.000)
    & (0.000)
    & (0.000) \\[0.2em]

$r_{-i}$
    & 0.509$^{***}$
    & 0.264$^{***}$
    & 0.510$^{***}$
    & 0.290$^{***}$
    & 0.206$^{***}$
    & 1.253$^{***}$
    & 0.387$^{***}$ \\
    & (0.071)
    & (0.079)
    & (0.076)
    & (0.048)
    & (0.042)
    & (0.309)
    & (0.108) \\

CVLTY
    & 0.047
    & -0.021
    & 0.055
    & 0.046$^{**}$
    & 0.041
    & 0.119$^{*}$
    & -0.037 \\
    & (0.030)
    & (0.025)
    & (0.045)
    & (0.019)
    & (0.037)
    & (0.070)
    & (0.056) \\

CSKEW
    & -0.073$^{***}$
    & -0.060$^{***}$
    & -0.074$^{***}$
    & -0.063$^{***}$
    & -0.065$^{***}$
    & -0.083$^{***}$
    & -0.056$^{***}$ \\
    & (0.005)
    & (0.004)
    & (0.007)
    & (0.008)
    & (0.011)
    & (0.011)
    & (0.005) \\

CKTS
    & 0.008$^{***}$
    & 0.001$^{*}$
    & 0.013$^{***}$
    & 0.007$^{***}$
    & 0.007$^{***}$
    & 0.008$^{***}$
    & 0.006$^{***}$ \\
    & (0.001)
    & (0.000)
    & (0.001)
    & (0.001)
    & (0.001)
    & (0.001)
    & (0.001) \\

COI
    & 0.002$^{***}$
    & 0.000
    & 0.003$^{***}$
    & 0.001$^{**}$
    & 0.003$^{**}$
    & 0.001
    & 0.000 \\
    & (0.001)
    & (0.001)
    & (0.001)
    & (0.001)
    & (0.001)
    & (0.001)
    & (0.001) \\

    \midrule
Commodities
    & 52 & 52 & 52 & 52 & 24 & 14 & 14 \\

Observations
    & 31,616 & 14,404 & 17,212 &  31,616 & 14,592 & 8,512 & 8,512 \\

    \bottomrule
  \end{tabular}
  }

  \vspace{0.5em}
  \begin{minipage}{0.99\textwidth}
    \footnotesize
    \textit{\textbf{Notes}:}
    Standard errors are reported in parentheses.
    Significance levels are $^{*}$ 10\%, $^{**}$ 5\%, and $^{***}$ 1\%.
    Column (1) reports the baseline Mean Group IV estimates for the full sample.
    Columns (2) and (3) re-estimate the same specification over 2014--2019 and
    2020--2026, respectively. Columns (4)--(7) replace the cross-sectional mean
    with the Median Group estimator, reported for the full sample and separately
    for agricultural, energy, and metal commodities. Standard errors for the
    Median Group estimates are obtained using the cross-sectional bootstrap
    procedure of \citet{Kapetanios2008}.
    OI and COI denote logged open interest and its leave-one-out market
    counterpart, respectively.
  \end{minipage}
\end{table}

Table~\ref{tab:robustness_dollar} examines the robustness of the Stage 2 economy-wide sensitivities. Columns (1)--(3) report Mean Group estimates for the full sample and separately for US dollar-denominated and non-US dollar-denominated commodities, while Columns (4)--(6) report the corresponding Median Group estimates. This distinction is particularly relevant for the sensitivity to the broad US dollar index, since its estimated effect could potentially depend on the currency in which the underlying futures contract is denominated.

The results are remarkably stable across both dimensions. All six macro-financial variables retain the same signs as in the baseline specification and remain statistically significant at the 1\% level under both Mean and Median Group aggregation and for both currency-denomination groups. In particular, the sensitivity to the broad US dollar index remains positive for non-US dollar-denominated commodities as well as for US dollar-denominated commodities. Under Mean Group aggregation, the estimates are 0.0026 and 0.0033, respectively, compared with 0.0028 for the full sample; the corresponding Median Group estimates are 0.0022 and 0.0035. The positive dollar sensitivity identified in the baseline results is therefore not specific to commodities whose futures contracts are denominated in US dollars.

The remaining economy-wide sensitivities display similarly stable patterns. MOVE, consumer sentiment, five-year expected inflation, and VIX retain negative sensitivities across all specifications, whereas oil-related geopolitical risk retains a positive sensitivity. The magnitudes are also generally similar under Mean and Median Group aggregation, although the absolute sensitivities to expected inflation and VIX are somewhat larger among non-US dollar-denominated commodities. Therefore, the Stage 2 results appear robust both to the method used to aggregate heterogeneous commodity-specific sensitivities and to the currency denomination of the underlying futures contracts.

\begin{table}[!htbp]
  \centering
  \caption{Robustness of 2nd Stage Group-Level Estimates}
  \label{tab:robustness_dollar}
  \footnotesize
  \setlength{\tabcolsep}{3.2pt}
  \renewcommand{\arraystretch}{0.95}
  \resizebox{\textwidth}{!}{
    \begin{tabular}{lcccccc}
    \toprule
    & \multicolumn{3}{c}{Mean Group}
    & \multicolumn{3}{c}{Median Group} \\
    \cmidrule(lr){2-4} \cmidrule(lr){5-7}
    
    & (1) & (2) & (3) & (4) & (5) & (6) \\
    & Full Sample
    & \shortstack{US Dollar-\\Denominated}
    & \shortstack{Non-US Dollar-\\Denominated}
    & Full Sample
    & \shortstack{US Dollar-\\Denominated}
    & \shortstack{Non-US Dollar-\\Denominated} \\
    \midrule

DTWEXBGS
    & 0.0028$^{***}$
    & 0.0026$^{***}$
    & 0.0033$^{***}$
    & 0.0025$^{***}$
    & 0.0022$^{***}$
    & 0.0035$^{***}$ \\
    & (0.0002)
    & (0.0002)
    & (0.0003)
    & (0.0002)
    & (0.0002)
    & (0.0003) \\

MOVE
    & -0.0006$^{***}$
    & -0.0006$^{***}$
    & -0.0007$^{***}$
    & -0.0006$^{***}$
    & -0.0005$^{***}$
    & -0.0007$^{***}$ \\
    & (0.0000)
    & (0.0000)
    & (0.0001)
    & (0.0001)
    & (0.0000)
    & (0.0001) \\

UMCSENT
    & -0.0009$^{***}$
    & -0.0009$^{***}$
    & -0.0011$^{***}$
    & -0.0008$^{***}$
    & -0.0008$^{***}$
    & -0.0012$^{***}$ \\
    & (0.0001)
    & (0.0001)
    & (0.0001)
    & (0.0001)
    & (0.0001)
    & (0.0001) \\

EXPINF5YR
    & -0.0372$^{***}$
    & -0.0344$^{***}$
    & -0.0455$^{***}$
    & -0.0324$^{***}$
    & -0.0286$^{***}$
    & -0.0467$^{***}$ \\
    & (0.0026)
    & (0.0032)
    & (0.0034)
    & (0.0031)
    & (0.0022)
    & (0.0047) \\

VIX
    & -0.0026$^{***}$
    & -0.0024$^{***}$
    & -0.0032$^{***}$
    & -0.0022$^{***}$
    & -0.0021$^{***}$
    & -0.0035$^{***}$ \\
    & (0.0002)
    & (0.0002)
    & (0.0003)
    & (0.0003)
    & (0.0001)
    & (0.0003) \\

GPR Oil
    & 0.0000$^{***}$
    & 0.0000$^{***}$
    & 0.0000$^{***}$
    & 0.0000$^{***}$
    & 0.0000$^{***}$
    & 0.0000$^{***}$ \\
    & (0.0000)
    & (0.0000)
    & (0.0000)
    & (0.0000)
    & (0.0000)
    & (0.0000) \\

    \midrule
Commodities
    & 52 & 39 & 13 & 52 & 39 & 13 \\

    \bottomrule
  \end{tabular}
  }

  \vspace{0.5em}
  \begin{minipage}{0.99\textwidth}
    \footnotesize
    \textit{\textbf{Notes}:}
    Standard errors are reported in parentheses.
    Significance levels are $^{*}$ 10\%, $^{**}$ 5\%, and $^{***}$ 1\%.     Columns (1)--(3) report Mean Group estimates for the full sample, US dollar-denominated commodities, and non-US dollar-denominated commodities, respectively. Columns (4)--(6) report the corresponding
    Median Group estimates. Standard errors for the Median Group estimates are obtained using the cross-sectional bootstrap procedure of \citet{Kapetanios2008}.
  \end{minipage}
\end{table}

\section{Key Steps of the Divide-and-Conquer Framework
\label{sec:Appendix_DivideConquer}}
\noindent
\textbf{Stage 1: Observed Micro and Market Risk}

\begin{itemize}
  \renewcommand\labelitemi{---}

  \item Let $\mathcal S_1,\ldots,\mathcal S_J$ denote the commodity sectors   and define the sector-average returns
  \[
  \bar{\mathbf r}_s
  =
  \frac{1}{N_s}
  \sum_{j\in\mathcal S_s}\mathbf r_j,
  \qquad
  s=1,\ldots,J.
  \]
  Collect these returns in $\mathbf R   =   \left(
  \bar{\mathbf r}_1,\ldots,\bar{\mathbf r}_J
  \right)$, and define $\mathbf M_{\mathbf R}   =   \mathbf I_T   -   \mathbf R   (\mathbf R'\mathbf R)^{-1}   \mathbf R'$.  Rather than partialling out a different leave-one-out sector return for
  each commodity, project all commodity-specific risk characteristics onto
  the orthogonal complement of the same observed sector-return space. For
  $i\in\mathcal S_s$,
  \[
  \mathbf M_{\mathbf R}\mathbf r_{-i}
  =
  -\frac{1}{N_s-1}
  \mathbf M_{\mathbf R}\mathbf r_i,
  \]
  so that the residual contribution of the leave-one-out sector return is asymptotically negligible as $N_s\to\infty$.

  \item For lagged characteristics, use the correspondingly lagged sector-return space. For $\tau=0,\ldots,\zeta_x$, define $\mathbf R_{-\tau}
  =   \left(   \bar{\mathbf r}_{1,-\tau},\ldots,
  \bar{\mathbf r}_{J,-\tau} \right)$ and $\mathbf M_{\mathbf R_{-\tau}}   =   \mathbf I_T   -
  \mathbf R_{-\tau}   \left(   \mathbf R_{-\tau}'\mathbf R_{-\tau}   \right)^{-1}   \mathbf R_{-\tau}'$. The latent factor component $\mathbf F^0_{-\tau}$ is defined as the
  pervasive common variation in the commodity-specific risk characteristics
  remaining after accounting for the observed sector-return space. Thus, for $i\in\mathcal S_s$,
  \[
  \mathbf M_{\mathbf R_{-\tau}}\mathbf X_{i,-\tau}
  =
  \mathbf F^0_{-\tau}\boldsymbol{\Lambda}_i'
  +
  \mathbf M_{\mathbf R_{-\tau}}\mathbf V_{i,-\tau}
  +
  \mathbf D_{i,-\tau},
  \]
  where $\mathbf D_{i,-\tau}   =   \mathbf M_{\mathbf R_{-\tau}}   \mathbf r_{-i,-\tau}\boldsymbol{\phi}_i'$ is asymptotically negligible as the sector size increases.

  \item Estimate the latent factor space separately for each   $\tau=0,\ldots,\zeta_x$ by applying PCA to the projected   commodity-specific risk characteristics. The estimated factor matrices   $\widehat{\mathbf F}_{-\tau}$ are obtained, up to rotation, from the
  leading eigenvectors of
  \[
  \frac{1}{NT}
  \sum_{i=1}^{N}
  \mathbf M_{\mathbf R_{-\tau}}
  \mathbf X_{i,-\tau}
  \mathbf X_{i,-\tau}'
  \mathbf M_{\mathbf R_{-\tau}}.
  \]
  Because the same sector-return space is removed from every commodity at a given lag, PCA can exploit the full cross-section while targeting the latent common variation that is incremental to observed sector-return movements.

  \item The factor space extracted from the commodity-specific risk characteristics need not coincide with the latent macro-risk space. Stage~1 requires only the spanning condition
  \[
  \mathrm{col}(\mathbf G)
  \subseteq
  \mathrm{col}(\mathbf F^0).
  \]
  Thus, $\mathbf F^0$ may contain common commodity-market variation beyond
  the latent macro-financial component $\mathbf G$, while being sufficiently
  rich to span the latter. Under this condition, defactoring the structural
  equation removes the latent macro component responsible for the
  endogeneity of the observed risk characteristics.

  \item Eliminate the influence of the estimated latent factor space from the entire model using
  \[
  \mathbf M_{\widehat{\mathbf F}}\mathbf r_i
  =
  \mathbf M_{\widehat{\mathbf F}}\mathbf r_{-i}\rho_i
  +
  \mathbf M_{\widehat{\mathbf F}}\mathbf X_i\boldsymbol{\alpha}_i
  +
  \mathbf M_{\widehat{\mathbf F}}\mathbf Y_i\boldsymbol{\beta}_i
  +
  \mathbf M_{\widehat{\mathbf F}}\mathbf u_i
  =
  \mathbf M_{\widehat{\mathbf F}}\mathbf C_i
  \boldsymbol{\theta}_i
  +
  \mathbf M_{\widehat{\mathbf F}}\mathbf u_i,
  \]
  where
  \[
  \boldsymbol{\theta}_i
  =
  (\rho_i,\boldsymbol{\alpha}_i',\boldsymbol{\beta}_i')'.
  \]

  \item Construct the instrument matrix
  \begin{align*}
  \widehat{\mathbf Z}_i
  =
  \Big(
  &
  \mathbf r_{-i},\;
  \mathbf r_{-i,-1},\ldots,
  \mathbf r_{-i,-\zeta_r},
  \nonumber\\
  &
  \mathbf M_{\widehat{\mathbf F}}\mathbf X_i,\;
  \mathbf M_{\widehat{\mathbf F}_{-1}}\mathbf X_{i,-1},
  \ldots,
  \mathbf M_{\widehat{\mathbf F}_{-\zeta_x}}
  \mathbf X_{i,-\zeta_x},
  \nonumber\\
  &
  \mathbf M_{\widehat{\mathbf F}}\mathbf Y_i,\;
  \mathbf M_{\widehat{\mathbf F}_{-1}}\mathbf Y_{i,-1},
  \ldots,
  \mathbf M_{\widehat{\mathbf F}_{-\zeta_y}}
  \mathbf Y_{i,-\zeta_y}
  \Big).
  \end{align*}
  The sector-return projection is used only for factor extraction and does not enter the final instrument transformation. The micro- and market-risk characteristics are used as instruments after removing their estimated latent common factor component, while the leave-one-out sector return and
  its selected lags enter the instrument set directly under the maintained orthogonality conditions.

  \item Estimate the commodity-specific coefficient vector $\boldsymbol{\theta}_i$ by IV using the defactored estimating equation and the instrument set $\widehat{\mathbf Z}_i$. This yields heterogeneous
  sensitivities of commodity returns to the observed sector-return, micro-risk, and market-risk sources. Commodity-specific over-identification tests based on the defactored moment conditions provide a diagnostic for instrument validity and, indirectly, for the
  maintained spanning condition.

\end{itemize}

\medskip
\noindent
\textbf{Stage 2: Latent Economy-Wide Common Risk Sources}

\begin{itemize}
  \renewcommand\labelitemi{---}

  \item Form the Stage~1 IV residuals $\widehat{\mathbf u}_i   =   \mathbf r_i   -   \mathbf C_i\widehat{\boldsymbol{\theta}}_i$. Under model~\eqref{model_with_u}, these residuals consistently estimate $\mathbf u_i   =   \mathbf G\boldsymbol{\gamma}_i   +  \boldsymbol{\varepsilon}_i$. Hence, after controlling for the observed sector-return, micro-risk, and
  market-risk sources, their common component is associated with the latent
  economy-wide macro-risk space $\mathbf G$.

  \item Estimate the dimension $r$ of this common residual variation using eigenvalue-based methods or information criteria, as in \citet{BaiNg2002} and \citet{AhnHorenstein2013}, and extract the   corresponding principal components from $
  \frac{1}{NT}   \sum_{i=1}^{N}   \widehat{\mathbf u}_i
  \widehat{\mathbf u}_i'$. Under the maintained factor structure, these principal components recover, up to rotation, the common residual space generated by $\mathbf G$. This PCA step is distinct from that in Stage~1: Stage~1 estimates a factor space from the commodity-specific risk characteristics to address endogeneity, whereas Stage~2 extracts the common component of the IV residuals in order to identify observable proxies for latent macro risk.

  \item For each principal component $j=1,\ldots,r$, apply the Boosting with Multiple Testing (BMT) procedure of \citet{KapetaniosEtal2026} to the high-dimensional set of candidate macro-financial variables
  $\boldsymbol{\mathcal Z}$. Let $\mathcal S_j$ denote the variables selected for component $j$ and define
  \[
  \mathcal S
  =
  \bigcup_{j=1}^{r}\mathcal S_j.
  \]
  The corresponding design matrix   $\boldsymbol{\mathcal Z}_{(\mathcal S)}$ provides a sparse set of observable macro-financial proxies for the latent macro-risk space.

  \item Estimate commodity-specific sensitivities to the selected macro-financial proxies by regressing
  $\widehat{\mathbf u}_i$ on $\boldsymbol{\mathcal Z}_{(\mathcal S)}$. Denote the resulting
  coefficient vector by $\widehat{\boldsymbol{\gamma}}_i$. Performing
  selection on the common residual components rather than separately for each commodity ensures that all commodities are subsequently evaluated against the same selected set of macro-financial proxies while allowing their sensitivities to remain heterogeneous.

\end{itemize}

Finally, the commodity-specific estimates
$\widehat{\boldsymbol{\theta}}_i$ and
$\widehat{\boldsymbol{\gamma}}_i$ are aggregated using Mean Group
estimators, both for the full cross-section and separately for agriculture,
energy, and metals, to characterise heterogeneity in return sensitivities
across the three risk layers.

Stage~1 extends the defactored IV approach of
\citet{NorkuteEtal2021} and \citet{CuiEtal2022} by accommodating observed sector-return, micro-risk, and market-risk channels alongside latent common variation. Projecting the commodity-specific risk characteristics on the orthogonal complement of the joint observed sector-return space before factor extraction separates observed sector-return-related common variation from the latent factor space used in the defactoring step. Importantly, this projection is common across commodities and is used for factor extraction only; the leave-one-out sector return remains in the structural equation and enters the final instrument set directly. Stage~2 then recovers the common variation remaining in the Stage~1 IV residuals and links it to observable
macro-financial proxies by combining PCA with the high-dimensional BMT procedure of \citet{KapetaniosEtal2026}.

\section{Proofs of Theoretical Results}
\label{reg_conditions_lemmas}
We start by collecting auxiliary moment, dependence, and boundedness conditions that support the asymptotic arguments in Section~3. These conditions are standard in the interactive fixed effects and factor-augmented IV literature and are imposed to ensure uniform laws of large numbers and central limit results.

Let $\|\mathbf{A}\|=\sqrt{\mathrm{tr}(\mathbf{A}'\mathbf{A})}$ denote the Frobenius norm, $\|\mathbf A\|_2$ the spectral (operator) norm, and let $\Delta<\infty$ be a generic constant whose value may change across expressions. Throughout,
\[
\delta_{NT}=\min\{\sqrt N,\sqrt T\}.
\]

\paragraph{Assumption A.0 (Sector-return projection and factor identification).}

Let $\mathcal S_1,\ldots,\mathcal S_J$ denote the commodity sectors,
where $J$ is fixed, $N_s=|\mathcal S_s|$, and
\[
\frac{N_s}{N}\longrightarrow\pi_s,
\qquad
0<\pi_s<1,
\qquad
s=1,\ldots,J.
\]
For each sector define
\[
\bar{\mathbf r}_s
=
\frac{1}{N_s}
\sum_{j\in\mathcal S_s}\mathbf r_j,
\]
and, for $i\in\mathcal S_s$,
\[
\mathbf r_{-i}
=
\frac{1}{N_s-1}
\sum_{\substack{j\in\mathcal S_s\\j\neq i}}
\mathbf r_j.
\]

For $\tau=0,\ldots,\max\{\zeta_x,\zeta_y\}$ let
\[
\mathbf R_{-\tau}
=
\left(
\bar{\mathbf r}_{1,-\tau},\ldots,
\bar{\mathbf r}_{J,-\tau}
\right)
\]
and define
\[
\mathbf M_{\mathbf R_{-\tau}}
=
\mathbf I_T
-
\mathbf R_{-\tau}
\left(
\mathbf R_{-\tau}'\mathbf R_{-\tau}
\right)^{-1}
\mathbf R_{-\tau}'.
\]
For $\tau=0$, write
$\mathbf R_{-0}=\mathbf R$ and
$\mathbf M_{\mathbf R_{-0}}=\mathbf M_{\mathbf R}$.

Assume
\[
\max_{1\leq i\leq N}
\frac{\|\mathbf r_{i,-\tau}\|}{\sqrt T}
=
O_p(1)
\]
uniformly over the finite set $\tau=0,\ldots,\max\{\zeta_x,\zeta_y\}$, and
\[
\lambda_{\min}
\left(
\frac{\mathbf R_{-\tau}'\mathbf R_{-\tau}}{T}
\right)
\geq c_R>0
\]
with probability approaching one, uniformly over
$\tau=0,\ldots,\max\{\zeta_x,\zeta_y\}$.

The latent factor component $\mathbf F^0_{-\tau}$ is identified as the
pervasive common component remaining after accounting for the observed
sector-return space. Accordingly, we impose the normalization
\[
\mathbf R_{-\tau}'\mathbf F^0_{-\tau}
=
\mathbf 0,
\qquad
\tau=0,\ldots,\max\{\zeta_x,\zeta_y\},
\]
so that
\[
\mathbf M_{\mathbf R_{-\tau}}\mathbf F^0_{-\tau}
=
\mathbf F^0_{-\tau}.
\]
This normalization identifies the latent factor component relative to the
observed sector-return space and does not require the economic forces
underlying the two components to be mutually independent.

\begin{lemma}[Residual leave-one-out sector-return component]
\label{LEM_MiMr}
Under Assumption~A.0, for every $i\in\mathcal S_s$ and
$\tau=0,\ldots,\max\{\zeta_x,\zeta_y\}$,
\begin{equation}
\mathbf M_{\mathbf R_{-\tau}}
\mathbf r_{-i,-\tau}
=
-\frac{1}{N_s-1}
\mathbf M_{\mathbf R_{-\tau}}
\mathbf r_{i,-\tau}.
\label{eq:sector_projection_identity}
\end{equation}
Consequently,
\begin{equation}
\max_{1\leq i\leq N}
\frac{
\left\|
\mathbf M_{\mathbf R_{-\tau}}
\mathbf r_{-i,-\tau}
\right\|
}{\sqrt T}
=
O_p(N^{-1})
=
O_p(\delta_{NT}^{-2}),
\label{eq:sector_projection_rate}
\end{equation}
uniformly over $\tau=0,\ldots,\max\{\zeta_x,\zeta_y\}$, when
$N/T\to c\in(0,\infty)$.
\end{lemma}

\begin{proof}
For $i\in\mathcal S_s$,
\[
\mathbf r_{-i,-\tau}
=
\frac{N_s}{N_s-1}\bar{\mathbf r}_{s,-\tau}
-
\frac{1}{N_s-1}\mathbf r_{i,-\tau}.
\]
Since $\bar{\mathbf r}_{s,-\tau}$ is a column of
$\mathbf R_{-\tau}$,
\[
\mathbf M_{\mathbf R_{-\tau}}
\bar{\mathbf r}_{s,-\tau}
=
\mathbf 0.
\]
Hence
\[
\mathbf M_{\mathbf R_{-\tau}}
\mathbf r_{-i,-\tau}
=
-\frac{1}{N_s-1}
\mathbf M_{\mathbf R_{-\tau}}
\mathbf r_{i,-\tau},
\]
which establishes Eq.~\eqref{eq:sector_projection_identity}.

Because an orthogonal projection has operator norm at most one,
\[
\frac{
\left\|
\mathbf M_{\mathbf R_{-\tau}}
\mathbf r_{-i,-\tau}
\right\|
}{\sqrt T}
\leq
\frac{1}{N_s-1}
\frac{\|\mathbf r_{i,-\tau}\|}{\sqrt T}
=
O_p(N_s^{-1}).
\]
Since $N_s/N\to\pi_s>0$ and $J$ is fixed,
$N_s^{-1}=O(N^{-1})$ uniformly across sectors. Therefore,
\[
\max_i
\frac{
\left\|
\mathbf M_{\mathbf R_{-\tau}}
\mathbf r_{-i,-\tau}
\right\|
}{\sqrt T}
=
O_p(N^{-1}).
\]
Finally, when $N/T\to c\in(0,\infty)$,
$N^{-1}=O(\delta_{NT}^{-2})$, which proves the result.
\end{proof}

\paragraph{Assumption A.1 (micro return errors)}
The disturbances $\varepsilon_{i,t}$ satisfy $E(\varepsilon_{i,t})=0$,
$E(\varepsilon_{i,t}^2)=\sigma^2_{\varepsilon,i,t}$, and
$E|\varepsilon_{i,t}|^{8+\delta}\le\Delta$ for some $\delta>0$.
Moreover, $\varepsilon_{i,t}$ are independent across $i$ and weakly dependent over $t$.

\paragraph{Assumption A.2 (Micro covariate errors)}
For each $k=1,\dots,K_x$, the processes $v_{k,i,t}$ satisfy:
\begin{enumerate}[label=(\roman*)]

\item $E(v_{k,i,t})=0$ and
$E|v_{k,i,t}|^{8+\delta}\leq\Delta$ for some $\delta>0$;

\item cross-sectional independence across $i$ and group-wise independence
from $\varepsilon_{i,t}$;

\item uniformly bounded temporal dependence, in the sense that, letting
\[
\Gamma_{k,i,T}(s,t)
=
\operatorname{Cov}(v_{k,i,s},v_{k,i,t}),
\]
we have
\[
\max_{1\leq s\leq T}
\sum_{t=1}^T
\left|
\Gamma_{k,i,T}(s,t)
\right|
\leq \Delta
\]
uniformly over $k$, $i$, and $T$;

\item bounded fourth moments of centered quadratic forms,
\[
E\Big\|
N^{-1/2}\sum_{i=1}^N
\big(
v_{k,i,s}v_{k,i,t}
-
E(v_{k,i,s}v_{k,i,t})
\big)
\Big\|^4
\leq \Delta;
\]

\item uniformly bounded aggregate covariances of fourth-order products.
\end{enumerate}

In addition, the corresponding moment, quadratic-form, and dependence bounds are assumed to hold for the projected processes
$ \widetilde{\mathbf v}_{k,i,-\tau}
= \mathbf M_{\mathbf R_{-\tau}}
\mathbf v_{k,i,-\tau}$, $\tau=0,\ldots,\max\{\zeta_x,\zeta_y\}$. In particular, letting
$ \widetilde{\Gamma}_{k,i,T}^{(\tau)}(s,t)
= \operatorname{Cov}
\left( \widetilde v_{k,i,-\tau,s},
\widetilde v_{k,i,-\tau,t} \right)$, we assume
\[
\max_{1\leq s\leq T}
\sum_{t=1}^T
\left|
\widetilde{\Gamma}_{k,i,T}^{(\tau)}(s,t)
\right|
\leq\Delta
\]
uniformly over $k$, $i$, $T$, and the finite set
$\tau=0,\ldots,\max\{\zeta_x,\zeta_y\}$.

\paragraph{Assumption A.3 (Random coefficient moments).}
The heterogeneous deviations $\boldsymbol{\eta}_i$ satisfy
$E\|\boldsymbol{\eta}_i\|^4\le\Delta$ and possess sub-exponential tails, i.e.,
\[
P(|\eta_{i,r}|>z)\le 2\exp\!\left(-\tfrac12\frac{z^2}{a+bz}\right),
\quad \forall z>0,
\]
for fixed constants $a,b>0$ and all relevant components $r$.

\paragraph{Assumption A.4 (Factor--covariate interactions).}

For $\tau=0,\ldots,\max\{\zeta_x,\zeta_y\}$, define $\widetilde{\mathbf V}_{i,-\tau} = \mathbf M_{\mathbf R_{-\tau}}\mathbf V_{i,-\tau}$. The following moment and interaction bounds hold uniformly over $i$ and the finite set of relevant lag orders:

\begin{enumerate}[label=(\roman*)]

\item
\[
E\left\|
T^{-1/2}
\widetilde{\mathbf V}_{i,-\tau}'
\mathbf F^0_{-\tau}
\right\|^4
\leq\Delta.
\]

\item
\[
E\left[
T^{-1/2}
\sum_{k=1}^{K_x}\sum_{t=1}^T
\left(
\widetilde v_{k,i,-\tau,t}^{\,2}
-
E(\widetilde v_{k,i,-\tau,t}^{\,2})
\right)
\right]^2
\leq\Delta.
\]

\item For $\tau=0,\ldots,\zeta_r$,
\[
\frac{
\left(
\widehat{\mathbf F}
-
\mathbf F^0\mathbf H_x
\right)'
\mathbf r_{-i,-\tau}
}{T}
=
O_p\left(\delta_{NT}^{-2}\right)
\]
uniformly over $i$ and the selected lag orders.

\end{enumerate}

Analogous moment and interaction bounds required below are assumed for
$\mathbf Y_i$ and its selected lags.

\paragraph{Assumption A.5 (Consistency of the long-run covariance estimator).}
Under the conditional framework of Assumption~5, the HAC estimator
$\widehat{\boldsymbol\Omega}_{iT}$ defined in Eq.~\eqref{hac_omega},
and computed using the baseline residuals
\[
\widehat{\mathbf u}_i
=
\mathbf r_i-\mathbf C_i\widehat{\boldsymbol\theta}_i,
\]
satisfies, for each fixed $i$,
\[
\widehat{\boldsymbol\Omega}_{iT}
\overset{p}{\longrightarrow}
\boldsymbol\Omega_i
\]
conditionally on $\mathcal H$, for almost every realization, as
$N,T\to\infty$ with $N/T\to c\in(0,\infty)$.

Assumption~A.5 is imposed separately from the covariance-summability condition in Assumption~5, which by itself does not ensure consistency of the feasible HAC estimator.

\begin{remark}
Assumptions A.1--A.4 provide high-level regularity conditions for the convergence of the relevant sample moments and the validity of the PCA-based factor approximations used in the commodity-specific IV estimation. Variants of these conditions are standard in panels with interactive effects and factor-augmented IV models; see, among others,
\citet{Bai2009,Pesaran2006,SarafidisWansbeek2012,MoonWeidner2015,JuodisSarafidis2022}. Assumption A.5 separately imposes consistency of the HAC estimator of the long-run covariance matrix required for inference.
\end{remark}

\begin{lemma}
\label{bailemma}
Under Assumptions 1--5 and A.0--A.4, as $N,T\to\infty$ such that
$N/T\to c\in(0,\infty)$, the PCA estimator based on
Eq.~\eqref{factor_covariance} satisfies, uniformly over the finite set $\tau=0,\ldots,\max\{\zeta_x,\zeta_y\}$,
\begin{equation}
T^{-1/2}
\left\|
\widehat{\mathbf F}_{-\tau}
-
\mathbf F^0_{-\tau}\mathbf H_{x,-\tau}
\right\|
=
O_p(\delta_{NT}^{-1}),
\label{propa1ii}
\end{equation}
for an invertible $K_f\times K_f$ rotation matrix
$\mathbf H_{x,-\tau}$.

For $\tau=0$, writing
$\mathbf H_x=\mathbf H_{x,0}$, we further have
\begin{equation}
\frac{
\left(
\widehat{\mathbf F}
-
\mathbf F^0\mathbf H_x
\right)'
\widehat{\mathbf F}
}{T}
=
O_p(\delta_{NT}^{-2}),
\label{lem_fh}
\end{equation}
\begin{equation}
\frac{
\left(
\widehat{\mathbf F}
-
\mathbf F^0\mathbf H_x
\right)'
\mathbf F^0
}{T}
=
O_p(\delta_{NT}^{-2}),
\label{lem_f0}
\end{equation}
\begin{equation}
\frac{
\left(
\widehat{\mathbf F}
-
\mathbf F^0\mathbf H_x
\right)'
\mathbf G
}{T}
=
O_p(\delta_{NT}^{-2}),
\label{lem_fy}
\end{equation}
\begin{equation}
\frac{
\left(
\widehat{\mathbf F}
-
\mathbf F^0\mathbf H_x
\right)'
\boldsymbol\varepsilon_i
}{T}
=
O_p(\delta_{NT}^{-2}),
\label{lem_eps}
\end{equation}
and
\begin{equation}
\frac{
\left(
\widehat{\mathbf F}
-
\mathbf F^0\mathbf H_x
\right)'
\widetilde{\mathbf v}_{k,i}
}{T}
=
O_p(\delta_{NT}^{-2}).
\label{lem_v}
\end{equation}
Moreover,
\begin{equation}
\frac{
\left(
\widehat{\mathbf F}
-
\mathbf F^0\mathbf H_x
\right)'
\mathbf C_i
}{T}
=
O_p(\delta_{NT}^{-2}),
\label{lem_W}
\end{equation}
\begin{equation}
\mathbf H_x\mathbf H_x'
-
\left(
\frac{\mathbf F^{0\prime}\mathbf F^0}{T}
\right)^{-1}
=
O_p(\delta_{NT}^{-2}),
\label{la7i}
\end{equation}
and
\begin{equation}
\frac{\mathbf F^{0\prime}\widehat{\mathbf F}}{T}
\overset{p}{\longrightarrow}
\boldsymbol\Lambda_F,
\label{lam}
\end{equation}
where $\boldsymbol\Lambda_F$ is nonsingular.
\end{lemma}

\begin{proof}
[\textbf{Proof of Lemma \ref{bailemma}}]

For $i\in\mathcal S_s$ and $\tau=0,\ldots,\max\{\zeta_x,\zeta_y\}$, the reduced-form
representation gives
\[
\mathbf X_{i,-\tau}
=
\mathbf r_{-i,-\tau}\boldsymbol\phi_i'
+
\mathbf F^0_{-\tau}\boldsymbol\Lambda_i'
+
\mathbf V_{i,-\tau}.
\]
Premultiplying by $\mathbf M_{\mathbf R_{-\tau}}$ yields
\begin{equation}
\mathbf M_{\mathbf R_{-\tau}}\mathbf X_{i,-\tau}
=
\mathbf F^0_{-\tau}\boldsymbol\Lambda_i'
+
\widetilde{\mathbf V}_{i,-\tau}
+
\mathbf D_{i,-\tau},
\label{eq:projected_factor_representation}
\end{equation}
where
\[
\widetilde{\mathbf V}_{i,-\tau}
=
\mathbf M_{\mathbf R_{-\tau}}\mathbf V_{i,-\tau},
\qquad
\mathbf D_{i,-\tau}
=
\mathbf M_{\mathbf R_{-\tau}}
\mathbf r_{-i,-\tau}\boldsymbol\phi_i'.
\]
By Lemma~\ref{LEM_MiMr}, the moment bounds on
$\boldsymbol\phi_i$, and $N_s/N\to\pi_s>0$,
\[
\frac{\|\mathbf D_{i,-\tau}\|}{\sqrt T}
=
O_p(N_s^{-1})
=
O_p(N^{-1})
=
O_p(\delta_{NT}^{-2})
\]
uniformly over $i$ and $\tau$.

Consequently, the covariance matrix in
Eq.~\eqref{factor_covariance} differs from that of the approximate factor
model
\[
\mathbf F^0_{-\tau}\boldsymbol\Lambda_i'
+
\widetilde{\mathbf V}_{i,-\tau}
\]
only by asymptotically negligible terms of order no larger than the
second-order PCA remainder.

By Assumption~3,
\[
\frac{1}{N_s}
\sum_{i\in\mathcal S_s}
\boldsymbol\Lambda_i\boldsymbol\Lambda_i'
\overset{p}{\longrightarrow}
\mathbf Q_{\Lambda,s}>0
\]
for every sector. Since $N_s/N\to\pi_s>0$,
\[
\frac1N
\sum_{i=1}^N
\boldsymbol\Lambda_i\boldsymbol\Lambda_i'
=
\sum_{s=1}^J
\frac{N_s}{N}
\left(
\frac1{N_s}
\sum_{i\in\mathcal S_s}
\boldsymbol\Lambda_i\boldsymbol\Lambda_i'
\right)
\overset{p}{\longrightarrow}
\sum_{s=1}^J
\pi_s\mathbf Q_{\Lambda,s}
\equiv
\mathbf Q_\Lambda>0.
\]
Hence the factor structure remains pervasive after projection on the
finite-dimensional observed sector-return space.

Under Assumptions A.1--A.4, the projected idiosyncratic components
$\widetilde{\mathbf V}_{i,-\tau}$ satisfy the moment and weak-dependence
conditions required for the standard PCA results. In particular,
Assumption~A.2 gives
\[
\lambda_{\max}
\left\{
\operatorname{Var}
(\widetilde{\mathbf v}_{k,i,-\tau})
\right\}
\leq
\max_{1\leq s\leq T}
\sum_{t=1}^T
\left|
\widetilde{\Gamma}_{k,i,T}^{(\tau)}(s,t)
\right|
\leq\Delta,
\]
where the first inequality is the standard row-sum bound for the spectral norm of a symmetric matrix.

The PCA results of \citet{Bai2003,Bai2009}, together with the negligible remainder $\mathbf D_{i,-\tau}$, therefore yield Eq.~\eqref{propa1ii} uniformly over the finite set of lag orders. The corresponding second-order PCA expansions, together with the normalization of the estimated factors and the moment and interaction
bounds in Assumptions A.1--A.4, yield Eqs.~\eqref{lem_fh}, \eqref{lem_f0}, \eqref{lem_eps},
\eqref{lem_v}, and \eqref{la7i}. The spanning condition
\[
\operatorname{col}(\mathbf G)
\subseteq
\operatorname{col}(\mathbf F^0)
\]
together with Eq.~\eqref{lem_f0} implies Eq.~\eqref{lem_fy}, while Eq.~\eqref{lem_W} follows from the same cross-product bounds applied to the blocks of $\mathbf C_i$. Finally, Eq.~\eqref{lam} follows from the usual rotation normalization
and factor pervasiveness.
\end{proof}

\bigskip

\begin{lemma}
\label{LEM_ZMU0}
Under Assumptions 1--5 and A.0--A.4, as $N,T\to\infty$ such that
$N/T\to c\in(0,\infty)$,
\begin{align}
T^{-1/2}\mathbf r_{-i,-\tau}'
\left(
\mathbf M_{\widehat{\mathbf F}}
-
\mathbf M_{\mathbf F^0}
\right)\mathbf u_i
&=
\sqrt T\,O_p(\delta_{NT}^{-2}),
\qquad
\tau=0,\ldots,\zeta_r,
\label{LEM_ZMU0_r}
\\
T^{-1/2}\mathbf X_i'
\left(
\mathbf M_{\widehat{\mathbf F}}
-
\mathbf M_{\mathbf F^0}
\right)\mathbf u_i
&=
\sqrt T\,O_p(\delta_{NT}^{-2}),
\label{LEM_ZMU0_e1}
\\
T^{-1/2}\mathbf X_{i,-\tau}'
\mathbf M_{\widehat{\mathbf F}_{-\tau}}
\left(
\mathbf M_{\widehat{\mathbf F}}
-
\mathbf M_{\mathbf F^0}
\right)\mathbf u_i
&=
\sqrt T\,O_p(\delta_{NT}^{-2}),
\label{LEM_ZMU0_e2}
\\
T^{-1/2}\mathbf X_{i,-\tau}'
\left(
\mathbf M_{\widehat{\mathbf F}_{-\tau}}
-
\mathbf M_{\mathbf F^0_{-\tau}}
\right)
\mathbf M_{\mathbf F^0}\mathbf u_i
&=
\sqrt T\,O_p(\delta_{NT}^{-2}),
\label{LEM_ZMU0_e3}
\end{align}
uniformly over the finite sets of relevant lag orders. The same results
hold with $\mathbf X_i$ replaced by $\mathbf Y_i$ under the analogous
moment conditions in Assumption~A.4.
\end{lemma}

\begin{proof}
[\textbf{Proof of Lemma \ref{LEM_ZMU0}}]

Let
\[
\mathbf E
=
\widehat{\mathbf F}
-
\mathbf F^0\mathbf H_x,
\qquad
\mathbf S_F
=
\frac{\mathbf F^{0\prime}\mathbf F^0}{T}.
\]
Since
$\widehat{\mathbf F}'\widehat{\mathbf F}/T=\mathbf I_{K_f}$,
\[
\mathbf M_{\widehat{\mathbf F}}
-
\mathbf M_{\mathbf F^0}
=
-
\left(
\frac{\widehat{\mathbf F}\widehat{\mathbf F}'}{T}
-
\mathbf P_{\mathbf F^0}
\right),
\]
and
\begin{align}
\frac{\widehat{\mathbf F}\widehat{\mathbf F}'}{T}
-
\mathbf P_{\mathbf F^0}
={}&
\frac{\mathbf E\mathbf H_x'\mathbf F^{0\prime}}{T}
+
\frac{\mathbf E\mathbf E'}{T}
+
\frac{\mathbf F^0\mathbf H_x\mathbf E'}{T}
\nonumber\\
&+
\frac{\mathbf F^0}{T}
\left[
\mathbf H_x\mathbf H_x'
-
\mathbf S_F^{-1}
\right]
\mathbf F^{0\prime}.
\label{eq:projection_expansion}
\end{align}

Lemma~\ref{bailemma} implies
\[
\frac{\|\mathbf E\|}{\sqrt T}
=
O_p(\delta_{NT}^{-1}),
\qquad
\left\|
\frac{\mathbf F^{0\prime}\mathbf E}{T}
\right\|
=
O_p(\delta_{NT}^{-2}),
\]
\[
\left\|
\frac{\mathbf G'\mathbf E}{T}
\right\|
=
O_p(\delta_{NT}^{-2}),
\qquad
\left\|
\frac{\boldsymbol\varepsilon_i'\mathbf E}{T}
\right\|
=
O_p(\delta_{NT}^{-2}),
\]
and
\[
\left\|
\mathbf H_x\mathbf H_x'
-
\mathbf S_F^{-1}
\right\|
=
O_p(\delta_{NT}^{-2}).
\]
Since
\[
\mathbf u_i
=
\mathbf G\boldsymbol\gamma_i
+
\boldsymbol\varepsilon_i,
\]
it follows that
\begin{equation}
\left\|
\frac{\mathbf E'\mathbf u_i}{T}
\right\|
=
O_p(\delta_{NT}^{-2}),
\label{eq:Eui_rate}
\end{equation}
while
\begin{equation}
\frac{\|\mathbf u_i\|}{\sqrt T}=O_p(1),
\qquad
\left\|
\frac{\mathbf F^{0\prime}\mathbf u_i}{T}
\right\|
=
O_p(1).
\label{eq:u_rates}
\end{equation}

For the leave-one-out return block, Eq.~\eqref{LEM_ZMU0_r} follows by substituting
Eq.~\eqref{eq:projection_expansion} and using the factor-estimation and cross-product bounds in Lemma~\ref{bailemma}.

For the contemporaneous micro block, Lemma~\ref{bailemma} and the
maintained moment conditions give
\[
\left\|
\frac{\mathbf X_i'\mathbf E}{T}
\right\|
=
O_p(\delta_{NT}^{-2}),
\qquad
\left\|
\frac{\mathbf X_i'\mathbf F^0}{T}
\right\|
=
O_p(1).
\]
Substituting these bounds together with
Eqs.~\eqref{eq:Eui_rate}--\eqref{eq:u_rates} into
Eq.~\eqref{eq:projection_expansion} yields
\[
T^{-1/2}\mathbf X_i'
\left(
\mathbf M_{\widehat{\mathbf F}}
-
\mathbf M_{\mathbf F^0}
\right)\mathbf u_i
=
\sqrt T\,O_p(\delta_{NT}^{-2}),
\]
which establishes Eq.~\eqref{LEM_ZMU0_e1}.

For $\tau=0,\ldots,\max\{\zeta_x,\zeta_y\}$, define
\[
\mathbf E_{-\tau}
=
\widehat{\mathbf F}_{-\tau}
-
\mathbf F^0_{-\tau}\mathbf H_{x,-\tau}.
\]
The PCA results of Lemma~\ref{bailemma} hold uniformly over this finite
set, so that
\[
\frac{\|\mathbf E_{-\tau}\|}{\sqrt T}
=
O_p(\delta_{NT}^{-1})
\]
together with the corresponding second-order cross-product bounds.

For Eq.~\eqref{LEM_ZMU0_e2}, write
\begin{align*}
&
\frac{1}{\sqrt T}
\mathbf X_{i,-\tau}'
\mathbf M_{\widehat{\mathbf F}_{-\tau}}
\left(
\mathbf M_{\widehat{\mathbf F}}
-
\mathbf M_{\mathbf F^0}
\right)\mathbf u_i
\\
={}&
\frac{1}{\sqrt T}
\mathbf X_{i,-\tau}'
\mathbf M_{\mathbf F^0_{-\tau}}
\left(
\mathbf M_{\widehat{\mathbf F}}
-
\mathbf M_{\mathbf F^0}
\right)\mathbf u_i
\\
&+
\frac{1}{\sqrt T}
\mathbf X_{i,-\tau}'
\left(
\mathbf M_{\widehat{\mathbf F}_{-\tau}}
-
\mathbf M_{\mathbf F^0_{-\tau}}
\right)
\left(
\mathbf M_{\widehat{\mathbf F}}
-
\mathbf M_{\mathbf F^0}
\right)\mathbf u_i.
\end{align*}
The first term is
$\sqrt T\,O_p(\delta_{NT}^{-2})$ by the same projection expansion as
above. For the second term,
\[
\left\|
\mathbf M_{\widehat{\mathbf F}_{-\tau}}
-
\mathbf M_{\mathbf F^0_{-\tau}}
\right\|_2
=
O_p(\delta_{NT}^{-1}),
\]
\[
\left\|
\mathbf M_{\widehat{\mathbf F}}
-
\mathbf M_{\mathbf F^0}
\right\|_2
=
O_p(\delta_{NT}^{-1}),
\]
and
\[
T^{-1/2}\|\mathbf X_{i,-\tau}\|=O_p(1),
\qquad
T^{-1/2}\|\mathbf u_i\|=O_p(1),
\]
so the second term is also
$\sqrt T\,O_p(\delta_{NT}^{-2})$. This proves
Eq.~\eqref{LEM_ZMU0_e2}.

Finally, apply the same projector expansion to
\[
\mathbf M_{\widehat{\mathbf F}_{-\tau}}
-
\mathbf M_{\mathbf F^0_{-\tau}}.
\]
The spanning condition
\[
\operatorname{col}(\mathbf G)
\subseteq
\operatorname{col}(\mathbf F^0)
\]
implies
\[
\mathbf M_{\mathbf F^0}\mathbf u_i
=
\mathbf M_{\mathbf F^0}\boldsymbol\varepsilon_i.
\]
The corresponding factor-error cross-products therefore satisfy the second-order bounds in  Lemma~\ref{bailemma}, yielding Eq.~\eqref{LEM_ZMU0_e3}. The same arguments apply to the market-risk block and its selected lags.
\end{proof}

\begin{proof}[\textbf{Proof of Theorem \ref{theorem:non_reducibility}}]
We prove each statement in turn.

\noindent
(i) Since
$y_{i,t}=(N_i-1)^{-1}
\sum_{j\in\mathcal S_i,j\neq i}x_{j,t}$,
the difference between the aggregated micro contribution of the remaining
commodities in sector $\mathcal S_i$ and the market contribution is
\[
\sum_{\substack{j\in\mathcal S_i\\j\neq i}}
\alpha_jx_{j,t}
-\beta_i y_{i,t}
=
\sum_{\substack{j\in\mathcal S_i\\j\neq i}}
\left(\alpha_j-\frac{\beta_i}{N_i-1}\right)x_{j,t}.
\]
Hence, the two contributions coincide at every sample date if and only if
\[
\sum_{\substack{j\in\mathcal S_i\\j\neq i}}
\left(\alpha_j-\frac{\beta_i}{N_i-1}\right)x_{j,t}=0
\qquad
\text{for every }t.
\]
It follows immediately that
$\alpha_j=\beta_i/(N_i-1)$ for all
$j\in\mathcal S_i$, $j\neq i$, is sufficient for equality.

To establish the additional necessity result, let
$\mathbf{x}_{-i,t}^{(s)}$ collect the characteristics $x_{j,t}$ for
$j\in\mathcal S_i$, $j\neq i$, and let $\mathbf d_i^{(s)}$ collect the
corresponding coefficient differences
$\alpha_j-\beta_i/(N_i-1)$. Define
\[
\mathbf Q_{i,T}^{(s)}
=
\frac{1}{T}\sum_{t=1}^T
\mathbf{x}_{-i,t}^{(s)}
\mathbf{x}_{-i,t}^{(s)\prime}.
\]
If the two contributions coincide at every sample date, then
\[
\mathbf d_i^{(s)\prime}\mathbf Q_{i,T}^{(s)}\mathbf d_i^{(s)}
=
\frac{1}{T}\sum_{t=1}^T
\left(
\mathbf d_i^{(s)\prime}\mathbf{x}_{-i,t}^{(s)}
\right)^2
=0.
\]
If $\mathbf Q_{i,T}^{(s)}$ is nonsingular, this implies
$\mathbf d_i^{(s)}=\mathbf 0$, and therefore
$\alpha_j=\beta_i/(N_i-1)$ for every
$j\in\mathcal S_i$, $j\neq i$.
Likewise, if equality is required as an identity for every possible
realization of the sector characteristics, setting each characteristic equal
to one in turn while setting the others equal to zero implies
$\mathbf d_i^{(s)}=\mathbf 0$.

\medskip
\noindent
(ii) Since $y_{i,t}=\lambda_i g_t+\eta_{i,t}$, the difference between the
macro and market contributions is
\[
\gamma_i g_t-\beta_i y_{i,t}
=
(\gamma_i-\beta_i\lambda_i)g_t-\beta_i\eta_{i,t}.
\]
Hence, the two contributions coincide at every sample date if and only if
\[
\beta_i\eta_{i,t}
=
(\gamma_i-\beta_i\lambda_i)g_t
\qquad
\text{for every }t.
\]
When $\beta_i\neq0$, division by $\beta_i$ gives
\[
\eta_{i,t}
=
\left(\frac{\gamma_i}{\beta_i}-\lambda_i\right)g_t
\qquad
\text{for every }t.
\]
Thus, equality requires $\eta_{i,t}$ to be exactly proportional to $g_t$.
If $g_t$ and $\eta_{i,t}$ are linearly independent over the sample, this is
impossible when $\beta_i\neq0$. When $\beta_i=0$, equality reduces to
$\gamma_i g_t=0$ for every $t$. Since $g_t$ is not identically zero, this
requires $\gamma_i=0$. Hence, under linear independence, equality is possible
only in the trivial case $\beta_i=\gamma_i=0$.
\end{proof}

The following proposition establishes the limiting property of $T^{-1/2}\widehat{\mathbf Z}_{i}'
\mathbf M_{\widehat{\mathbf F}}\mathbf u_i$
under the maintained assumptions and is used in the proof of Theorem~\ref{TH2}.

\begin{proposition}
\label{PROP_ZMFi}
Consider the model in Eq.~\eqref{compact} and the instrument matrix
in Eq.~\eqref{Instruments_matrix}. Define the infeasible population
counterpart of $\widehat{\mathbf Z}_i$ as
\begin{align}
\mathbf Z_i
=
\Big(
&
\mathbf r_{-i},\;
\mathbf r_{-i,-1},\ldots,
\mathbf r_{-i,-\zeta_r},
\nonumber\\
&
\mathbf M_{\mathbf F^0}\mathbf X_i,\;
\mathbf M_{\mathbf F^0_{-1}}\mathbf X_{i,-1},
\ldots,
\mathbf M_{\mathbf F^0_{-\zeta_x}}
\mathbf X_{i,-\zeta_x},
\nonumber\\
&
\mathbf M_{\mathbf F^0}\mathbf Y_i,\;
\mathbf M_{\mathbf F^0_{-1}}\mathbf Y_{i,-1},
\ldots,
\mathbf M_{\mathbf F^0_{-\zeta_y}}
\mathbf Y_{i,-\zeta_y}
\Big).
\label{population_instruments}
\end{align}
Under Assumptions 1--5 and the regularity conditions stated in this
Appendix, as $N,T\to\infty$ such that $N/T\to c\in(0,\infty)$,
\begin{equation}
T^{-1/2}
\widehat{\mathbf Z}_i'
\mathbf M_{\widehat{\mathbf F}}
\mathbf u_i
=
T^{-1/2}
\mathbf Z_i'
\mathbf M_{\mathbf F^0}
\mathbf u_i
+
\sqrt T\,
O_p(\delta_{NT}^{-2}).
\label{prop_zmfi_result}
\end{equation}
Since $N/T\to c\in(0,\infty)$,
\[
\sqrt T\,\delta_{NT}^{-2}=O(T^{-1/2})=o(1).
\]
\end{proposition}

\begin{proof}
[\textbf{Proof of Proposition \ref{PROP_ZMFi}}]

Recall that
\[
\mathbf r_i
=
\mathbf C_i\boldsymbol\theta_i+\mathbf u_i,
\qquad
\mathbf C_i
=
(\mathbf r_{-i},\mathbf X_i,\mathbf Y_i).
\]
We consider each block of
\[
T^{-1/2}
\widehat{\mathbf Z}_i'
\mathbf M_{\widehat{\mathbf F}}
\mathbf u_i
\]
in turn.

\medskip
\noindent
\textit{Leave-one-out return block.}

For $\tau=0,\ldots,\zeta_r$,
\begin{align}
\frac{1}{\sqrt T}
\mathbf r_{-i,-\tau}'
\mathbf M_{\widehat{\mathbf F}}
\mathbf u_i
&=
\frac{1}{\sqrt T}
\mathbf r_{-i,-\tau}'
\mathbf M_{\mathbf F^0}
\mathbf u_i
\nonumber\\
&\quad+
\frac{1}{\sqrt T}
\mathbf r_{-i,-\tau}'
\left(
\mathbf M_{\widehat{\mathbf F}}
-
\mathbf M_{\mathbf F^0}
\right)
\mathbf u_i
\nonumber\\
&=
\frac{1}{\sqrt T}
\mathbf r_{-i,-\tau}'
\mathbf M_{\mathbf F^0}
\mathbf u_i
+
\sqrt T\,O_p(\delta_{NT}^{-2}),
\label{prop_r_block}
\end{align}
by Lemma~\ref{LEM_ZMU0}.

\medskip
\noindent
\textit{Contemporaneous micro block.}

The contemporaneous micro instrument is
$\mathbf M_{\widehat{\mathbf F}}\mathbf X_i$. By symmetry and
idempotency of $\mathbf M_{\widehat{\mathbf F}}$, its contribution to
the sample moment is
\[
\frac{1}{\sqrt T}
\mathbf X_i'
\mathbf M_{\widehat{\mathbf F}}
\mathbf u_i.
\]
Hence
\begin{align}
\frac{1}{\sqrt T}
\mathbf X_i'
\mathbf M_{\widehat{\mathbf F}}
\mathbf u_i
&=
\frac{1}{\sqrt T}
\mathbf X_i'
\mathbf M_{\mathbf F^0}
\mathbf u_i
\nonumber\\
&\quad+
\frac{1}{\sqrt T}
\mathbf X_i'
\left(
\mathbf M_{\widehat{\mathbf F}}
-
\mathbf M_{\mathbf F^0}
\right)
\mathbf u_i
\nonumber\\
&=
\frac{1}{\sqrt T}
\mathbf X_i'
\mathbf M_{\mathbf F^0}
\mathbf u_i
+
\sqrt T\,O_p(\delta_{NT}^{-2}),
\label{prop_x0}
\end{align}
by Lemma~\ref{LEM_ZMU0}.

\medskip
\noindent
\textit{Lagged micro block.}

For $\tau=1,\ldots,\zeta_x$, the corresponding sample moment is
\[
\frac{1}{\sqrt T}
\mathbf X_{i,-\tau}'
\mathbf M_{\widehat{\mathbf F}_{-\tau}}
\mathbf M_{\widehat{\mathbf F}}
\mathbf u_i.
\]
Adding and subtracting $\mathbf M_{\mathbf F^0}$ gives
\begin{align*}
&
\frac{1}{\sqrt T}
\mathbf X_{i,-\tau}'
\mathbf M_{\widehat{\mathbf F}_{-\tau}}
\mathbf M_{\widehat{\mathbf F}}
\mathbf u_i
\\
={}&
\frac{1}{\sqrt T}
\mathbf X_{i,-\tau}'
\mathbf M_{\widehat{\mathbf F}_{-\tau}}
\mathbf M_{\mathbf F^0}
\mathbf u_i
\\
&+
\frac{1}{\sqrt T}
\mathbf X_{i,-\tau}'
\mathbf M_{\widehat{\mathbf F}_{-\tau}}
\left(
\mathbf M_{\widehat{\mathbf F}}
-
\mathbf M_{\mathbf F^0}
\right)
\mathbf u_i.
\end{align*}
By Eq.~\eqref{LEM_ZMU0_e2}, the second term is
$\sqrt T\,O_p(\delta_{NT}^{-2})$.

Next,
\begin{align*}
&
\frac{1}{\sqrt T}
\mathbf X_{i,-\tau}'
\mathbf M_{\widehat{\mathbf F}_{-\tau}}
\mathbf M_{\mathbf F^0}
\mathbf u_i
\\
={}&
\frac{1}{\sqrt T}
\mathbf X_{i,-\tau}'
\mathbf M_{\mathbf F^0_{-\tau}}
\mathbf M_{\mathbf F^0}
\mathbf u_i
\\
&+
\frac{1}{\sqrt T}
\mathbf X_{i,-\tau}'
\left(
\mathbf M_{\widehat{\mathbf F}_{-\tau}}
-
\mathbf M_{\mathbf F^0_{-\tau}}
\right)
\mathbf M_{\mathbf F^0}
\mathbf u_i.
\end{align*}
By Eq.~\eqref{LEM_ZMU0_e3}, the second term is again
$\sqrt T\,O_p(\delta_{NT}^{-2})$. Therefore,
\begin{equation}
\frac{1}{\sqrt T}
\mathbf X_{i,-\tau}'
\mathbf M_{\widehat{\mathbf F}_{-\tau}}
\mathbf M_{\widehat{\mathbf F}}
\mathbf u_i
=
\frac{1}{\sqrt T}
\mathbf X_{i,-\tau}'
\mathbf M_{\mathbf F^0_{-\tau}}
\mathbf M_{\mathbf F^0}
\mathbf u_i
+
\sqrt T\,O_p(\delta_{NT}^{-2}).
\label{prop_xtau}
\end{equation}

\medskip
\noindent
\textit{Market-risk block.}

The same arguments apply to the contemporaneous and lagged market-risk
characteristics. In particular,
\begin{equation}
\frac{1}{\sqrt T}
\mathbf Y_i'
\mathbf M_{\widehat{\mathbf F}}
\mathbf u_i
=
\frac{1}{\sqrt T}
\mathbf Y_i'
\mathbf M_{\mathbf F^0}
\mathbf u_i
+
\sqrt T\,O_p(\delta_{NT}^{-2}),
\label{prop_y0}
\end{equation}
and, for $\tau=1,\ldots,\zeta_y$,
\begin{align}
&
\frac{1}{\sqrt T}
\mathbf Y_{i,-\tau}'
\mathbf M_{\widehat{\mathbf F}_{-\tau}}
\mathbf M_{\widehat{\mathbf F}}
\mathbf u_i
\nonumber\\
&=
\frac{1}{\sqrt T}
\mathbf Y_{i,-\tau}'
\mathbf M_{\mathbf F^0_{-\tau}}
\mathbf M_{\mathbf F^0}
\mathbf u_i
+
\sqrt T\,O_p(\delta_{NT}^{-2}).
\label{prop_ytau}
\end{align}

Collecting Eqs.~\eqref{prop_r_block}, \eqref{prop_x0},
\eqref{prop_xtau}, \eqref{prop_y0}, and \eqref{prop_ytau} across the
fixed number of instrument blocks yields
\[
T^{-1/2}
\widehat{\mathbf Z}_i'
\mathbf M_{\widehat{\mathbf F}}
\mathbf u_i
=
T^{-1/2}
\mathbf Z_i'
\mathbf M_{\mathbf F^0}
\mathbf u_i
+
\sqrt T\,O_p(\delta_{NT}^{-2}),
\]
which proves the result.
\end{proof}

\begin{proof}[\textbf{Proof of Theorem \ref{TH2}}]
We work conditionally on $\mathcal H$ for a fixed $i$. Substituting $\mathbf r_i=\mathbf C_i\boldsymbol\theta_i+\mathbf u_i$ into the estimator gives
\[
\sqrt T
\left(
\widehat{\boldsymbol\theta}_i-\boldsymbol\theta_i
\right)
=
\widehat{\mathbf K}_{iT}
\frac{1}{\sqrt T}
\widehat{\mathbf Z}_i'
\mathbf M_{\widehat{\mathbf F}}
\mathbf u_i,
\]
where
\[
\widehat{\mathbf K}_{iT}
=
\left(
\widehat{\mathbf A}_{i,T}'
\widehat{\mathbf B}_{i,T}^{-1}
\widehat{\mathbf A}_{i,T}
\right)^{-1}
\widehat{\mathbf A}_{i,T}'
\widehat{\mathbf B}_{i,T}^{-1}.
\]
By Assumption~5,
\[
\widehat{\mathbf K}_{iT}
\overset{p}{\longrightarrow}
\mathbf K_i
=
\left(
\mathbf A_i'\mathbf B_i^{-1}\mathbf A_i
\right)^{-1}
\mathbf A_i'\mathbf B_i^{-1}.
\]
Moreover, Proposition~\ref{PROP_ZMFi} gives
\[
\frac{1}{\sqrt T}
\widehat{\mathbf Z}_i'
\mathbf M_{\widehat{\mathbf F}}
\mathbf u_i
=
\frac{1}{\sqrt T}
\mathbf Z_i'
\mathbf M_{\mathbf F^0}
\mathbf u_i
+o_p(1).
\]
By the population-score representation and the conditional central limit
theorem in Assumption~5,
\[
\frac{1}{\sqrt T}
\mathbf Z_i'
\mathbf M_{\mathbf F^0}
\mathbf u_i
=
\mathbf a_{iT}+o_p(1)
\overset{d}{\longrightarrow}
N(\mathbf 0,\boldsymbol\Omega_i).
\]
Hence, by Slutsky's theorem,
\[
\sqrt T
\left(
\widehat{\boldsymbol\theta}_i-\boldsymbol\theta_i
\right)
\overset{d}{\longrightarrow}
N\left(
\mathbf 0,\,
\mathbf K_i\boldsymbol\Omega_i\mathbf K_i'
\right)
\]
conditionally on $\mathcal H$. Since
\[
\mathbf K_i\boldsymbol\Omega_i\mathbf K_i'
=
(\mathbf A_i'\mathbf B_i^{-1}\mathbf A_i)^{-1}
\mathbf A_i'\mathbf B_i^{-1}
\boldsymbol\Omega_i
\mathbf B_i^{-1}\mathbf A_i
(\mathbf A_i'\mathbf B_i^{-1}\mathbf A_i)^{-1}
=
\mathbf V_i,
\]
the stated result follows.
\end{proof}

\begin{proof}[\textbf{Proof of Theorem \ref{TH3}}]
We work conditionally on $\mathcal H$ for a fixed $i$. Define
\[
\widehat{\mathbf s}_{iT}
=
\frac{1}{\sqrt T}
\widehat{\mathbf Z}_i'
\mathbf M_{\widehat{\mathbf F}}
\mathbf u_i.
\]
By Proposition~\ref{PROP_ZMFi} and Assumption~5,
\[
\widehat{\mathbf s}_{iT}
\overset{d}{\longrightarrow}
N(\mathbf 0,\boldsymbol\Omega_i).
\]
Moreover, by Assumption~A.5,
\[
\widehat{\boldsymbol\Omega}_{iT}
\overset{p}{\longrightarrow}
\boldsymbol\Omega_i.
\]

Using
$\mathbf r_i=\mathbf C_i\boldsymbol\theta_i+\mathbf u_i$ and the definition
of $\widetilde{\boldsymbol\theta}_i$,
\[
\sqrt T\,
\overline{\mathbf g}_{iT}
(\widetilde{\boldsymbol\theta}_i)
=
\widehat{\mathbf s}_{iT}
-
\widehat{\mathbf A}_{i,T}
\sqrt T
(\widetilde{\boldsymbol\theta}_i-\boldsymbol\theta_i).
\]
The first-order conditions for
$\widetilde{\boldsymbol\theta}_i$ imply
\[
\sqrt T
(\widetilde{\boldsymbol\theta}_i-\boldsymbol\theta_i)
=
\left(
\widehat{\mathbf A}_{i,T}'
\widehat{\boldsymbol\Omega}_{iT}^{-1}
\widehat{\mathbf A}_{i,T}
\right)^{-1}
\widehat{\mathbf A}_{i,T}'
\widehat{\boldsymbol\Omega}_{iT}^{-1}
\widehat{\mathbf s}_{iT}.
\]
Hence,
\[
\widehat{\boldsymbol\Omega}_{iT}^{-1/2}
\sqrt T\,
\overline{\mathbf g}_{iT}
(\widetilde{\boldsymbol\theta}_i)
=
\left(
\mathbf I_{q_i}-\widehat{\mathbf P}_{iT}
\right)
\widehat{\boldsymbol\Omega}_{iT}^{-1/2}
\widehat{\mathbf s}_{iT},
\]
where
\[
\widehat{\mathbf P}_{iT}
=
\widehat{\boldsymbol\Omega}_{iT}^{-1/2}
\widehat{\mathbf A}_{i,T}
\left(
\widehat{\mathbf A}_{i,T}'
\widehat{\boldsymbol\Omega}_{iT}^{-1}
\widehat{\mathbf A}_{i,T}
\right)^{-1}
\widehat{\mathbf A}_{i,T}'
\widehat{\boldsymbol\Omega}_{iT}^{-1/2}.
\]
Since $\mathbf A_i$ has full column rank $p$,
$\widehat{\mathbf P}_{iT}$ converges in probability to a rank-$p$
orthogonal projection matrix. Therefore,
$\mathbf I_{q_i}-\widehat{\mathbf P}_{iT}$ converges to an idempotent
projection matrix of rank $q_i-p$. Since
\[
\boldsymbol\Omega_i^{-1/2}\widehat{\mathbf s}_{iT}
\overset{d}{\longrightarrow}
N(\mathbf 0,\mathbf I_{q_i}),
\]
the quadratic form
\[
S_{iT}
=
T\,
\overline{\mathbf g}_{iT}
(\widetilde{\boldsymbol\theta}_i)'
\widehat{\boldsymbol\Omega}_{iT}^{-1}
\overline{\mathbf g}_{iT}
(\widetilde{\boldsymbol\theta}_i)
\]
converges in distribution to
$\chi^2_{q_i-p}$ conditionally on $\mathcal H$.
\end{proof}

\section{Monte Carlo Study\label{sec:Appendix_MC}}

This section evaluates the finite-sample performance of the PCA-BMT method using diagnostic metrics commonly used in high-dimensional variable selection. These metrics are formally defined below.

\subsection{Design}

The data-generating process (DGP) is given by
\begin{equation}
u_{i,t}=\boldsymbol{\gamma}_{i}^{\prime} \mathbf{g}_{t} + \varepsilon_{i,t}, \label{DGP_MC}
\end{equation}
where $\mathbf{g}_{t}=\left(g_{1,t},\dots,g_{r,t}\right)^{\prime}$ denotes an $r \times 1$ vector of true common factors (or ``signals'') affecting all individuals albeit with heterogeneous loadings  $\boldsymbol{\gamma}_{i}$. The outcome $u_{i,t}$ represents the population counterpart of the first-stage residual $\widehat{u}_{i,t}$ used in Stage~2 of the empirical
procedure. The micro errors $\varepsilon_{i,t}$ are drawn as $ i.i.d.N(0,1)$ for all $i$ and $t$.

The common component is drawn as
\begin{equation}
g_{j,t}=\rho g_{j,t-1} +\sqrt{\left(1-\rho^2\right)}\epsilon_{j,t}; \quad j=1,\dots, r,
\end{equation}
where $\epsilon_{t} \sim i.i.d.N(0,1)$. The loadings are drawn from
\begin{equation}
\boldsymbol{\gamma}_{i}  \sim \mathcal{N}_r(\boldsymbol{\mu}_{r}, \boldsymbol{\Sigma}_{r}),
\end{equation}
with $\boldsymbol{\mu}_{r}=\phi\mathbf{1}_{r},$, where $\mathbf{1}_{r}=\left(1,1,\dots,1\right)^{\prime}$. The covariance matrix $\boldsymbol{\Sigma}_{r} \in \mathbb{R}^{r \times r}$ is given by
\begin{equation}
\boldsymbol{\Sigma}_{r} = 0.5  \mathbf{1}_{r} \mathbf{1}_{r}^{\prime} + 0.5 \cdot \mathbf{I}_{r},
\end{equation}
so that all diagonal entries equal 1 and all off-diagonal entries equal 0.5, implying substantial mutual correlation among the loadings.

In addition to the common components entering in Eq. \eqref{DGP_MC}, we generate $n-r$ additional common components that do not enter the DGP in Eq. \eqref{DGP_MC}:
\begin{equation}
\psi_{j,t} =
\begin{cases}
\rho \psi_{j,t-1} + \sqrt{1 - \rho^2} \left( \pi \upsilon_{j,t} + \sqrt{\left(1 - \pi\right)} \epsilon_{j,t} \right), & \text{for } j = 1,\dots, r, \\
\rho \psi_{j,t-1} + \sqrt{1 - \rho^2} \upsilon_{j,t}, & \text{for } j = r+1,\dots,n-r,
\end{cases} \label{Two_cases}
\end{equation}
where $\upsilon_{j,t} \sim i.i.d.N(0,1)$ and is independent of $\varepsilon_{i,t}$. The first $r$ variables in Eq.~\eqref{Two_cases} are referred to as ``pseudo-signals'' because, although irrelevant for the DGP, they are correlated with the true factors and may therefore be mistakenly selected
as informative. The remaining $n-2r$ variables are pure noise variables. This structure allows us to evaluate the performance of variable selection methods in the presence of confounding and irrelevant variables.

We consider two levels of signal dimensionality, $r\in\{2,5\}$, with $\rho=0.5$ and $\pi=0.5$, allowing for persistence in the common components and correlation between the true signals and pseudo-signals. The location parameter $\phi\in\{0,1\}$ determines whether the factor loadings are centred at zero, thereby affecting the performance of selection methods that rely
on pooling across $i$. We consider
\[
T\in\{25,50,100\},
\qquad
n\in\{25,50,100\},
\]
and set $N=T$. For each value of $r$, combining the three values of $T$, the three values of $n$, and the two values of $\phi$ yields a total of 18 simulation designs. For each choice of $r$, the analysis covers a total of 18 configurations.

\subsection{The Boosting with Multiple Testing Method}\label{subsec:BMT_method}

The Boosting with Multiple Testing (BMT) approach of \citet{KapetaniosEtal2026} is a forward selection procedure designed to recover a sparse set of relevant predictors from a high-dimensional set. It constructs the model iteratively by testing one covariate at a time, updating the model step-by-step, and dynamically adjusting the selection threshold to control the family-wise error rate. In what follows, we provide a high-level description of the method.

Let $\mathbf{e}\in \mathbb{R}^{T \times 1}$ denote one of the $r$ principal components extracted from the sample variance-covariance matrix of \( u_{i,t} \). The pool of \( n \) candidate variables, stored in matrix \( \boldsymbol{\mathcal{Z}}=\left(\boldsymbol{\mathcal{Z}}_{1}, \dots, \boldsymbol{\mathcal{Z}}_{n} \right)=\left(\mathbf{G}, \boldsymbol{\Psi} \right) \in \mathbb{R}^{T \times n} \), where $\mathbf{G}=\left(\mathbf{g}_{1}, \dots, \mathbf{g}_{T} \right)^{\prime}$ and $\boldsymbol{\Psi}=\left(\boldsymbol{\psi}_{1}, \dots, \boldsymbol{\psi}_{T} \right)^{\prime}$ with $\boldsymbol{\psi}_{t}=\left(\psi_{1,t}, \dots, \psi_{n-r,t} \right)^{\prime}$.

In the first stage, we run $n$ univariate regressions of $\mathbf{e}_{j}$ on $\boldsymbol{\mathcal{Z}}_{i}$ for $i=1,\dots,n$. Let $t_{i}$ denote the $t$-ratio from the regression involving $\boldsymbol{\mathcal{Z}}_{i}$, and write this as $t_{i,(1)}$, where the subscript $(1)$ indexes the stage. We select the variable with the largest absolute $t$-ratio, that is, the $\boldsymbol{\mathcal{Z}}_{i}$ for which $|t_{i,(1)}|$ is maximised over $i$. Denote the index of this selected variable by $\mathcal{S}_{(1)}^{c}$, and let $\boldsymbol{\mathcal{Z}}_{(1)}^{o}$ be the selected $T \times 1$ observation vector in the first stage. Define $\boldsymbol{\mathcal{Z}}_{(1)}=\boldsymbol{\mathcal{Z}}_{(1)}^{o}$, and set $\mathcal{S}_{(1)} = \mathcal{S}_{(1)}^{c}$ and $\mathfrak{A}_{(2)} = \{1,\dots,n\} \setminus \mathcal{S}_{(1)}$. In the second stage, we run regressions of $\mathbf{e}_{j}$ on $\boldsymbol{X}_{(1)}$ and $\boldsymbol{\mathcal{Z}}_{i}$, for $i \in \mathfrak{A}_{(2)}$, and obtain the corresponding $t$-ratios for $\boldsymbol{\mathcal{Z}}_{i}$, denoted by $t_{i,(2)}$.
We construct selection indicators, given by
\begin{equation}
\widehat{\mathcal{J}}_{i,(2)}=I[|t_{i,(2)}|>c_{p}\left(  n,\varphi\right)
]\text{, for }i\in\mathfrak{A}_{\left(  2\right)  }\text{, }%
\end{equation}
where $c_{p}(n,\varphi)$ is the critical value function, defined by%
\begin{equation}
c_{p}\left(  n,\varphi\right)  =\Phi^{-1}\left(  1-\frac{p}{2f\left(
n,\varphi\right)  }\right)  \text{,} \label{cvf}%
\end{equation}
$\Phi^{-1}\left(  .\right)  $ is the inverse of standard normal distribution
function, $f\left(  n,\delta\right)  =cn^{\varphi}$ for some positive constants
$\varphi$ and $c$, and $p$ ($0<p<1$) is the nominal size of the individual
tests to be set by the investigator. We refer to $\varphi$ as the critical
value exponent. Among regressors with $\widehat{\mathcal{J}}_{i,(2)}=1$, we select
the one with the maximum absolute value of $t_{i,(2)}$. We denote the index of the
selected unit by $\mathcal{S}_{\left(  2\right)  }^{c}$ and the $T\times1$
observation vector of the selected variable by $\boldsymbol{\mathcal{Z}}_{(2)}^{o}$. We let $\boldsymbol{\mathcal{Z}}_{(2)}=\left(\boldsymbol{\mathcal{Z}}_{(1)},\boldsymbol{\mathcal{Z}}_{(2)}^{o}  \right)$, a $T \times 2$ matrix, and $\mathcal{S}_{\left(  2\right)
}=\mathcal{S}_{\left(  1\right)  }\cup\mathcal{S}_{\left(  2\right)  }^{c}$,
$\mathfrak{A}_{\left(  3\right)  }=\left\{  1,...,n\right\}
\setminus\mathcal{S}_{\left(  2\right)  }$.

We proceed to subsequent stages $\ell = 3,4,\ldots,\widehat{\mathcal{L}}$ in the same manner, defining
$\boldsymbol{\mathcal{Z}}_{(\ell)} = \left(\boldsymbol{\mathcal{Z}}_{(\ell-1)}, \boldsymbol{\mathcal{Z}}_{(\ell)}^{o}\right)$, $\mathcal{S}_{(\ell)} = \mathcal{S}_{(\ell-1)} \cup \mathcal{S}_{(\ell)}^{c}$ and $\mathfrak{A}_{(\ell+1)} = \{1,\dots,n\} \setminus \mathcal{S}_{(\ell)}$; here, $\widehat{\mathcal{L}}$ denotes the (random) last stage at which a regressor is admitted to the model. 
At each stage $\ell$, the value of $n$ in $c_p(n,\delta)$ is set equal to the cardinality of $\mathfrak{A}_{(\ell)}$. In other words, at stage $\ell$ we condition on all variables selected up to stage $\ell-1$, and regress $\mathbf{e}_{1}$ on each remaining variable in the active set in turn. Finally, the algorithm returns a selected index set $\mathcal{S} \subset \{1, \dots, n\}$, and the corresponding design matrix $\boldsymbol{\mathcal{Z}}_{(\mathcal{S})}\in \mathbb{R}^{T \times |S|}$, formed by the columns of $\boldsymbol{\mathcal{Z}}$ indexed by $S$, where $|S|$ denotes the cardinality of the selected set.

\medskip
\noindent
The above procedure is applied separately to each extracted principal component $\mathbf e_{j}$, for $j=1,\dots,r$, yielding index sets $\mathcal S_{j}$. The overall macro-financial predictor set is then constructed as
\[
\mathcal S = \bigcup_{j=1}^{r} \mathcal S_{j},
\]
and the corresponding design matrix $\boldsymbol{\mathcal Z}_{(\mathcal S)}$ is used in the final estimation of commodity-specific macro risk sensitivities.
This iterative framework is designed to recover the true signals while filtering out pseudo-signals and pure noise variables.

\subsection{Lasso-type Methods}\label{subsec:Lasso_method}

In addition to the combined PCA-BMT approach discussed above, we consider two alternative approaches: a pooled Lasso estimator that imposes homogeneous coefficients across units, and an individual-level Lasso estimator that allows for heterogeneous coefficients.

The pooled Lasso (``p-Lasso'', hereafter) estimator solves the following optimization problem:
\begin{equation}
\widehat{\boldsymbol{\gamma}}^{\text{p-Lasso}} = \arg\min_{\boldsymbol{\gamma} \in \mathbb{R}^{n}} \left\{\frac{1}{NT} \left\| \mathbf{u}_i - \boldsymbol{\mathcal{Z}} \boldsymbol{\gamma} \right\|_2^2
+ \xi \sum_{j=1}^{n} |\gamma_{j}|
\right\}, \label{eq:plasso}
\end{equation}
where $\xi \geq 0$ is a tuning parameter selected via 10-fold cross-validation. Although this model is mis-specified under the underlying DGP (since this assumes heterogeneous coefficients), the pooled estimator $\widehat{\boldsymbol{\gamma}}$ may still be effective for variable selection under sparsity, particularly when the individual-specific loadings are not centered around zero, i.e., $\phi \neq 0$. It can be interpreted as a regularized projection of the cross-sectional signal onto the predictor space spanned by $\boldsymbol{\mathcal{Z}}$.

The individual-level Lasso (``i-Lasso'', hereafter) estimator instead estimates a separate coefficient vector for each unit $i = 1, \dots, N$:
\begin{equation}
\widehat{\boldsymbol{\gamma}}_i^{\mathrm{i\text{-}Lasso}}
=
\arg\min_{\boldsymbol{\gamma}_i\in\mathbb R^n}
\left\{
\frac{1}{2T}
\left\|
\mathbf u_i-\boldsymbol{\mathcal Z}\boldsymbol{\gamma}_i
\right\|_2^2
+
\xi_i\sum_{j=1}^{n}|\gamma_{i,j}|
\right\}. \label{eq:ilasso}
\end{equation}
where $\boldsymbol{\gamma}_i$ is penalized using an individual-specific regularization parameter $\xi_{i}$,  selected via 10-fold cross-validation. To identify relevant common components, we apply a stability selection rule: for each factor $j$, we count how often it is selected across the $N$ individual Lasso estimates (i.e., when $\widehat{\gamma}_{i,j} \neq 0$). Predictors that appear in at least 25\% of units are retained. The final model size therefore reflects the number of variables that meet this frequency threshold, providing a measure of their cross-sectional relevance.\footnote{Naturally, the performance of i-Lasso can be sensitive to the choice of the stability selection threshold. While exploring its optimal calibration (e.g., via cross-validation) might improve results in small samples, such an investigation lies beyond the scope of this paper.}

\subsection{Results}

To assess the variable selection performance of PCA-BMT and alternative methods, we employ a comprehensive set of diagnostic metrics tailored to high-dimensional selection tasks, as detailed in \citet{KapetaniosEtal2026}. These include the True Positive Rate (TPR), False Positive Rate (FPR), True Discovery Rate (TDR), False Discovery Rate (FDR), the F1 score, and the Matthews Correlation Coefficient (MCC). These metrics are computed based on the confusion matrix that categorizes selected variables as true or false discoveries.

The TPR (or recall) measures the proportion of true signals correctly identified, while the TDR (or precision) reflects the proportion of selected variables that are indeed true signals. FPR and FDR capture the share of irrelevant variables mistakenly selected.

The F1 score, defined as the harmonic mean of TPR and TDR, summarizes the trade-off between precision and recall. The MCC offers a more comprehensive measure, incorporating all four components of the confusion matrix and accounting for both Type I and Type II errors. It is particularly valuable in imbalanced settings like ours, where the number of true signals is small relative to the total number of candidates.

We also report the model size, defined as the number of selected variables, as a measure of sparsity. Further details on the computation and interpretation of these metrics are provided in \citet{KapetaniosEtal2026}.

We focus primarily on the MCC criterion. Tables \ref{tab:MCC_r2}-\ref{tab:MCC_r5} present summary statistics across 18 designs for each sparsity level ($r = 2$ and $r = 5$).

PCA-BMT performs strongly, achieving the highest median MCC in both settings (0.958 for $r = 2$, 0.922 for $r = 5$), with relatively narrow interquartile ranges and strong ranking outcomes. Its MCC exceeds the 0.8 threshold in all cases for $r = 2$, and in approximately 78\% of cases for $r = 5$. As a result, it ranks first in the vast majority of configurations. Moreover, in the few instances where it ranks second, the difference in MCC compared to the top-ranked method is negligible, as indicated in Table \ref{tab:MCC}.

p-Lasso performs reasonably well in some instances (maximum MCC up to 0.854), but its results are more variable. For $r = 2$, it has a median MCC of 0.713 and a wider IQR (0.228). Further, it does not rank first in any configuration. As discussed further below, its performance appears sensitive to specific design features, particularly whether $\phi$, the common mean parameter of the loading distribution, equals zero or not.

i-Lasso outperforms p-Lasso, with median MCC values of 0.871 ($r = 2$) and 0.779 ($r = 5$). However, it generally falls short of PCA-BMT in both accuracy and stability across designs.

\begin{table}[th]
\caption{Summary Statistics for Matthews Correlation Coefficient (MCC) with $r=2$}%
\label{tab:MCC_r2}
\centering
\begin{tabular}{lccccccc}
\hline
& \textbf{median} & \textbf{IQR} & \textbf{min} & \textbf{max} & \textbf{prop. $>$ 0.8} & \textbf{rank} & \textbf{prop. 1st} \\
\hline
PCA-BMT   & 0.958 & 0.045 & 0.878 & 0.979 & 100.0\% & 1.056 & 94.44\% \\
p-Lasso  & 0.713 & 0.228 & 0.403 & 0.851 & 27.8\%  & 2.833 & 0.00\% \\
i-Lasso & 0.871 & 0.254 & 0.529 & 0.977 & 66.7\%  & 2.111 & 5.56\% \\
\hline
\end{tabular}
\vspace{0.75em}

\noindent\parbox{\linewidth}{%
\small\textit{\textbf{Notes}:} MCC ranges from $-1$ (perfect misclassification) to $+1$ (perfect selection), with $0$ indicating random guessing. ``rank'' refers to the average position of each method when ranked by MCC across all configurations. ``prop. $> 0.8$'' reports the proportion of cases across all 18 designs in which MCC exceeds 0.8. ``prop. 1st'' reports the proportion of cases in which a method ranks first based on MCC performance.
}
\end{table}

\begin{table}[th]
\caption{Summary Statistics for Matthews Correlation Coefficient (MCC) with $r=5$}%
\label{tab:MCC_r5}
\centering
\begin{tabular}{lccccccc}
\hline
& \textbf{median} & \textbf{IQR} & \textbf{min} & \textbf{max} & \textbf{prop. $>$ 0.8} & \textbf{rank} & \textbf{prop. 1st} \\
\hline
PCA-BMT   & 0.922 & 0.109 & 0.683 & 0.970 & 77.8\%  & 1.111 & 88.89\% \\
p-Lasso  & 0.705 & 0.201 & 0.360 & 0.854 & 16.7\%  & 2.667 & 0.00\% \\
i-Lasso & 0.779 & 0.235 & 0.587 & 0.943 & 50.0\%  & 2.222 & 11.11\% \\
\hline
\end{tabular}
\vspace{0.75em}

\noindent\parbox{\linewidth}{%
\small\textit{\textbf{Notes}:}  MCC ranges from $-1$ (perfect misclassification) to $+1$ (perfect selection), with $0$ indicating random guessing. ``rank'' refers to the average position of each method when ranked by MCC across all configurations. ``prop. $> 0.8$'' reports the proportion of cases across all 18 designs in which MCC exceeds 0.8. ``prop. 1st'' reports the proportion of cases in which a method ranks first based on MCC performance.
}
\end{table}

To gain further insight into method performance across several $T,n$
combinations, Figure~\ref{fig:MCC}
presents a visual summary of MCC-based
performance under $r=2,5$ and $\phi=0,1$. Detailed results are reported in Tables~\ref{tab:MCC}--\ref{tab:FPR}.

The graphical results reveal clear and systematic differences in performance across methods. For \( r = 2 \) and \( \phi = 1 \), PCA-BMT consistently achieves the highest MCC values across all nine configurations. Its performance improves steadily with larger \( T \) or smaller \( n \), rising from 0.901 in the smallest panel \((25,100)\) to 0.979 in the largest \((100,25)\). This monotonic gain is consistent with the idea that the method benefits from longer time series. At the same time, for any fixed \( T \), expanding the candidate variable set increases the complexity of selection, leading to reduced accuracy.

p-Lasso performs moderately well in this setting, with MCC values increasing from 0.723 to 0.851 as the time-series dimension increases. i-Lasso is more sensitive to \( T \): while it performs relatively poorly for short panels (e.g., 0.529 in \((25,25)\)), its accuracy increases markedly with time series length, closely tracking PCA-BMT in high-\( T \) configurations.

When \( \phi = 0 \), i.e., when the average loading is zero and pooling is less appropriate, differences across methods become more pronounced. PCA-BMT remains stable and continues to dominate, with MCC rising from 0.878 to 0.979 across configurations. p-Lasso, by contrast, is heavily affected: its MCC values remain below 0.71 in all cases and drop as low as 0.403 in \((25,100)\). i-Lasso is more resilient than p-Lasso, showing improved performance with increasing \( T \), and attaining MCC levels close to those of PCA-BMT in large panels.

As the signal dimension increases to \( r = 5 \), the estimation task becomes more challenging due to the higher complexity of the underlying structure. Nonetheless, PCA-BMT maintains strong performance, with MCC rising from 0.683 in a small-\( T \) setting with a large number of candidate factors (four times the number of time series observations \((25,100)\)) to approximately 0.970 in the largest configurations. The method continues to clearly outperform the other approaches.
While both p-Lasso and i-Lasso improve with panel size, they remain substantially less accurate than PCA-BMT. i-Lasso outperforms p-Lasso in nearly all cases and comes close to PCA-BMT when \( T \) is large. However, its gains are more modest compared to the lower-\( r \) case, highlighting the increased difficulty of signal recovery in higher dimensions.

The contrast sharpens when \( \phi = 0 \). PCA-BMT remains stable, with MCC ranging from 0.702 to 0.944, while p-Lasso performs poorly -- especially in short panels -- and i-Lasso shows improvement but consistently underperforms relative to PCA-BMT, unless $T$ is large.

Overall, PCA-BMT exhibits accurate and stable performance across all designs considered. Its robustness to variations in signal dimensionality, loading structure, and sample size suggests strong reliability in high-dimensional settings characterised by complex dependence structures.

\setcounter{figure}{0}
\begin{landscape}
\thispagestyle{empty} 
\vfill
\centering
\begin{minipage}{0.95\linewidth}
\begin{minipage}{0.53\linewidth}
\centering
\includegraphics[height=0.36\textheight]{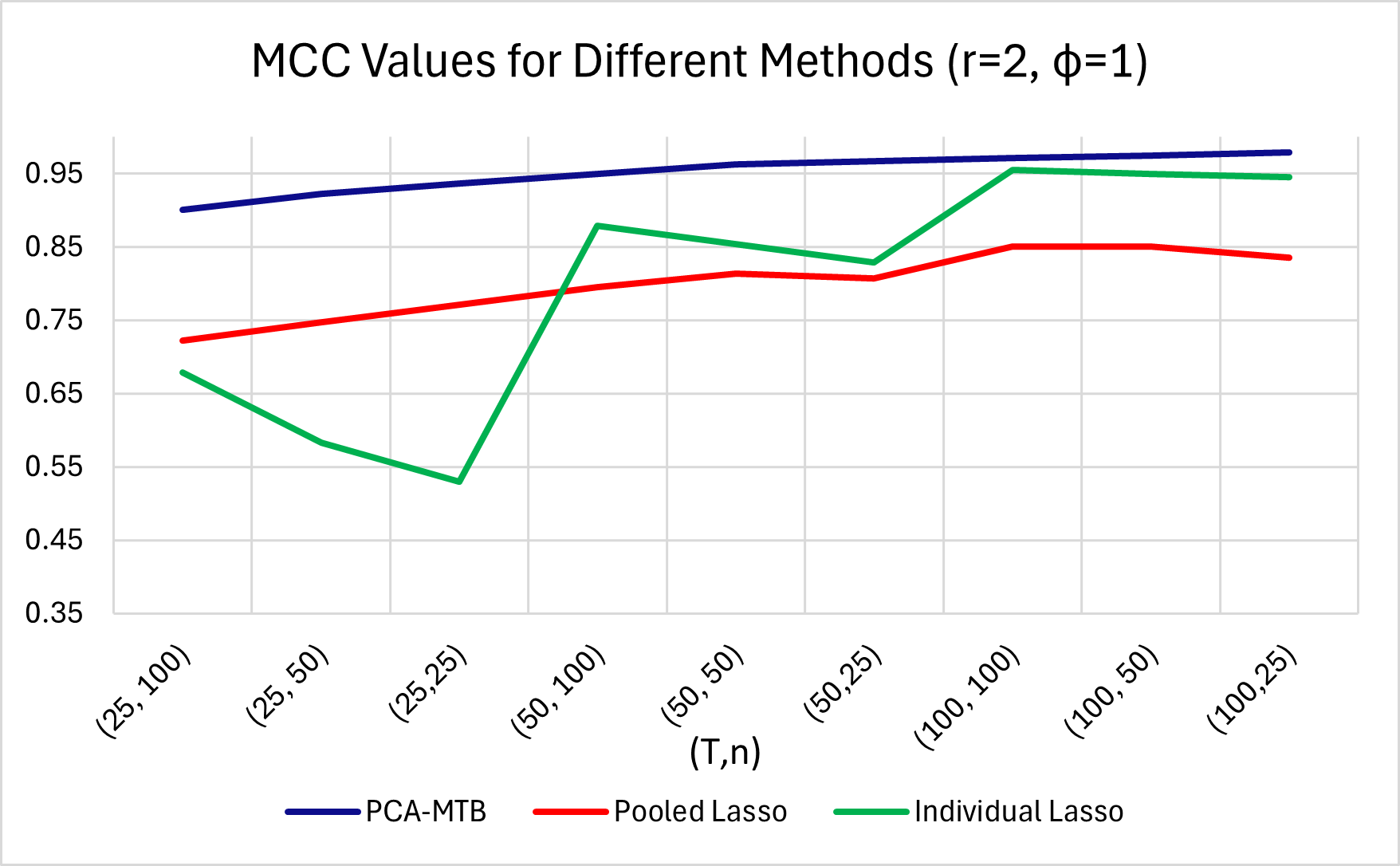}\\
\textbf{MCC}: $r=2$, $\phi=1$
\end{minipage}
\hfill
\begin{minipage}{0.53\linewidth}
\centering
\includegraphics[height=0.36\textheight]{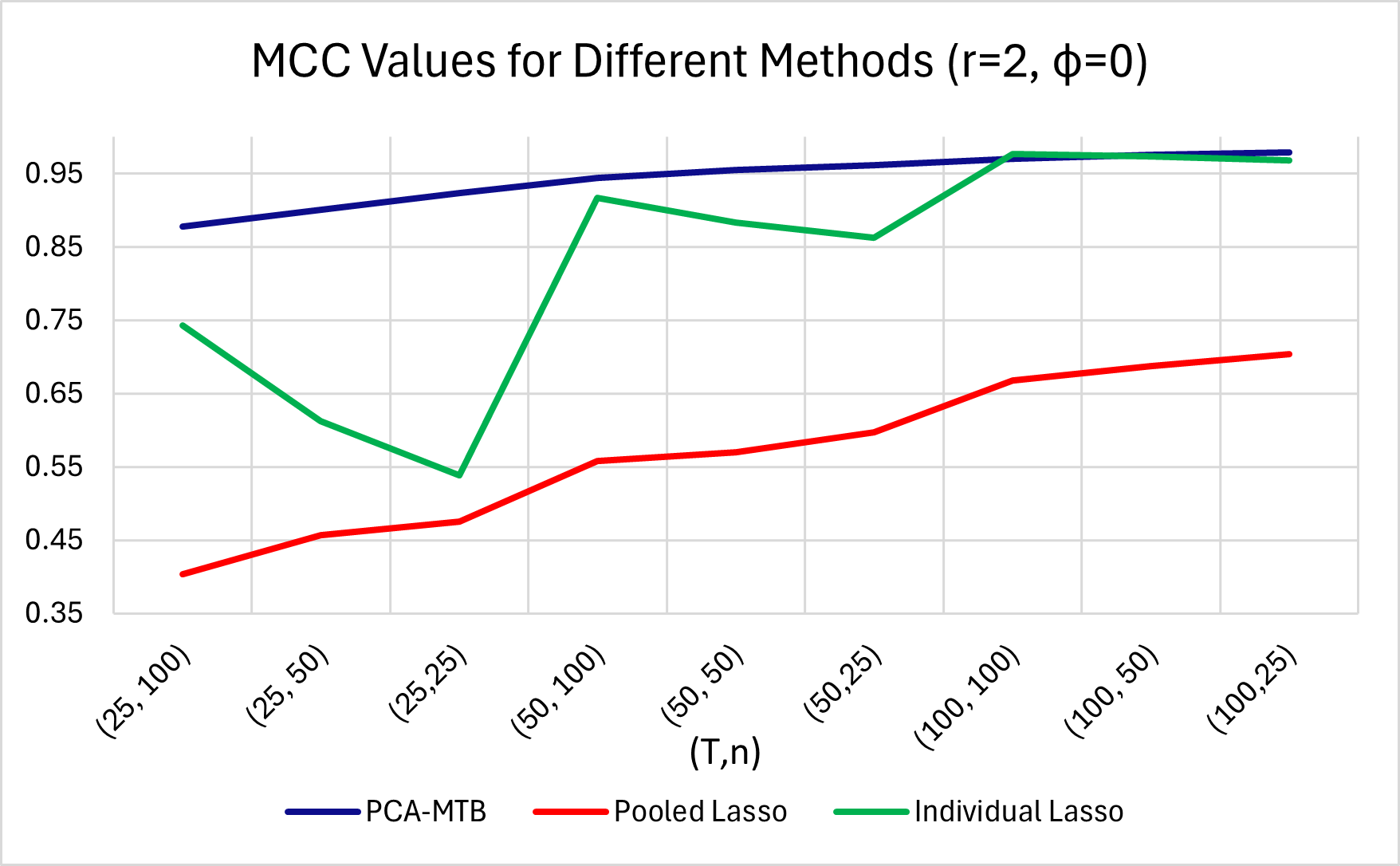}\\
\textbf{MCC}: $r=2$, $\phi=0$
\end{minipage}
\vspace{1em}
\begin{minipage}{0.53\linewidth}
\centering
\includegraphics[height=0.36\textheight]{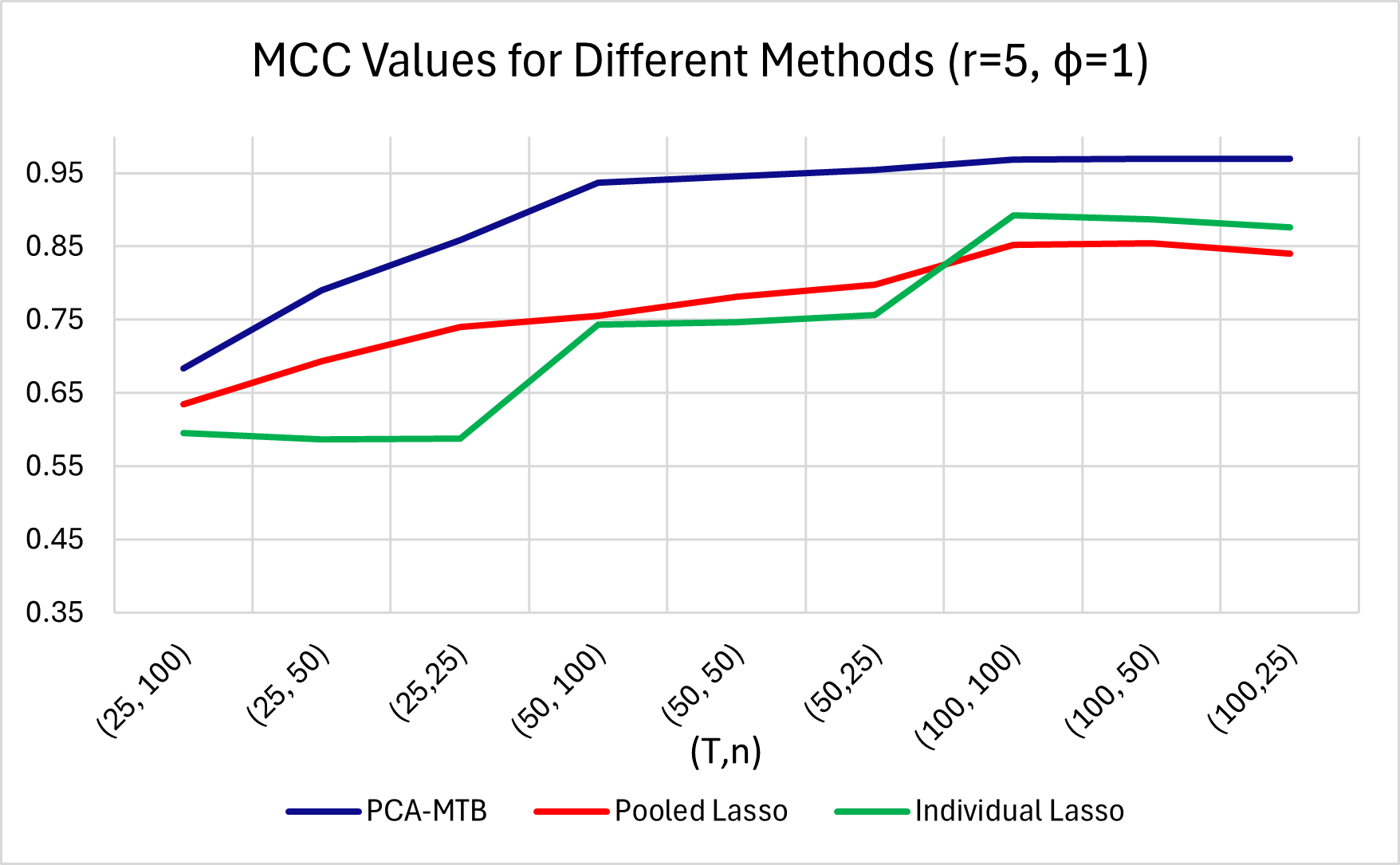}\\
\textbf{MCC}: $r=5$, $\phi=1$
\end{minipage}
\hfill
\begin{minipage}{0.53\linewidth}
\centering
\includegraphics[height=0.36\textheight]{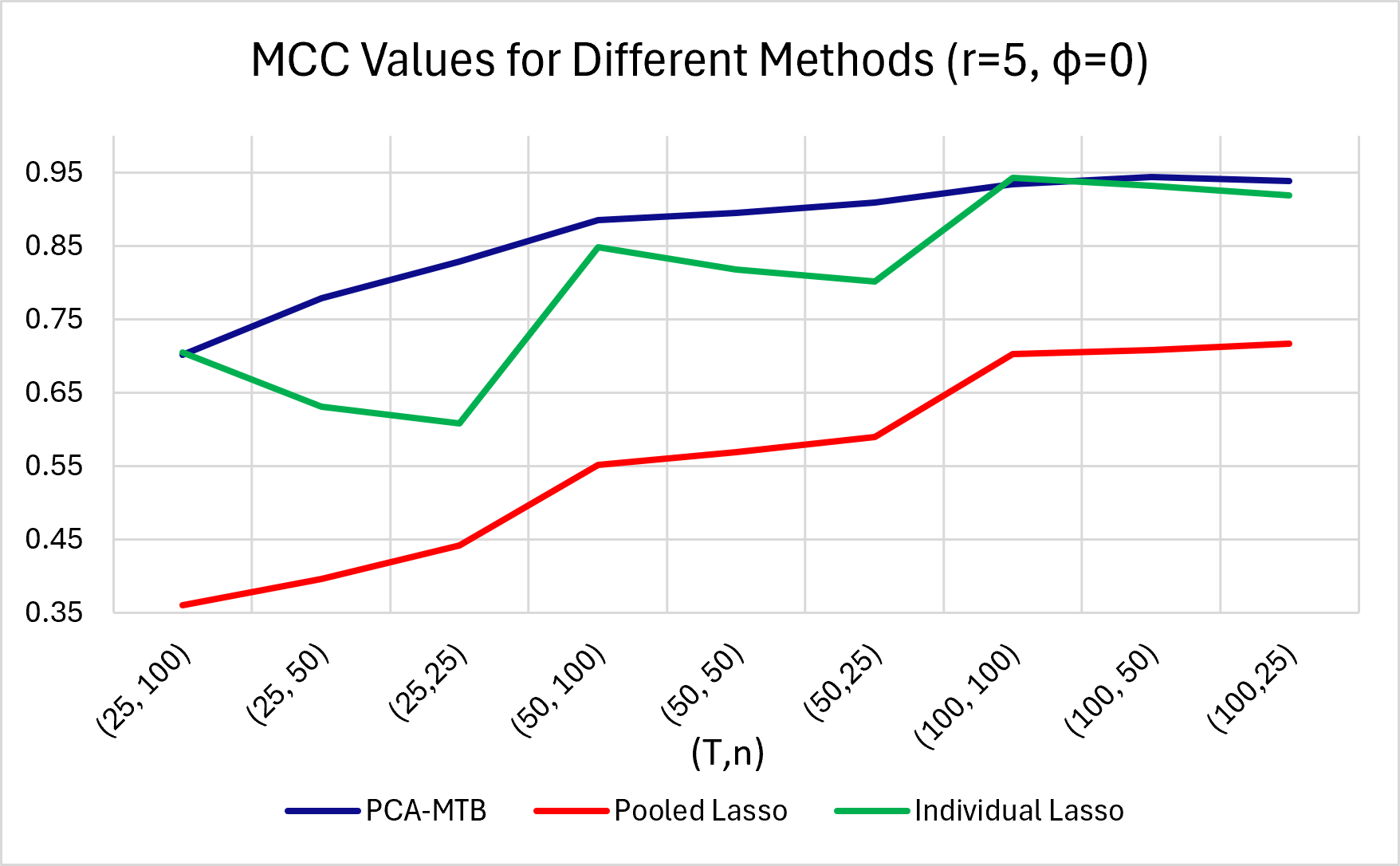}\\
\textbf{MCC}: $r=5$, $\phi=0$
\end{minipage}
\vspace{1em}
\vspace{1em}
\captionof{figure}{MCC Performance Evaluation over different (T,N) values.}
\label{fig:MCC}
\end{minipage}
\vfill
\end{landscape}

\clearpage
\begin{table}[htbp]
\centering
\begin{threeparttable}
\caption{MCC comparison across methods}
\label{tab:MCC}
\begin{tabular}{|l|l|ccc|ccc|ccc|}
\hline
\multicolumn{2}{|c|}{} & \multicolumn{3}{c|}{PCA-MTB} & \multicolumn{3}{c|}{p-Lasso} & \multicolumn{3}{c|}{i-Lasso} \\
\hline
Design & T/n & 25 & 50 & 100 & 25 & 50 & 100 & 25 & 50 & 100 \\
\hline
\multirow{3}{*}{$r=2$, $\phi=1$}
 & 25  & 0.937 & 0.923 & 0.901 & 0.772 & 0.747 & 0.723 & 0.530 & 0.584 & 0.679 \\
 & 50  & 0.967 & 0.964 & 0.951 & 0.808 & 0.814 & 0.796 & 0.830 & 0.854 & 0.879 \\
 & 100 & 0.979 & 0.976 & 0.972 & 0.836 & 0.851 & 0.852 & 0.946 & 0.950 & 0.956 \\
\hline
\multirow{3}{*}{$r=2$, $\phi=0$}
 & 25  & 0.924 & 0.901 & 0.878 & 0.476 & 0.457 & 0.403 & 0.539 & 0.613 & 0.744 \\
 & 50  & 0.962 & 0.956 & 0.945 & 0.598 & 0.570 & 0.558 & 0.863 & 0.884 & 0.918 \\
 & 100 & 0.980 & 0.976 & 0.971 & 0.704 & 0.688 & 0.669 & 0.968 & 0.974 & 0.978 \\
\hline
\multirow{3}{*}{$r=5$, $\phi=1$}
 & 25  & 0.859 & 0.791 & 0.684 & 0.741 & 0.693 & 0.635 & 0.588 & 0.587 & 0.596 \\
 & 50  & 0.955 & 0.946 & 0.937 & 0.798 & 0.782 & 0.756 & 0.757 & 0.747 & 0.744 \\
 & 100 & 0.970 & 0.970 & 0.969 & 0.841 & 0.854 & 0.853 & 0.876 & 0.887 & 0.893 \\
\hline
\multirow{3}{*}{$r=5$, $\phi=0$}
 & 25  & 0.829 & 0.779 & 0.702 & 0.442 & 0.396 & 0.360 & 0.608 & 0.631 & 0.705 \\
 & 50  & 0.910 & 0.896 & 0.886 & 0.590 & 0.569 & 0.552 & 0.802 & 0.819 & 0.849 \\
 & 100 & 0.940 & 0.944 & 0.934 & 0.718 & 0.709 & 0.703 & 0.920 & 0.933 & 0.944 \\
\hline
\end{tabular}
\vspace{0.5em}
\begin{minipage}{\linewidth}
\footnotesize
\textit{\textbf{Notes:}} Matthews Correlation Coefficient (MCC) is defined as
\[
\text{MCC} = \frac{\text{TP} \cdot \text{TN} - \text{FP} \cdot \text{FN}}{\sqrt{(\text{TP} + \text{FP})(\text{TP} + \text{FN})(\text{TN} + \text{FP})(\text{TN} + \text{FN})}},
\]
where TP, FP, TN, and FN denote the number of true positives, false positives, true negatives, and false negatives, respectively. MCC measures the correlation between predicted and true binary outcomes, reflecting how well the selected variables align with the actual ones. It ranges from $-1$ (perfect misclassification) to $+1$ (perfect selection), with 0 indicating random guessing.
\end{minipage}
\end{threeparttable}
\end{table}

\begin{table}[htbp]
\centering
\begin{threeparttable}
\caption{F1 comparison across methods}
\label{tab:F1_comparison}
\begin{tabular}{|l|l|ccc|ccc|ccc|}
\hline
\multicolumn{2}{|c|}{} & \multicolumn{3}{c|}{PCA-MTB} & \multicolumn{3}{c|}{p-Lasso} & \multicolumn{3}{c|}{i-Lasso} \\
\hline
Design & T/n & 25 & 50 & 100 & 25 & 50 & 100 & 25 & 50 & 100 \\
\hline
\multirow{3}{*}{$r=2$, $\phi=1$}
 & 25  & 0.936 & 0.918 & 0.893 & 0.763 & 0.723 & 0.687 & 0.502 & 0.535 & 0.637 \\
 & 50  & 0.966 & 0.961 & 0.946 & 0.801 & 0.798 & 0.772 & 0.826 & 0.844 & 0.868 \\
 & 100 & 0.979 & 0.974 & 0.970 & 0.832 & 0.839 & 0.836 & 0.945 & 0.947 & 0.952 \\
\hline
\multirow{3}{*}{$r=2$, $\phi=0$}
 & 25  & 0.924 & 0.896 & 0.872 & 0.476 & 0.450 & 0.390 & 0.512 & 0.569 & 0.713 \\
 & 50  & 0.961 & 0.953 & 0.940 & 0.596 & 0.559 & 0.544 & 0.861 & 0.877 & 0.911 \\
 & 100 & 0.979 & 0.974 & 0.968 & 0.699 & 0.676 & 0.653 & 0.968 & 0.972 & 0.976 \\
\hline
\multirow{3}{*}{$r=5$, $\phi=1$}
 & 25  & 0.887 & 0.824 & 0.735 & 0.767 & 0.689 & 0.601 & 0.626 & 0.576 & 0.557 \\
 & 50  & 0.960 & 0.950 & 0.939 & 0.820 & 0.783 & 0.740 & 0.782 & 0.747 & 0.727 \\
 & 100 & 0.974 & 0.971 & 0.969 & 0.859 & 0.857 & 0.847 & 0.891 & 0.890 & 0.891 \\
\hline
\multirow{3}{*}{$r=5$, $\phi=0$}
 & 25  & 0.862 & 0.810 & 0.746 & 0.461 & 0.397 & 0.348 & 0.644 & 0.623 & 0.683 \\
 & 50  & 0.919 & 0.901 & 0.891 & 0.611 & 0.567 & 0.536 & 0.823 & 0.821 & 0.843 \\
 & 100 & 0.946 & 0.946 & 0.934 & 0.739 & 0.709 & 0.692 & 0.930 & 0.935 & 0.943 \\
\hline
\end{tabular}
\vspace{0.5em}
\begin{minipage}{\linewidth}
\footnotesize
\textit{\textbf{Notes:}} F1 is the harmonic mean of precision and recall, defined as
\[
\text{F1} = \frac{2 \cdot \text{TP}}{2 \cdot \text{TP} + \text{FP} + \text{FN}},
\]
where TP, FP, and FN denote the number of true positives, false positives, and false negatives, respectively. It balances the trade-off between discovery (recall) and accuracy (precision), and is especially informative under sparsity and class imbalance.
\end{minipage}
\end{threeparttable}
\end{table}

\begin{table}[htbp]
\centering
\begin{threeparttable}
\caption{Model size comparison across methods}
\label{tab:modelsize_comparison}
\begin{tabular}{|l|l|ccc|ccc|ccc|}
\hline
\multicolumn{2}{|c|}{} & \multicolumn{3}{c|}{PCA-MTB} & \multicolumn{3}{c|}{p-Lasso} & \multicolumn{3}{c|}{i-Lasso} \\
\hline
Design & T/n & 25 & 50 & 100 & 25 & 50 & 100 & 25 & 50 & 100 \\
\hline
\multirow{3}{*}{$r=2$, $\phi=1$}
 & 25  & 2.322 & 2.427 & 2.600 & 3.578 & 3.933 & 4.359 & 6.361 & 5.904 & 4.655 \\
 & 50  & 2.172 & 2.204 & 2.286 & 3.220 & 3.276 & 3.505 & 2.969 & 2.864 & 2.722 \\
 & 100 & 2.109 & 2.133 & 2.157 & 3.001 & 2.956 & 2.991 & 2.282 & 2.273 & 2.246 \\
\hline
\multirow{3}{*}{$r=2$, $\phi=0$}
 & 25  & 2.250 & 2.407 & 2.548 & 1.735 & 1.745 & 1.633 & 6.204 & 5.437 & 3.889 \\
 & 50  & 2.179 & 2.215 & 2.293 & 2.024 & 1.918 & 1.886 & 2.750 & 2.666 & 2.467 \\
 & 100 & 2.106 & 2.132 & 2.162 & 2.200 & 2.086 & 1.947 & 2.163 & 2.142 & 2.123 \\
\hline
\multirow{3}{*}{$r=5$, $\phi=1$}
 & 25  & 5.722 & 5.628 & 5.516 & 8.297 & 9.963 & 12.406 & 11.209 & 12.727 & 13.514 \\
 & 50  & 5.412 & 5.511 & 5.617 & 7.385 & 8.031 & 8.892 & 7.948 & 8.623 & 9.051 \\
 & 100 & 5.095 & 5.118 & 5.142 & 6.780 & 6.819 & 6.981 & 6.308 & 6.319 & 6.317 \\
\hline
\multirow{3}{*}{$r=5$, $\phi=0$}
 & 25  & 5.223 & 5.281 & 5.239 & 3.177 & 3.146 & 3.131 & 10.741 & 11.347 & 9.968 \\
 & 50  & 4.952 & 4.978 & 5.021 & 4.063 & 3.864 & 3.918 & 7.272 & 7.324 & 7.018 \\
 & 100 & 4.837 & 4.852 & 4.818 & 4.771 & 4.477 & 4.370 & 5.818 & 5.748 & 5.657 \\
\hline
\end{tabular}
\vspace{0.5em}
\begin{minipage}{\linewidth}
\footnotesize
\textit{\textbf{Notes:}} Model size refers to the number of predictors selected by each method in a given replication. Values shown are the mean model sizes across Monte Carlo iterations.
\end{minipage}
\end{threeparttable}
\end{table}

\begin{table}[htbp]
\centering
\begin{threeparttable}
\caption{TDR comparison across methods}
\label{tab:TDR_comparison}
\begin{tabular}{|l|l|ccc|ccc|ccc|}
\hline
\multicolumn{2}{|c|}{} & \multicolumn{3}{c|}{PCA-MTB} & \multicolumn{3}{c|}{p-Lasso} & \multicolumn{3}{c|}{i-Lasso} \\
\hline
Design & T/n & 25 & 50 & 100 & 25 & 50 & 100 & 25 & 50 & 100 \\
\hline
\multirow{3}{*}{$r=2$, $\phi=1$}
 & 25  & 0.900 & 0.874 & 0.835 & 0.651 & 0.600 & 0.558 & 0.344 & 0.378 & 0.488 \\
 & 50  & 0.946 & 0.937 & 0.913 & 0.698 & 0.696 & 0.664 & 0.725 & 0.753 & 0.790 \\
 & 100 & 0.965 & 0.957 & 0.950 & 0.741 & 0.753 & 0.749 & 0.909 & 0.913 & 0.921 \\
\hline
\multirow{3}{*}{$r=2$, $\phi=0$}
 & 25  & 0.902 & 0.865 & 0.824 & 0.725 & 0.703 & 0.685 & 0.354 & 0.413 & 0.579 \\
 & 50  & 0.941 & 0.929 & 0.908 & 0.756 & 0.755 & 0.746 & 0.777 & 0.802 & 0.855 \\
 & 100 & 0.966 & 0.958 & 0.949 & 0.789 & 0.792 & 0.810 & 0.947 & 0.954 & 0.960 \\
\hline
\multirow{3}{*}{$r=5$, $\phi=1$}
 & 25  & 0.833 & 0.767 & 0.664 & 0.635 & 0.539 & 0.443 & 0.460 & 0.410 & 0.393 \\
 & 50  & 0.933 & 0.915 & 0.898 & 0.707 & 0.656 & 0.601 & 0.651 & 0.605 & 0.582 \\
 & 100 & 0.970 & 0.965 & 0.961 & 0.763 & 0.762 & 0.747 & 0.812 & 0.811 & 0.812 \\
\hline
\multirow{3}{*}{$r=5$, $\phi=0$}
 & 25  & 0.846 & 0.784 & 0.705 & 0.764 & 0.705 & 0.673 & 0.480 & 0.458 & 0.529 \\
 & 50  & 0.937 & 0.917 & 0.900 & 0.794 & 0.783 & 0.762 & 0.708 & 0.706 & 0.740 \\
 & 100 & 0.969 & 0.968 & 0.960 & 0.822 & 0.830 & 0.826 & 0.876 & 0.886 & 0.899 \\
\hline
\end{tabular}
\vspace{0.5em}
\begin{minipage}{\linewidth}
\footnotesize
\textit{\textbf{Notes:}} The True Discovery Rate (TDR), also known as precision, is defined as
\[
\text{TDR} = \frac{\text{TP}}{\text{TP} + \text{FP}},
\]
where TP and FP denote the number of true and false positives, respectively. TDR reflects the proportion of selected variables that are truly relevant. Higher values indicate greater accuracy in identifying true signals while avoiding false discoveries.
\end{minipage}
\end{threeparttable}
\end{table}

\begin{table}[htbp]
\centering
\begin{threeparttable}
\caption{FDR comparison across methods}
\label{tab:FDR_comparison}
\begin{tabular}{|l|l|ccc|ccc|ccc|}
\hline
\multicolumn{2}{|c|}{} & \multicolumn{3}{c|}{PCA-MTB} & \multicolumn{3}{c|}{p-Lasso} & \multicolumn{3}{c|}{i-Lasso} \\
\hline
Design & T/n & 25 & 50 & 100 & 25 & 50 & 100 & 25 & 50 & 100 \\
\hline
\multirow{3}{*}{$r=2$, $\phi=1$}
 & 25  & 0.100 & 0.126 & 0.165 & 0.349 & 0.400 & 0.442 & 0.656 & 0.622 & 0.512 \\
 & 50  & 0.054 & 0.063 & 0.087 & 0.302 & 0.304 & 0.336 & 0.275 & 0.247 & 0.210 \\
 & 100 & 0.035 & 0.043 & 0.050 & 0.259 & 0.247 & 0.251 & 0.091 & 0.087 & 0.079 \\
\hline
\multirow{3}{*}{$r=2$, $\phi=0$}
 & 25  & 0.098 & 0.135 & 0.176 & 0.275 & 0.297 & 0.315 & 0.646 & 0.587 & 0.421 \\
 & 50  & 0.059 & 0.071 & 0.092 & 0.244 & 0.245 & 0.254 & 0.223 & 0.198 & 0.145 \\
 & 100 & 0.034 & 0.042 & 0.051 & 0.211 & 0.208 & 0.190 & 0.053 & 0.046 & 0.040 \\
\hline
\multirow{3}{*}{$r=5$, $\phi=1$}
 & 25  & 0.167 & 0.233 & 0.336 & 0.365 & 0.461 & 0.557 & 0.540 & 0.590 & 0.607 \\
 & 50  & 0.067 & 0.085 & 0.102 & 0.293 & 0.344 & 0.399 & 0.349 & 0.395 & 0.418 \\
 & 100 & 0.030 & 0.035 & 0.039 & 0.237 & 0.238 & 0.253 & 0.188 & 0.189 & 0.188 \\
\hline
\multirow{3}{*}{$r=5$, $\phi=0$}
 & 25  & 0.154 & 0.216 & 0.295 & 0.236 & 0.295 & 0.327 & 0.520 & 0.542 & 0.471 \\
 & 50  & 0.063 & 0.083 & 0.100 & 0.206 & 0.217 & 0.238 & 0.292 & 0.294 & 0.260 \\
 & 100 & 0.031 & 0.032 & 0.040 & 0.178 & 0.170 & 0.174 & 0.124 & 0.114 & 0.101 \\
\hline
\end{tabular}
\vspace{0.5em}
\begin{minipage}{\linewidth}
\footnotesize
\textit{\textbf{Notes:}} The False Discovery Rate (FDR) is defined as
\[
\text{FDR} = \frac{\text{FP}}{\text{TP} + \text{FP}},
\]
where FP and TP denote the number of false and true positives, respectively. FDR reflects the proportion of selected predictors that are in fact irrelevant. Lower values indicate greater reliability in variable selection.
\end{minipage}
\end{threeparttable}
\end{table}

\begin{table}[htbp]
\centering
\begin{threeparttable}
\caption{TPR comparison across methods}
\label{tab:TPR_comparison}
\begin{tabular}{|l|l|ccc|ccc|ccc|}
\hline
\multicolumn{2}{|c|}{} & \multicolumn{3}{c|}{PCA-MTB} & \multicolumn{3}{c|}{p-Lasso} & \multicolumn{3}{c|}{i-Lasso} \\
\hline
Design & T/n & 25 & 50 & 100 & 25 & 50 & 100 & 25 & 50 & 100 \\
\hline
\multirow{3}{*}{$r=2$, $\phi=1$}
 & 25  & 0.996 & 0.995 & 0.995 & 1.000 & 1.000 & 1.000 & 1.000 & 1.000 & 1.000 \\
 & 50  & 1.000 & 1.000 & 1.000 & 1.000 & 1.000 & 1.000 & 1.000 & 1.000 & 1.000 \\
 & 100 & 1.000 & 1.000 & 1.000 & 1.000 & 1.000 & 1.000 & 1.000 & 1.000 & 1.000 \\
\hline
\multirow{3}{*}{$r=2$, $\phi=0$}
 & 25  & 0.970 & 0.966 & 0.962 & 0.520 & 0.495 & 0.433 & 1.000 & 1.000 & 1.000 \\
 & 50  & 0.995 & 0.994 & 0.995 & 0.652 & 0.606 & 0.589 & 1.000 & 1.000 & 1.000 \\
 & 100 & 1.000 & 1.000 & 1.000 & 0.763 & 0.726 & 0.688 & 1.000 & 1.000 & 1.000 \\
\hline
\multirow{3}{*}{$r=5$, $\phi=1$}
 & 25  & 0.943 & 0.867 & 0.755 & 1.000 & 1.000 & 0.999 & 1.000 & 1.000 & 1.000 \\
 & 50  & 0.996 & 0.993 & 0.990 & 1.000 & 1.000 & 1.000 & 1.000 & 1.000 & 1.000 \\
 & 100 & 0.982 & 0.981 & 0.982 & 1.000 & 1.000 & 1.000 & 1.000 & 1.000 & 1.000 \\
\hline
\multirow{3}{*}{$r=5$, $\phi=0$}
 & 25  & 0.878 & 0.827 & 0.749 & 0.437 & 0.382 & 0.338 & 1.000 & 1.000 & 1.000 \\
 & 50  & 0.917 & 0.901 & 0.891 & 0.599 & 0.548 & 0.524 & 1.000 & 1.000 & 1.000 \\
 & 100 & 0.932 & 0.934 & 0.919 & 0.743 & 0.697 & 0.677 & 1.000 & 1.000 & 1.000 \\
\hline
\end{tabular}
\vspace{0.5em}
\begin{minipage}{\linewidth}
\footnotesize
\textit{\textbf{Notes:}} The True Positive Rate (TPR), also known as recall or sensitivity, is defined as
\[
\text{TPR} = \frac{\text{TP}}{\text{TP} + \text{FN}},
\]
where TP and FN denote the number of true positives and false negatives, respectively. TPR measures the proportion of relevant variables correctly identified. A TPR of 1 indicates perfect recovery of all true signals.
\end{minipage}
\end{threeparttable}
\end{table}

\begin{table}[htbp]
\centering
\begin{threeparttable}
\caption{FPR comparison across methods}
\label{tab:FPR}
\begin{tabular}{|l|l|ccc|ccc|ccc|}
\hline
\multicolumn{2}{|c|}{} & \multicolumn{3}{c|}{PCA-MTB} & \multicolumn{3}{c|}{p-Lasso} & \multicolumn{3}{c|}{i-Lasso} \\
\hline
Design & T/n & 25 & 50 & 100 & 25 & 50 & 100 & 25 & 50 & 100 \\
\hline
\multirow{3}{*}{$r=2$, $\phi=1$}
 & 25  & 0.013 & 0.009 & 0.006 & 0.063 & 0.039 & 0.024 & 0.162 & 0.075 & 0.026 \\
 & 50  & 0.007 & 0.004 & 0.003 & 0.049 & 0.026 & 0.015 & 0.036 & 0.017 & 0.007 \\
 & 100 & 0.004 & 0.003 & 0.002 & 0.040 & 0.019 & 0.010 & 0.010 & 0.005 & 0.002 \\
\hline
\multirow{3}{*}{$r=2$, $\phi=0$}
 & 25  & 0.012 & 0.009 & 0.006 & 0.028 & 0.015 & 0.008 & 0.156 & 0.066 & 0.019 \\
 & 50  & 0.008 & 0.005 & 0.003 & 0.029 & 0.014 & 0.007 & 0.028 & 0.013 & 0.005 \\
 & 100 & 0.004 & 0.003 & 0.002 & 0.027 & 0.013 & 0.006 & 0.006 & 0.003 & 0.001 \\
\hline
\multirow{3}{*}{$r=5$, $\phi=1$}
 & 25  & 0.040 & 0.026 & 0.017 & 0.132 & 0.099 & 0.074 & 0.207 & 0.140 & 0.081 \\
 & 50  & 0.017 & 0.011 & 0.007 & 0.095 & 0.061 & 0.039 & 0.098 & 0.066 & 0.039 \\
 & 100 & 0.007 & 0.004 & 0.002 & 0.071 & 0.036 & 0.020 & 0.044 & 0.024 & 0.013 \\
\hline
\multirow{3}{*}{$r=5$, $\phi=0$}
 & 25  & 0.033 & 0.023 & 0.015 & 0.040 & 0.025 & 0.014 & 0.191 & 0.115 & 0.047 \\
 & 50  & 0.015 & 0.009 & 0.006 & 0.043 & 0.023 & 0.013 & 0.076 & 0.042 & 0.019 \\
 & 100 & 0.007 & 0.004 & 0.002 & 0.042 & 0.020 & 0.010 & 0.027 & 0.014 & 0.006 \\
\hline
\end{tabular}
\vspace{0.5em}
\begin{minipage}{\linewidth}
\footnotesize
\textit{\textbf{Notes:}} The False Positive Rate (FPR) is defined as
\[
\text{FPR} = \frac{\text{FP}}{\text{FP} + \text{TN}},
\]
where FP and TN denote the number of false positives and true negatives, respectively. FPR measures the proportion of irrelevant variables that were incorrectly selected. A lower FPR indicates fewer false discoveries.
\end{minipage}
\end{threeparttable}
\end{table}

\clearpage

\putbib[commodities]
\end{bibunit}

\end{document}